%% file: arxiv.tex
\documentclass[acmsmall]{acmart}
\AtBeginDocument{%
  }

\setcopyright{cc}
\setcctype{by}
\acmJournal{PACMMOD}
\acmYear{2026} \acmVolume{4} \acmNumber{2 (PODS)} \acmArticle{110}
\acmMonth{5} \acmDOI{10.1145/3801906}
\received{December 2025}
\received[accepted]{February 2026}
\usepackage{tikz}
\newcommand*\circled[2][draw,inner sep=1pt]{%
  \tikz[baseline=(char.base)]{
    \node[circle,#1] (char) {$\displaystyle #2$};}}
\newcommand{\boundary}[1]{\ensuremath{{\tiny \circled{#1}}}}
\usepackage[skip=4pt,font=small,labelformat=simple]{subcaption}
\usepackage{adjustbox}
\usepackage{float}
\usepackage[most]{tcolorbox}
\usepackage{graphicx}
\usepackage{multirow}
\usepackage{mathrsfs}
\usepackage{amsmath}
\usepackage{mathtools}
\usepackage{amsthm}
\usepackage{algpseudocodex}
\usepackage{flushend}
\usepackage{hyperref}
\usepackage{caption}
\usepackage{cleveref}
\usepackage{wrapfig}
\usepackage{extpfeil}
\usepackage{tabularx}
\usepackage[inline]{enumitem}
\setlist[enumerate]{leftmargin=20pt,itemindent=0pt,itemsep=0pt,parsep=0pt,topsep=2pt,partopsep=0pt}

\usepackage{amsthm}
\usepackage{setspace,etoolbox}

\newcommand{\ThmPreSkip}{3pt}   % space above env
\newcommand{\ThmPostSkip}{3pt}  % space below env
\newcommand{\ThmHeadSep}{0.3em} % space after head
\newcommand{\ThmStretch}{0.96}  % line spacing inside env (0.94~0.98 typical)

\newtheoremstyle{tightplain}% name
  {\ThmPreSkip}{\ThmPostSkip}% above/below
  {\itshape}% body font
  {}% indent
  {\bfseries}% head font
  {.}% punctuation after head
  {\ThmHeadSep}% after head
  {\thmname{#1}\thmnumber{ #2}\thmnote{ [#3]}}

\newtheoremstyle{tightdef}% name
  {\ThmPreSkip}{\ThmPostSkip}%
  {}% upright body
  {}%
  {\bfseries}%
  {.}%
  {\ThmHeadSep}%
  {\thmname{#1}\thmnumber{ #2}\thmnote{ [#3]}}

\newtheoremstyle{tightremark}% name
  {\ThmPreSkip}{\ThmPostSkip}%
  {}% upright body
  {}%
  {\itshape}% italic head (remark-style)
  {.}%
  {\ThmHeadSep}%
  {\thmname{#1}\thmnumber{ #2}\thmnote{ [#3]}}

\theoremstyle{tightplain}
\newtheorem{theorem}{Theorem}[section]
\newtheorem{lemma}[theorem]{Lemma}

\newtheorem{claim}[theorem]{Claim}

\theoremstyle{tightdef}
\newtheorem{definition}[theorem]{Definition}

\theoremstyle{tightremark}

\newtheorem{fact}[theorem]{Fact}

\AtBeginEnvironment{theorem}{\begingroup\setstretch{\ThmStretch}}
\AtEndEnvironment{theorem}{\endgroup}
\AtBeginEnvironment{lemma}{\begingroup\setstretch{\ThmStretch}}
\AtEndEnvironment{lemma}{\endgroup}
\AtBeginEnvironment{proposition}{\begingroup\setstretch{\ThmStretch}}
\AtEndEnvironment{proposition}{\endgroup}
\AtBeginEnvironment{corollary}{\begingroup\setstretch{\ThmStretch}}
\AtEndEnvironment{corollary}{\endgroup}
\AtBeginEnvironment{claim}{\begingroup\setstretch{\ThmStretch}}
\AtEndEnvironment{claim}{\endgroup}

\AtBeginEnvironment{claim}{\begingroup\setstretch{\ThmStretch}}
\AtEndEnvironment{claim}{\endgroup}

\AtBeginEnvironment{definition}{\begingroup\setstretch{\ThmStretch}}
\AtEndEnvironment{definition}{\endgroup}
\AtBeginEnvironment{assumption}{\begingroup\setstretch{\ThmStretch}}
\AtEndEnvironment{assumption}{\endgroup}
\AtBeginEnvironment{example}{\begingroup\setstretch{\ThmStretch}}
\AtEndEnvironment{example}{\endgroup}

\AtBeginEnvironment{remark}{\begingroup\setstretch{\ThmStretch}}
\AtEndEnvironment{remark}{\endgroup}
\AtBeginEnvironment{note}{\begingroup\setstretch{\ThmStretch}}
\AtEndEnvironment{note}{\endgroup}
\AtBeginEnvironment{observation}{\begingroup\setstretch{\ThmStretch}}
\AtEndEnvironment{observation}{\endgroup}

\AtBeginEnvironment{proof}{\begingroup\setstretch{\ThmStretch}}
\AtEndEnvironment{proof}{\endgroup}

\usepackage{stfloats}
\usepackage[textsize=tiny]{todonotes}

\usepackage[utf8]{inputenc} % allow utf-8 input
\usepackage[T1]{fontenc}    % use 8-bit T1 fonts
\usepackage{hyperref}       % hyperlinks
\usepackage{url}            % simple URL typesetting
\usepackage{booktabs}       % professional-quality tables
\usepackage{amsfonts}       % blackboard math symbols
\usepackage{nicefrac}       % compact symbols for 1/2, etc.
\usepackage{microtype}      % microtypography
\usepackage[dvipsnames]{xcolor}
\definecolor{myblue}{RGB}{0,6,209}
\usepackage{xspace}
\usepackage{stfloats}
\usepackage[ruled,vlined,linesnumbered]{algorithm2e}
\usepackage{wrapfig}
\SetKwInput{KwInput}{Input}                % Set the Input
\SetKwInput{KwOutput}{Output}              % set the Output
\crefname{theorem}{lemma}{lemmas}
\Crefname{theorem}{Lemma}{Lemmas}
\DeclareMathOperator{\bigo}{\mathcal{O}}

\DeclareMathOperator{\algo}{\mathcal{A}}

\DeclareMathOperator{\pr}{\mathbb{P}}
\DeclareMathOperator{\E}{\mathbb{E}}

\NewDocumentCommand{\haoyu}{mg}{%
  \IfNoValueTF{#2}%
    {\textcolor{teal}{#1}}%
    {\textcolor{teal}{$\blacktriangleright$}#1
     \textcolor{teal}{$\triangleright$ #2$\blacktriangleleft$}}%
}

\newcommand{\mypara}[1]{%
    \vspace{2pt}
    \par\noindent\textbf{#1}
    \noindent}

\newcommand{\query}[1]{{\textsf{#1}}}

\newtheorem*{lemmanrestate}{Lemma}

\AtBeginDocument{%
  }
\newcommand{\algoa}{\textbf{RoundWalks}}
\newcommand{\algor}{\textbf{DistWalks}}
\newcommand{\algorr}{\textbf{RoundEst}}

\newcommand{\newalgor}{\textbf{DistWalks}\textsf{(Adaptive-Threshold)}}
\newcommand{\newalgorr}{\textbf{RoundEst}\textsf{(Graph-Prebuilt)}}

\begin{document}

\title[Near-Optimality for Single-Source Personalized PageRank]
{Near-Optimality for Single-Source Personalized PageRank}
\titlenote{Authors are listed in alphabetical order. Siqiang Luo is the corresponding author.}

\begin{abstract}
The Single Source Personalized PageRank (SSPPR) query is a fundamental problem in graph OLAP, and has been extensively studied across various data-driven applications. 
For a source node $s$ in a graph $G$, SSPPR measures the PPR score $\pi(s,t)$ for every node $t\in G$, where $\pi(s,t)$ is the probability that an $\alpha$-decay random walk starting from $s$ terminates at $t$. Despite decades of research, a significant gap remains between the known. Under two standard accuracy guarantees—\emph{absolute error} (SSPPR-A), requiring $|\hat{\pi}(s,t)-\pi(s,t)| \le \epsilon$ for all $t$, and \emph{relative error} (SSPPR-R), requiring $|\hat{\pi}(s,t)-\pi(s,t)| \le c\cdot\pi(s,t)$ for all $t$ with $\pi(s,t)\ge\delta$—the best-known upper bounds by~\cite{mc,fora_j,speedppr} collectively yield {\scriptsize$\bigo\!\left(\min({\log(1/\epsilon)}/{\epsilon^{2}},{\sqrt{m\log n}}/{\epsilon},m\log({1}/{\epsilon}))\right)$} for SSPPR-A and {\scriptsize$\bigo\!\left(\min\left({\log(1/n)}/{\delta},\sqrt{m\log(n/\delta)},m\log(\frac{\log(n)}{m\delta})\right)\right)$} for SSPPR-R ($n,m$ denote the number of nodes and edges). Meanwhile, only trivial lower bounds are known—{\scriptsize$\Omega(\min(n,1/\epsilon))$}~\cite{wei2024approximating} and {\scriptsize$\Omega(\min(n,1/\delta))$}—leaving a substantial theoretical gap yet to be closed. This work narrows or even closes this gap. On the upper bound side, we 
tighten the computational complexities of the SSPPR-A and SSPPR-R queries to {\scriptsize$\bigo(1/\epsilon^2)$} and {\scriptsize$\bigo\!\left(\min\left({\log(1/\delta)}/{\delta},m+n\log(n)\log(\frac{\log(n)}{m\delta})\right)\right)$}, respectively. These results improve upon the best-known bounds by factors of {\scriptsize$\log(1/\epsilon)$} and {\scriptsize$\log(\frac{\log(n)}{m\delta})$}. On the lower bound side, under the standard arc-centric model, we establish significantly stronger results: {\scriptsize$\Omega(\min(m,1/\epsilon^2))$} for SSPPR-A (improving by $m/n$ or $1/\epsilon$) and {\scriptsize$\Omega(\min(m,{\log(1/\delta)}/{\delta}))$} for SSPPR-R (improving by $m/n$ or $\log(1/\delta)$). These findings substantially strengthen the theoretical foundations of PPR computation: For SSPPR-R, our new upper and lower bounds coincide for graphs with $m\in \Omega (n\log^2(n))$ and any threshold \(\delta, 1/\delta \in \bigo(\text{poly}(n)) \), indicating %that our algorithm \algor{} 
that we have achieved \textit{theoretical optimality} under most graph regimes; the SSPPR-A query attains partial optimality under rough absolute error requirements (i.e., larger $\epsilon$), where the upper bound simplifies to {\scriptsize$\bigo(1/\epsilon^2)$}. %via \algoa{}, 
This also matches our new lower bound result. To the best of our knowledge, this is the first work to establish an \textit{optimal algorithmic result} for SSPPR queries. Further, our proposed techniques and results are generalizable to other relevant topics. Our lower-bound framework naturally extends to the \textit{Single-Target Personalized PageRank} (STPPR) query, improving its lower bound from {\scriptsize$\Omega(\min(n,1/\delta))$} to {\scriptsize$\Omega(\min(m,\tfrac{n}{\delta}\cdot\log n))$} under standard arc-centric model. This new result also matches the upper bound established in~\cite{rbs}, revealing the \textit{optimality} and underscoring its theoretical versatility.

\end{abstract}

\author{Xinpeng Jiang}
%\authornotemark[2]
\orcid{0009-0009-5840-9881}
\affiliation{%
  \department{College of Computing and Data Science}
  \institution{Nanyang Technological University}
  \country{Singapore}
}
\email{xinpeng006@e.ntu}

\author{Haoyu Liu}
%\authornotemark[2]
\orcid{0000-0002-0839-5460}
\affiliation{%
  \department{College of Computing and Data Science}
  \institution{Nanyang Technological University}
  \country{Singapore}
}
\email{haoyu.liu@ntu.edu.sg}

\author{Siqiang Luo}
%\authornotemark[2]
\orcid{0000-0001-8197-0903}
\affiliation{%
  \department{College of Computing and Data Science}
  \institution{Nanyang Technological University}
  \country{Singapore}
}
\email{siqiang.luo@ntu.edu.sg}

\author{Xiaokui Xiao}
%\authornotemark[2]
\orcid{0000-0003-0914-4580}
\affiliation{%
  \department{School of Computing}
  \institution{National University of Singapore}
  \country{Singapore}
}
\email{xkxiao@nus.edu.sg}




\maketitle

\input{text/intro}
%\input{text/relatedwork}
\input{text/prelims}
\input{text/improved_upper_bounds}
\input{text/linear_alg}

\input{text/ssppr_r_lower_dense}

\input{text/ssppr_a_lower}

\input{text/stppr_lower}

\input{text/conclusion}
\input{text/ack}
\bibliographystyle{ACM-Reference-Format}
\bibliography{main, add}
\input{text/appendix}

\end{document}

%% file: text/intro.tex
\section{Introduction}
\label{sec:intro}
The \emph{Single-Source Personalized PageRank} (SSPPR) query, first proposed in~\cite{haveliwala2002topic}, has become a fundamental graph operation, attracting sustained interest from both academia and industry.  
Given a graph $G$, the SSPPR vector $\pi(s,\cdot)$ with respect to a source node $s$—PPR values of all nodes from $s$, is defined as the probability of an $\alpha$—decay Random Walk ($\alpha$-RW), with probability $\alpha$ to stop at the current node and probability $1-\alpha$ to randomly transition to neighbors, starting from $s$ and terminating at any node $t$ as $\pi(s,t)$. This definition naturally captures the importance of any node $t$ relative to $s$, giving rise to widely used applications such as web search~\cite{brin1998anatomy,wu2008top}, recommendation systems~\cite{drossman2021review,bhpp,nguyen2015evaluation,liu2024bird}, graph sparsification~\cite{vattani2011preserving,yang2019homogeneous}, and graph neural networks~\cite{appnp,pprgo,agp,scara,scara_j}. %{\color{red} Please mostly mention them as data-driven applications or put it under data management context. Make sure the refs of SIGMOD, PODS, VLDB, ICDE, ICDT are sufficient.}

Exact computation of all SSPPR vectors is expensive, requiring matrix inversion with $\Omega(n^2\log n)$ time~\cite{tveit2003complexity}. %, which is prohibitive for large real-world graphs.  
Consequently, considerable effort has focused on approximation algorithms for SSPPR query to balance efficiency with error guarantees. As such, a series of approximation algorithms have emerged~\cite{mc,fora,wu2021unifying,luo2019baton,yoon2018tpa,avrachenkov2007monte,lrh2022efficient,lrh2023efficient,liu2016powerwalk}, focusing on improved computational upper bounds based on two error settings: \emph{absolute error} (SSPPR-A), which requires $|\hat{\pi}(s,t)-\pi(s,t)| \le \epsilon$ for all $t$, and \emph{relative error} (SSPPR-R), which requires $|\hat{\pi}(s,t)-\pi(s,t)| \le c\cdot\pi(s,t)$ for all $t$ with $\pi(s,t)\ge\delta$. Among the literature, MC~\cite{mc}, FORA~\cite{fora_j} and SpeedPPR~\cite{speedppr} together give the best known upper bound results of {\scriptsize$\bigo\!\left(\min({\log(1/\epsilon)}/{\epsilon^{2}},{\sqrt{m\log n}}/{\epsilon},m\log({1}/{\epsilon}))\right)$} for SSPPR-A and {\scriptsize$\bigo\!\left(\min({\log(1/\delta)}/{\delta},\sqrt{m\log(n/\delta)},m\log({\log(n)}/{(\delta m)}))\right)$} for SSPPR-R query. Though large amount of subsequent works~\cite{liao2022efficient, liao2023efficient, liu2024bird, wang2019parallelizing, hou2023personalized,hou2021massively} try to further improve the computing overhead, such theoretical upper bound has been struck for nearly half decade. Despite the vitality in upper bounds, the lower bounds of SSPPR query remains far less developed. The best known lower bounds are $\Omega\!\left(\min(n,1/\varepsilon)\right)$ for SSPPR-A with $\epsilon$ absolute error~\cite{wei2024approximating} and $\Omega\!\left(\min(n,1/\delta)\right)$ for SSPPR-R with a constant relative error, both largely derived from worst-case output-size arguments or from the simpler \emph{single-pair} PPR setting, due to the difficulties of constructing hard instances that force \emph{any} algorithm to incur substantial computation on all nodes to meet the error requirement. This leaves a substantial theoretical gap yet to be closed.

In this work, we largely close the gap between upper-bounds and lower-bounds in SSPPR queries, as summarized in Table~\ref{tab:pprs}. On the \textbf{upper-bound} side, we propose two new algorithms, \algoa{} and \algor{}, which together tighten the computational complexities of the SSPPR-A and SSPPR-R queries. % to {\scriptsize$\bigo(1/\epsilon^2)$} and {\scriptsize$\bigo\!\left(\min({\log(1/\delta)}/{\delta},m+n\log(n)\log^3(1/(n\delta))\log n)\right)$}. 
These results improve upon the best-known bounds by factors of {\scriptsize$\log(1/\epsilon)$} and {\scriptsize$\log({\log(n)}/{(m\delta)})$}. On the \textbf{lower-bound} side, %under the standard arc-centric model, 
we establish significantly stronger results: {\scriptsize$\Omega(\min(m,1/\epsilon^2))$} for SSPPR-A (improving by $m/n$ or $1/\epsilon$) and {\scriptsize$\Omega(\min(m,{\log(1/\delta)}/{\delta}))$} for SSPPR-R (improving by $m/n$ or $\log(1/\delta)$). %These findings substantially strengthen the understanding of SSPPR computation. Most interestingly, 
We highlight that, for SSPPR-R, our upper and lower bounds coincide for most graphs with $m\in \Omega (n\log^2(n))$ and any threshold \(\delta, 1/\delta \in \bigo(\text{poly}(n)) \), indicating that our algorithm \algor{} achieves \textit{theoretical optimality} under most graph regimes. Besides, the SSPPR-A query attains partial optimality under rough absolute error requirements (i.e., larger $\epsilon$), where the upper bound simplifies to {\scriptsize$\bigo(1/\epsilon^2)$} via \algoa{}, matching our new lower bound result. To the best of our knowledge, this is the first work to establish an \textit{optimal algorithmic result} for SSPPR queries, marking a significant advancement in the foundational study of graph analysis. 

Our lower-bound results are based on constructing a non-trivial hard instance, and this technique has impact beyond the SSPPR topic. %Notably, our proposed techniques and results are broadly generalizable. 
For example, we successfully extend our lower-bound framework to establish a much stronger result for \textit{Single-Target Personalized PageRank} (STPPR) query, improving its lower bound from {\scriptsize$\Omega(\min(n,1/\delta))$} to {\scriptsize$\Omega(\min(m,\frac{n}{\delta}\cdot\log(n)))$} under standard arc-centric model. This new result also matches the computational upper bound established in~\cite{rbs}, thereby revealing the \textit{optimality} of our analysis and underscoring its theoretical robustness and versatility. %In the remainder of this section, we formally state the problem setting, review best-known results, and outline our new results and contributions.

\subsection{Problem Formulation}
% \haoyu{highlight $\delta$ setting is $\geq \frac{1}{n}$ is mostly used, and we achieve matching results. And less interesting parts, we improve previous gaps.}
%Among the comprehensive approximation measurements including the $\ell_1$, $\ell_2$, absolute, relative and degree-normalized absolute error, we focus on the research problem of SSPPR query with relative and additive error guarantee unifiedly. These two settings are widely utilized and studied in mainstreaming PPR-related researches~\cite{fora, wu2021unifying, mc, rbs} which requires the estimator to be accurate within chosen error requirement. We give their formal definition below:
We focus on two error settings: relative and additive error guarantees, which are widely adopted in mainstream PPR studies~\cite{fora_j, wu2021unifying, wang2019parallelizing, topppr, liu2024bird, zhu2024personalized, hou2021massively, liao2022efficient, liao2023efficient, zhou2025one, setpush, wang2024revisiting}.

\begin{definition}
\label{df_ssppr_r}
SSPPR with Relative Error Guarantee (SSPPR-R). Given a graph $G$, a decay factor $\alpha$, a source node $s$, a relative error parameter $ c \in (0,0.5) $, an overall failure probability parameter $ p_f $ and choosing a targeting PPR threshold $\delta \in (0,1)$, SSPPR-R query requires an estimated SSPPR vector $\hat{\pi}(s,\cdot)$ such that with at least $1-p_f$ probability, $|\hat{\pi}(s,t)-\pi(s,t)| \leq c \cdot \pi(s,t)$
holds for all node $t$ with its PPR value $\pi(s,t)\geq \delta$. Note that the threshold $\delta$ can be chosen arbitrary close to 0 or 1, which correspond to the shallow or large targeting PPR values. Following conventions in prior work~\cite{pprscaling,bippr,wang2024revisiting,liu2024bird}, parameters $\alpha$, $c$, and $p_f$ are treated as constants.
\end{definition}

\begin{definition}
\label{df_ssppr_a}
SSPPR with Additive Error Guarantee (SSPPR-A). Given a graph $G$, a decay factor $\alpha$, a source node $s$, an additive error parameter $\varepsilon \in (0,1)$ and an overall failure probability parameter $ p_f $, SSPPR-A query requires an estimated SSPPR vector $\hat{\pi}(s,\cdot)$ such that with at least $1-p_f$ probability, $|\hat{\pi}(s,t)-\pi(s,t)| \leq \varepsilon $
holds for all node $t$. 
\end{definition}

\begin{definition}
\label{df_stppr}
Single Target Personalized PageRank Query (STPPR). Given a graph $G$, a decay factor $\alpha$, a target node $s$, a relative error parameter $ c \in (0,0.5) $, an overall failure probability parameter $ p_f $ and choosing a targeting PPR threshold $\delta \in (0,1)$, STPPR query requires an estimated STPPR vector $\hat{\pi}(s,\cdot)$ such that with at least $1-p_f$ probability, $|\hat{\pi}(s,t)-\pi(s,t)| \leq c \cdot \pi(s,t)$
holds for all node $s$ with its PPR value $\pi(s,t)\geq \delta$.
%The definition of STPPR follows SSPPR-R by changing source node $s$ to target node $t$ and estimates STPPR vector $\hat{\pi}(\cdot,t)$.
%Given a graph $G$, a decay factor $\alpha$, a target node $t$, a relative error parameter $ c \in (0,0.5) $, an overall failure probability parameter $ p_f $ and choosing a targeting PPR threshold $\delta \in (0,1)$, STPPR query requires an estimated STPPR vector $\hat{\pi}(\cdot, t)$ such that with at least $1-p_f$ probability, $|\hat{\pi}(s,t)-\pi(s,t)| \leq c \cdot \pi(s,t)$
%holds for all node $s$ with its PPR value $\pi(s,t)\geq \delta$. Note that the threshold $\delta$ can be chosen arbitrary close to 0 or 1, which correspond to the shallow or large targeting PPR values. 
\end{definition}

\subsection{Existing Results for SSPPR-A, SSPPR-R and STPPR}
%We summarize the best-known upper and lower bounds for SSPPR-A, SSPPR-R and STPPR queries. % overview inTable~\ref{tab:pprs}. %An overview of these results is presented in Table~\ref{tab:pprs}.
%\haoyu{re-write with : first list the three base: MC, FP and backward push. Then summarize the sota methods and bounds.}

\mypara{Upper Bounds.} %We briefly summarize the best-known computational upper bounds for answering SSPPR-A, SSPPR-R and STPPR queries here. 
Prior arts primarily use three foundational algorithmic paradigms, namely, Monte Carlo (MC), Forward Push (FP), and Backward Push (BP), and have been further refined through various state-of-the-art (SOTA) techniques. The MC method provides a sample-based approximation to the SSPPR vector using $\alpha$-decay random walks, and achieves high-probability relative error guarantees with expected time complexity {\small$\bigo\left(\log(n)/\delta\right)$} for SSPPR-R, or {\small$\bigo\left(\log(1/\varepsilon)/\varepsilon^2\right)$} for SSPPR-A. The FP method is a local-push algorithm that simulates random walks deterministically via graph traversal. To guarantee an additive error of $\varepsilon$, it requires setting a residue threshold $r_{\max} = \varepsilon/m$, resulting in a worst-case time complexity of $\bigo(m/\varepsilon)$. The BP method is similarly structured but performs the push operation in reverse, tailored to STPPR. With \textit{global} BP, $\bigo(m\log(1/\delta))$ is required for solving STPPR. BP can also be applied naively to SSPPR-A for all target nodes with $\bigo(m/\varepsilon)$ complexity. To overcome the inefficiencies of basic PPR algorithms, hybrid methods that combine Monte Carlo (MC) sampling with local-push strategies have been developed. 
\citet{fora} improves upon MC by first using FP and then refining the result via targeted MC. Balancing these two phases, it achieves an expected complexity of $\bigo\!\bigl(\sqrt{(m\log n)/\delta}\bigr)$ for SSPPR-R query. ~\citet{wu2021unifying} restructures FP to enforce a more regular push schedule, reducing total residual mass more efficiently and yielding $\bigo\!\bigl(m\log(1/\varepsilon)\bigr)$ for SSPPR-A and comparable efficiency for SSPPR-R when combined with MC.  
More recently,~\citet{wei2024approximating} combine MC with an adaptive variant of Backward Push (BP). Their approach first constructs a candidate set via a coarse MC estimation, and then refines it using selective backward pushes together with additional random walks, achieving a sublinear expected complexity of $\bigo\!\bigl(\sqrt{m\log n}/\varepsilon\bigr)$ for SSPPR-A. For the single-target case,~\citet{rbs} introduce randomness into the deterministic BP procedure, yielding an unbiased estimator with bounded variance while avoiding the propagation of negligible probability mass. Their method answers STPPR in $\bigo\!\bigl(\frac{n}{\delta}\cdot\log n\bigr)$ time, with an additional preprocessing cost of $O(m)$ to sort in-neighbors. Subsequent work~\cite{bertram2025estimating} further improves the complexity to $\bigo\!\left(\min\!\left(m\log\!\bigl(\tfrac{1}{\delta}\bigr),\,\tfrac{n}{\delta}\cdot\log n\right)\right)$ by combining classical Power Iteration with MC techniques.%These methods illustrate how combining sampling with localized deterministic updates leads to near-optimal performance.% and greatly improves the scalability of PPR computation on large-scale graphs.

\input{floats/tab_best_known}

\mypara{Lower Bounds.} Unlike the extensive body of work on improving \emph{upper} bounds for SSPPR computation, rigorous \emph{lower} bounds remain sparse and often heuristic, offering limited insight into the fundamental complexity of the problem. The earliest bound relevant to Personalized PageRank (PPR) is due to ~\citet{lofgren2014fast}, which gives a lower bound of $\Omega(\sqrt{1/\delta})$ for \emph{single-pair} PPR with relative error. However, this is dominated by the trivial output-size bound $\Omega\!\bigl(\min(n,1/\delta)\bigr)$ obtained by noting that there can be up to $\min(n,1/\delta)$ targets $t$ with $\pi(s,t)\!\ge\!\delta$. To the best of our knowledge, no other nontrivial \emph{direct} lower bounds are known for SSPPR-R. For the absolute-error variant (SSPPR-A), \citet{wei2024approximating} proves $\Omega\!\bigl(\min(n,1/\varepsilon)\bigr)$ by the same output-size reasoning: any algorithm must return nonzero estimates for all $t$ with $\pi(s,t)\!\ge\!\varepsilon$, of which there may be $\Omega\!\bigl(\min(n,1/\varepsilon)\bigr)$. While conceptually simple, such arguments understate the difficulty of producing accurate estimates for \emph{all} nodes within prescribed error margins. Likewise, STPPR query, \citet{rbs} shows lower bounds of $\Omega(n)$ in the worst case and $\Omega\!\bigl(\min(n,1/\delta)\bigr)$ on average, which still leaves a substantial gap relative to known upper bounds unlike other problems~\cite{lu2022optimal}.

\subsection{Our Results and Discussion}
In this section, we present our main results summarized in Table~\ref{tab:pprs}. All our arguments are developed under the standard \emph{arc-centric model~\cite{goldreich1998property,goldreich1997property,goldreich2008computational,goldreich2017introduction, goldreich1999combinatorial}}, which will be formally defined in Section~\ref{sec:notation}.

\begin{theorem}[Improved upper bound of SSPPR-R]\label{thm:ssppr-r-upper} 
We present \algor{} (Algorithm~\ref{algo:dist_walks}), which solves the SSPPR-R query within $\bigo\left(m+n\log(n)\log(\frac{\log(n)}{m\delta})\right)$ queries in expectation.
\end{theorem}

\noindent
\algor{} fundamentally overcomes the limitations of prior state-of-the-art methods that rely on deterministic push operations for initialization. Instead, it introduces a novel decomposition of PPR by accumulating contributions from already discovered effective paths, thereby avoiding redundant random walks over previously visited nodes. This refined design achieves, for the first time, linear dependency on the graph size ($m$) in relatively dense graphs. This significantly tightens the upper-bound widely. % without introducing an annoying $\frac{1}{\delta}$ factor (in fact, we make it a $\log{\frac{1}{\delta}}$ factor). 

\begin{theorem}[Improved upper bound of SSPPR]\label{thm:ssppr-a-upper} 
We present \algoa{} (Algorithm~\ref{algo:discover}), which solves the SSPPR-A (SSPPR-R) query within $\bigo(1/\epsilon^2)$ ($\bigo(\log(1/\delta)/\delta)$) queries in expectation.
\end{theorem}

\noindent 
The proposed Algorithm~\ref{algo:discover} employs a two-phase estimator, i.e., \emph{discovery} and \emph{estimation}, to improve upon naive Monte Carlo (MC) methods. This design achieves $\log(1/\epsilon)$ improvement of SSPPR-A. %It first identifies high-probability candidate nodes through random walks, then concentrates precise estimation on this subset. 

\begin{theorem}[Improved lower bound of SSPPR-R]\label{thm:ssppr-r}
Choose any decay factor \(\alpha\in (0,1)\), failure probability \(p\in (0,1)\), error parameter \(c\in(0,\frac{1}{2}]\), and any functions \(\delta(n)\in(0,1), m(n)\in\Omega(n)\cap \bigo(n^2)\). 
Consider any oracle (possibly randomized) algorithm within the arc-centric model that, for an arbitrary graph \(G\) with \(\bigo(m)\) edges and \(\bigo(n)\) nodes and any source node \(s\), with probability at least \(1-p\) the oracle algorithm gives an estimation \(\hat\pi(s,\cdot)\) for SSPPR vector \(\pi(s,\cdot)\), satisfying $|\hat{\pi}(s,v)-\pi(s,v)| \leq c \cdot \pi(s,v)$ for each node \(v\) at \(\pi(s,v) \geq \delta\).
Then there exists a graph \(\mathrm{U}\) and a source node \(s\) such that:
\begin{enumerate}[label=(\roman*).]
    \item \(\mathrm{U}\) contains  \(\bigo(m)\) edges and \(\bigo(n)\) nodes;
    \item the algorithm requires \(\Omega\left(\min(m, \frac{\log(1/\delta)}{\delta})\right)\) queries to succeed in expectation.
\end{enumerate}
\end{theorem}

\noindent
Theorem~\ref{thm:ssppr-r} improves the best-known lower bound from $\Omega\!\bigl(\min(n,1/\delta)\bigr)$ to $\Omega\!\bigl(\min(m,\log(1/\delta)/\delta)\bigr)$. 
This improvement is substantial, adding a multiplicative factor of up to $m/n$, which can reach $\Theta(n)$ on dense graphs such as cliques when high approximation precision is required. Moreover, the entire proof framework is readily adaptable. By inverting all edge directions in the SSPPR-R construction, we derive  Theorem~\ref{thm:stppr} for the STPPR query:

%It further reveals a key insight: when the graph is at least sparse-free ($m=\Omega(n\log n)$) and the error threshold satisfies $\delta=1/n$—the standard setting in PPR query computation~\cite{fora,speedppr,liu2024bird,bippr}—our lower bound matches the known computational upper bounds of MC and FORA. This indicates that MC and FORA are in fact \emph{optimal} algorithms for computing SSPPR-R queries in practical regimes.  

\begin{theorem}[Improved Lower Bound of STPPR Query]\label{thm:stppr}
Choose any decay factor \(\alpha\in (0,1)\), failure probability \(p\in (0,1)\), error parameter \(c\in(0,\frac{1}{2}]\), and any functions \(\delta(n)\in(0,1), m(n)\in\Omega(n)\cap \bigo(n^{2})\). 
Consider any oracle (possibly randomized) algorithm within the arc-centric model that, for an arbitrary graph \(G\) with \(\bigo(m)\) edges and \(\bigo(n)\) nodes and any target node \(t\), with probability at least \(1-p\) the oracle algorithm outputs an estimation \(\hat\pi(\cdot,t)\) for STPPR vector \(\pi(\cdot,t)\), satisfying $|\hat{\pi}(v,t)-\pi(v,t)| \leq c \cdot \pi(v,t)$ for each node \(v\) at \(\pi(v,t) \geq \delta\).
Then there exists a graph \(\mathrm{U}\) and a target node \(t\) such that:
\begin{enumerate}[label=(\roman*).]
\item \(\mathrm{U}\) contains  \(\bigo(m)\) edges and \(\bigo(n)\) nodes;
\item the algorithm requires \(\Omega\left(\min(m, \frac{n}{\delta}\cdot\log(n))\right)\) queries to succeed in the worst case.
\end{enumerate}
\end{theorem}

\noindent
Benefiting from our flexible proof framework, Theorem~\ref{thm:stppr} significantly strengthens the best-known worst-case lower bound from $\Omega\!\bigl(\min(n,1/\delta)\bigr)$ to $\Omega\!\bigl(\min(m, \frac{n}{\delta}\cdot \log(n))\bigr)$.  
Crucially, this lower bound matches the known computational upper bound in the prevalent setting $\delta = \bigo(1/n)$, which is standard in PPR query computation~\cite{fora,speedppr,liu2024bird,bippr}.  
This implies that the computational complexity of STPPR query is essentially \textit{optimal}, unveiling the maturity in PPR algorithm development.  
%Moreover, the improved lower bound in Theorem~\ref{thm:ssppr-a} further narrows the remaining gap for SSPPR-A. 
%We also observe that with refined hard-instance constructions, our framework can yield even tighter bounds, as developed in the subsequent theorem.

\begin{theorem}[Improved Lower Bound of SSPPR-A Query]\label{thm:ssppr-a-v1}
Choose any decay factor \(\alpha\in (0,1)\), failure probability \(p\in (0,1)\),  and any functions \(\varepsilon(n)\in(0,1), m(n)\in\Omega(n)\cap \bigo(n^{2})\).
Consider any oracle (possibly randomized) algorithm within the arc-centric model that, for an arbitrary graph \(G\) with \(\bigo(m)\) edges and \(\bigo(n)\) nodes and any source node \(s\), with probability at least \(1-p\) the oracle algorithm gives an estimation \(\hat\pi(s,\cdot)\) for SSPPR vector \(\pi(s,\cdot)\), satisfying $|\hat{\pi}(s,v)-\pi(s,v)| \leq \varepsilon$ for each node \(v\).
Then there exists a graph \(\mathrm{U}\) and a source node \(s\) such that:
\begin{enumerate}[label=(\roman*).]
\item \(\mathrm{U}\) contains  \(\bigo(m)\) edges and \(\bigo(n)\) nodes;
\item the algorithm requires \(\Omega\left(\min(m, \frac{1}{\varepsilon^2})\right)\) queries to succeed in expectation.
\end{enumerate}
\end{theorem}

\noindent
Theorem~\ref{thm:ssppr-a-v1} strengthens the best-known lower bound from $\Omega\!\left(\min(n,1/\epsilon)\right)$ to $\Omega\!\left(\min(m,1/\epsilon^{2})\right)$.  
This improvement is substantial and, strikingly, the $\tfrac{1}{\epsilon^{2}}$ term partially matches our improved computational upper bound for the SSPPR-A query.  
This alignment indicates that our new algorithm (Algorithm~\ref{algo:discover}) is essentially \emph{optimal} for SSPPR-A when the error requirement is relatively coarse.

%\mypara{Remark.}
%As summarized in Table~\ref{tab:pprs}, we substantially narrow the gap between known upper and lower bounds for PPR queries. In particular, we obtain almost \emph{matching} upper and lower bounds: For SSPPR-R, \haoyu{add res diss}. 
%These results strongly suggest the \emph{optimality} of our bounds and provide valuable guidance for future research on the complexity of PPR queries. 

%% file: floats/tab_best_known.tex
\vspace{\intextsep}
\begin{table}[!hb]
\Large
\vspace{-\intextsep}
\caption{Best-Known results and our Improvements of SSPPR-A, SSPPR-R and STPPR query. We highlight \textcolor{myblue}{our new results in blue} and bold the \textbf{Improvements} when error parameters $\epsilon, \delta$ are sufficiently small. For simplicity, we denote $F_A(m,n,\epsilon):=\min(\frac{1}{\epsilon}{\sqrt{m \log(n)}}\text{~\cite{wei2024approximating}}, m \log\frac{1}{\epsilon}\text{~\cite{wu2021unifying}})$.} %For ease of simplicity, we here assume the number of edges, $m\in \bigo(n^{2-\beta})$ for any constant $\beta\in(0,1)$ in our SSPPR-A Lower Bounds.}
%\vspace{-\intextsep}
\centering
\label{tab:pprs}
\renewcommand{\arraystretch}{2.2}
\begin{minipage}{1.\linewidth}
\resizebox{1.\textwidth}{!}{
\begin{tabular}{|c|c|c|c|}
\toprule
\textbf{Query} & \textbf{Best-Known Lower Bounds} & \textbf{New Lower Bounds}\footnote{A concurrent work~\cite{bertram2025estimating} establishes lower bounds of $\Omega(\min(m,1/\delta))$ for SSPPR (via a direct application of their Single-Node PPR result) and $\Omega(\min(m,n/\delta))$ for STPPR. Developed independently using different techniques, our results further surpass these bounds by introducing additional $\log(1/\delta)$ and $\log(n)$ factors, respectively.} & \textbf{Improvements} \\
\midrule
SSPPR-R & $\Omega\left(\min(n, \frac{1}{\delta})\right)$%~[\textcolor{myblue}{TRIVIAL}] 
& {\color{myblue}$\Omega\left(\min(m, \frac{\log(\frac{1}{\delta})}{\delta})\right)$} & $\log(\frac{1}{\delta})$ or ${\mathbf{\frac{m}{n}}}$
\\ \hline
SSPPR-A & $\Omega\left(\min(n, \frac{1}{\epsilon})\right)$~\cite{wei2024approximating} & {\color{myblue}$\Omega\left(\min(m, \frac{1}{\epsilon^2})\right)$} & $\frac{1}{\epsilon}$ or ${\mathbf{\frac{m}{n}}}$ 
\\ \hline
STPPR & $\Omega\left(\min(n, \frac{1}{\delta})\right)$~\cite{rbs} & {\color{myblue}$\Omega\left(\min(m, \frac{n}{\delta}\cdot\log(n))\right)$} & $n\log(n)$ or $\mathbf{\frac{m}{n}}$% , \textcolor{red}{[Optimal]}\\
\\ 
\midrule

\textbf{Query} & \textbf{Best-Known Upper Bounds} & \textbf{New Upper Bounds} & \textbf{Improvements} \\
\midrule

\multirow{2}{*}{SSPPR-R} 
& $\bigo \Big(\min\big(\frac{\log(n)}{\delta}\text{\cite{mc}}, \sqrt{\frac{m \log(n)}{\delta}}\text{\cite{fora}},$ 
& {\color{myblue}$\bigo\Big(\min\big(\tfrac{\log(\frac{1}{\delta})}{\delta},$} & $\frac{\log(n)}{\log(\frac{1}{\delta})}$ \\

& $\quad\;\;  m \log(\frac{\log(n)}{\delta m})\text{\cite{wu2021unifying}}\big)\Big)$ 
&{\color{myblue} $\quad\;\; m+n\log(n)\log(\frac{\log(n)}{\delta m})\big)\Big)$} 
& or $\approx \log(\frac{\log(n)}{\delta m})$
\\
\hline

SSPPR-A &
{$\bigo\Big(\min(\frac{\log(1/\epsilon)}{\epsilon^2}\text{\cite{mc}}, F_A(m,n,\epsilon)\big)\Big)$ } &
{$\bigo\Big(\min\big({\color{myblue}\frac{1}{\epsilon^2}}, F_A(m,n,\epsilon)\big)\Big)$} & $\log(\frac{1}{\epsilon})$
\\
\hline

STPPR & \multicolumn{1}{c|}{$\bigo\left(\min(m\log(\frac{1}{\delta}), \frac{n}{\delta}\cdot\log(n))\right)$\cite{bertram2025estimating}} & --- & --- \\ 
\midrule

\multicolumn{1}{|c|}{\textbf{Query}}
&\multicolumn{3}{c|}{\textbf{Optimality}} \\ 
\midrule
\multicolumn{1}{|c|}{SSPPR-R} & {\color{myblue}\algor{} [Ours]} & \multicolumn{2}{c|}{any small \(\delta, 1/\delta \in \bigo(\text{poly}(n)) \) and $m\in \Omega (n\log^2(n))$} \\ \hline
\multicolumn{1}{|c|}{SSPPR-A} & {\color{myblue}\algoa{} [Ours]} & \multicolumn{2}{c|}{any small $\epsilon,{1}/{\epsilon}\in \bigo(\sqrt{m})$}   \\ \hline
\multicolumn{1}{|c|}{STPPR} & RBS~\cite{rbs} & \multicolumn{2}{c|}{any small $\delta, 1/\delta \in \bigo(\frac{m}{n\log(n)})$} \\

\bottomrule
\end{tabular}}
\end{minipage}
\vspace{-\intextsep}
\end{table}

%% file: text/prelims.tex
\section{Notations and Tools}
\label{sec:notation}

\mypara{Notations.} Given an underlying directed graph $G = (V_G, E_G)$, let $n = |V_G|$ and $m = |E_G|$ denote the number of nodes and edges, respectively. Each edge $e = (u, v) \in E_G$ indicates that node $u$ is a parent of node $v$, and correspondingly, $v$ is a child of $u$. For a node set $V$, we denote its elements either as $v \in V$ or using index notation $V[i]$ to refer to the node at index $i$. The out-degree and in-degree of a node $u$ are denoted by $d_{\text{out}}(u)$ and $d_{\text{in}}(u)$, respectively. In particular, we denote by $\deg_V(v)$ the number of edges between node $v$ and all potential nodes in the set $V$. A permutation of size $n$ is defined as a bijection $\mathrm{p}^n\in \mathrm{S}^n: \{1, 2, \ldots, n\} \to \{1, 2, \ldots, n\}$; equivalently, $\mathrm{p^n}$ is an element of the symmetric group $\mathrm{S}^n$. For any index $0 \leq i < n$, we use $\mathrm{p}[i]$ to represent the permuted index. We summarize all frequently used notations of this paper in Table~\ref{tab:notations} for reference.%We summarize all frequently used notations of this paper in Table~\ref{tab:notations} for reference.% We consider random walks on $G$ where, at each step, the walk proceeds to a uniformly random child of the current node. Additionally, we study $\alpha$-decay random walks, where at each step, the walk continues with probability $\alpha \in (0,1)$. Unless otherwise specified, we assume $\alpha$ is a constant. Finally, we allow the graph to contain parallel edges (i.e., multiple edges between the same pair of nodes), which are formally defined in Definition~\ref{def:multigraph} below. %Additional frequently used notations are summarized in Table~\ref{tab:fre_notations}. \haoyu{add table.}

\begin{definition}[multigraph and edge multiplicity~\cite{martens2022complexity}]\label{def:multigraph}
A \emph{multigraph} is a graph that may contain parallel edges (multiple edges between the same pair of nodes); a \emph{simple graph} contains no such parallel edges. The \emph{edge multiplicity} of a multigraph $\mathrm{U}$ is defined as $$\mathrm{MUL}(G) \;=\; \max_{u,v \in V} \bigl|\{\,e \in E_{\mathrm{U}} \mid e = (u,v)\}\bigr|,$$ representing the maximum number of parallel edges between any node pair.
\end{definition}

\begin{comment}
\begin{definition}[Conditional uniformity]
Let \(S\) be a non-empty finite set. Whenever this abuse of notation causes no ambiguity, we treat \(S\) both as a set and as the uniform distribution over its elements; that is, an element \(s \in S\) is selected with probability \(1/|S|\).

Given a predicate \(\mathrm{cond}: S \to \{\mathrm{True}, \mathrm{False}\}\), the \emph{conditioning} of the uniform distribution on the event \(\mathrm{cond}(s) = \mathrm{True}\) yields the uniform distribution over the subset
\[
S_{\mathrm{cond}} \;=\; \{\, s \in S \mid \mathrm{cond}(s) = \mathrm{True} \}.
\]
\end{definition}
\end{comment}

%\begin{definition}[Maximum posterior estimation]\label{df:MAP}
%Let \(\Theta\) be a prior distribution over a finite set of possible values, and let \(\mathbf{x}\) denote the observed data. Then, there exists a maximum a posteriori (MAP) estimate \(\hat{\theta}_{\mathrm{MAP}} \in \Theta\) such that the success probability of any algorithm that predicts \(\theta \in \Theta\) based on \(\mathbf{x}\) is at most \(\Pr[\hat{\theta}_{\mathrm{MAP}} \mid \mathbf{x}]\). Furthermore, replacing \(\mathbf{x}\) with any function of \(\mathbf{x}\) (i.e., using only partial or derived information) cannot increase the prediction success probability.
%\end{definition}

\begin{definition}[Arc–Centric Graph Query Model~\cite{goldreich1998property, goldreich1997property}]\label{def:arc}
We focus on the query complexity under the standard RAM model. To formalize graph interactions, we allow algorithms to access the underlying graph $G$ under the standard \emph{arc-centric graph-access model}~\cite{goldreich1998property, goldreich1997property}, where the oracle supports the following queries in unit time:
\begin{enumerate}[label=(\roman*).]
  \item \query{JUMP}(): returns a uniformly random node $v \in V$ and marks $v$ as \textit{covered};
  \item \query{INDEG}$(v)$/\query{OUTDEG}$(v)$: for any \textit{covered} node $v \in V$, returns its in/out-degree $d_{in}(v)$/$d_{out}(v)$;
  \item \query{ADJ$_{in}$}$(v,i)$/\query{ADJ$_{out}$}$(v,i)$: for any \textit{covered} node $v \in V$ and integer $1 \le i \le d_{in}(v)/d_{out}(v)$, returns the $i^{th}$ in/out-neighbor of node $v$.
\end{enumerate}
% This model has been widely used in prior work~\cite{lofgren2014fast, bippr, wang2024revisiting} to establish lower bound complexities. 
%In our analysis, we strengthen the adjacency query \query{ADJ$_{out}$}$(v,i)$ (analogously for \query{ADJ}$_{in}$) to not only return the out-neighbor node $w$ of $v$ but also provide the label of the in-edge $(v,w)$ for node $w$ denoted as $k$, such that edge $(v,w)$ is the $i$-th out edge for $v$ and $k$-th in-edge for $w$. Note that this stronger model, we adopt by default, reveals more information than the original, ensuring that any lower bounds we derive remain valid.% for the standard arc-centric model. %For simplicity and clarity, we adopt this enhanced model by default later. %In the context of graph algorithms, \textit{query complexity} refers to the number of queries the algorithm invokes to the graph oracle. Hence, query complexity serves as a lower bound to computational complexity.
\end{definition}

\begin{definition}[Graph with edge and node indices]\label{def:arc-centric}
Following the notation in~\cite{goldreich1998property}, we refer to a graph $G = (V, E)$ with edge and node indices with the following labels:
\begin{enumerate*}[label=(\roman*).]
    \item Node labels. Each node \(v \in V\) is assigned a unique label from the set \(\{1, \ldots, |V|\}\);
    \item Edge labels. For each node \(v \in V\), its in-edges and out-edges are labeled with integers of \(\{1, 2, \ldots, d_{\mathrm{in}}(v)\}\) and \(\{1, 2, \ldots, d_{\mathrm{out}}(v)\}\).
\end{enumerate*}
\end{definition}

% \mypara{Remark.} Due to space limitations, we defer all proof details to the online full version~\cite{jiang2025near}. In particular, detailed proofs for ~\cref{lemma:all_rounds,lemma:data_process,lemma:discover,lemma:GDB_prop,lemma:legal_instance,lemma:newdecom,lemma:one_round,lemma:PPR_R_const,lemma:SSPPR-A1,lemma:statistical_hardness,lemma:targetProb} are provided there.
\begin{comment}
\begin{definition}[State]
\label{df:state}
When interacting with the underlying graph via queries, we define the \emph{current state} after performing $t$ queries as  
\[
\mathrm{s_t} = \bigl( (Q_1, R_1), (Q_2, R_2), \ldots, (Q_t, R_t) \bigr),
\]
where $Q_i$ denotes the $i$-th query executed under the allowed query model, and $R_i$ is the corresponding response, such as an edge returned by $Q_i$. Meanwhile, we denote the distribution of state as $\mathrm{S_t}$.
\end{definition}
\end{comment}

%% file: text/improved_upper_bounds.tex
\section{Improved Upper Bounds for SSPPR Queries}

{In this section, we present our advances on the upper bounds for approximating SSPPR queries. Existing upper bounds do \emph{not} imply any optimality, even when combined with our new lower bounds in Table~\ref{tab:pprs}. Specifically, the best-known SSPPR-R upper bound {\footnotesize $F_R=\bigo\!\left(\min\!\left(\tfrac{\log(n)}{\delta},\sqrt{\tfrac{m\log(n)}{\delta}},\, m\log\!\tfrac{\log(n)}{\delta m}\right)\right)$} reveals two fundamental gaps from the lower bound {\footnotesize $f_R=\Omega\!\left(\min\!\left(\tfrac{\log(1/\delta)}{\delta},\,m\right)\right)$}: for moderate thresholds (e.g.,\ $\delta=1/\log(n)$) a persistent $\log(n)$ gap remains, and for very small $\delta$, an extra $\log(\tfrac{\log(n)}{\delta m})$ factor prevents matching the $\Omega(m)$ bound. To close these two gaps, we introduce two complementary algorithms: \algoa{}, which refines the classical MC estimator to eliminate the unnecessary global $\log(n)$ dependence and achieves {\footnotesize $\bigo(\log(1/\delta)/\delta)$} for moderate-$\delta$ regimes; and \algor{}, which leverages new insights into state-of-the-art SSPPR techniques and a novel path-based PPR decomposition to achieve {\footnotesize $\bigo\!\left(m+n\log(n)\log(\frac{\log(n)}{\delta m})\right)$}, matching the $\Omega(m)$ lower bound on sufficiently dense graphs (e.g.,\ $m\in\Omega(n\log^2 n)$). Together, these results yield an improved overall upper bound {\footnotesize $\bigo\!\left(\min\!\left(\tfrac{\log(1/\delta)}{\delta},\, m+n\log(n)\log(\frac{\log(n)}{\delta m})\right)\right)$}, matching the lower bounds for most practical settings where $1/\delta\in\mathrm{poly}(n)$ and $m$ is moderately large. We first introduce the main idea of \algoa{}.}

%In this section, we present our advances on the algorithmic upper bounds for SSPPR approximation.  
%We begin by revisiting the limitations of the classical MC method and introduce our refined estimator \algoa{}. The idea is simple, yet it surprisingly resolves these issues and achieves the upper bounds stated in Theorem~\ref{thm:ssppr-a-upper}.  
%Building on new insights into state-of-the-art SSPPR techniques and a novel path-based decomposition of PPR, we further develop \algor{}—the first algorithm to obtain \emph{linear} query complexity in the number of edges \(m\), as formalized in Theorem~\ref{thm:ssppr-r}.  
%This establishes a significant breakthrough in the upper-bound landscape for SSPPR queries, resolving long-standing barriers that have persisted for years.

\mypara{Improve the MC Estimator.}
The classical MC estimator~\cite{mc} enforces a \emph{uniform} accuracy guarantee over all \(n\) nodes, which introduces an unavoidable \(\log(n)\) factor via a union bound (appearing as \(\frac{\log(n)}{\delta}\)). This is well matched to regimes with \emph{extreme} precision (e.g., \(\delta\approx 1/n\)) but overly conservative for \emph{moderate} ones (e.g., \(\delta\approx 1/\log(n)\)). Such a global \(\log(n)\) term becomes a bottleneck for overall approximation. Motivated by this, our idea is, we can first conduct a \emph{modest} number of random walks to obtain an initial estimate that, with high probability, contains all required target nodes—i.e., a node set \(R_{\text{out}}\)—and then we restrict high-precision estimation only to nodes in \(R_{\text{out}}\) via another batch of random walks. This way, we can replace the global \(\log(n)\) term with a \(\log(1/\delta)\) dependence governed by the candidate size. Accordingly, we present \algoa{} (Algorithm~\ref{algo:discover}), which consists of the following two phases:

\input{floats/algo_discovery}

\begin{itemize}[leftmargin=10pt,itemsep=0pt,parsep=0pt,topsep=2pt,partopsep=0pt]
\item {Discovery (lines 1-4).} We first conduct \(T_1\) \(\alpha\)-RWs to find an initial candidate node set \(V_{\mathrm{out}}\) such that, with high probability, all nodes of interest \(V_{\ge \delta}:=\{u:\pi(s,u)\ge \delta\mid u\in V\}\) are included.
\item {Estimation (lines 5-9).} We then restrict high-precision estimation to \(V_{\mathrm{out}}\) by running \(T_2\) additional walks. By filtering nodes with tiny values, we'll ensure that the final candidate set $V_{\mathrm{out}}$ contains only meaningful nodes that $V_{\mathrm{out}} \subseteq V_{\ge \delta/(5e)}:=\{v:\pi(s,v)\ge \delta/(5e)\}$. 
\end{itemize}
By rigorously setting \(T_1\) and \(T_2\), we conclude the following approximation guarantee in Lemma~\ref{lemma:discover}.
\begin{lemma}\label{lemma:discover}(Proof in Appendix~\ref{sec:proof:lemma:discover}.)
With probability $\geq1-p_f$, the node set \(V_{\mathrm{out}}\) returned by Algorithm~\ref{algo:discover} satisfies: 
$$V_{\ge \delta}\subseteq V_{\mathrm{out}}\subseteq V_{\ge \delta/(5e)}.$$
\end{lemma}

Based on this discovered candidate set from Algorithm~\ref{algo:discover}, we can then estimate SSPPR values only for nodes in \(R_{\text{out}}\) using the standard MC. 
For SSPPR-A query, we further run $$T = 256[\,{\log\!\big(\tfrac{1}{\delta/(5e)}\big)+\log\!\big(\tfrac{2}{p}\big)}]/[{c^2\,(\delta/(5e))}]=\bigo\!\left({\log(1/\delta)}/{\delta}\right)$$ independent $\alpha$-RWs and for each \(u\in V_{\mathrm{out}}\) estimate \(\pi(s,u)\) by empirical frequency, we then answer SSPPR-R query with overall failure probability \(\le p_f\) combining the Chernoff and union bound. 

For SSPPR-A query, we then replace the error threshold $\delta$ with $\epsilon$. By Lemma~\ref{lemma:discover}, with probability \(\ge 1-p_f/2\) we obtain
\(
V_{\ge \varepsilon} \subset V_{\mathrm{out}} \subset V_{\ge \varepsilon/(5e)}.
\)
Again, we execute $T = 1024{\log\!\big(\tfrac{2}{p_f}\big)}/{\big(\tfrac{\varepsilon}{5e}\big)^2}=\bigo\!\big(\tfrac{1}{\varepsilon^2}\big)
$ independent $\alpha$-RWs from $s$, and estimate \(\pi(s,u)\) for \(u\in V_{\mathrm{out}}\) by the empirical frequency with absolute tolerance \(\varepsilon/(5e)\). Since \(\pi(s,u)\ge \varepsilon/(5e)\) for all \(u\in V_{\mathrm{out}}\), the Chernoff bound yields 
$$\pr[\text{fail at }u]
\le 2\exp\!\Big(-\tfrac{1}{3}\,\pi(s,u)\Big(\tfrac{\varepsilon/(5e)}{\pi(s,u)}\Big)^{\!2}\,T\Big)
\le 2\exp\!\Big(-\ln\tfrac{2}{p} - \tfrac{1}{\pi(s,u)}\Big)
\le \tfrac{p_f}{2}\cdot\pi(s,u).$$ %where the last inequality uses \(e^{-1/x}\le x/2\) for \(x\in(0,1]\), absorbed into the constant in \(T\).  
Summing over \(u\) with \(\sum_u \pi(s,u)=1\) gives total estimation failure at most \(p_f/2\). Together with the probability in discovery phase, the overall failure probability is bounded at most \(p_f\). 

Combining these steps, we finally achieve optimized complexities: SSPPR-A at \(\mathcal{O}(1/\varepsilon^2)\) and SSPPR-R at \(\mathcal{O}(\log(1/\delta)/\delta)\), supporting the main claims in Theorem~\ref{thm:ssppr-a-upper}.

%% file: floats/algo_discovery.tex
\begin{algorithm}[!t]
\caption{\algoa{} $(G,s,\alpha,\delta,p_f)$}\label{algo:discover}
\DontPrintSemicolon
\KwInput{Graph $G$, source $s\in V$, decay factor $\alpha$, threshold $\delta$, failure probability $p_f$.}
\KwOutput{$V_{\mathrm{out}}\subseteq V$.}

$T_1 \gets \tfrac{1}{\delta}\,\bigl(\log \tfrac{1}{\delta}+\log \tfrac{2}{p_f}\bigr)$; $T_2 \gets \tfrac{20}{\delta}\,\bigl(\log \tfrac{1}{\delta}+\log \tfrac{2}{p_f}\bigr)$; $V^{\rm raw}_{\text{out}}\gets \varnothing$; $H(\cdot)\gets \mathbf{0}$;\\
\For{$i \gets 1$ \KwTo $T_1$}{
  $v \gets$ end of an $\alpha$-RW from $s$;\\
  $V^{\rm raw}_{\text{out}}\gets V^{\rm raw}_{\text{out}}\cup \{v\}$;
}
$\theta \gets 10\bigl(\log \tfrac{1}{\delta}+\log \tfrac{2}{p_f}\bigr)$; \# For refining the returned node set.\\
\For{$i \gets 1$ \KwTo $T_2$}{
  $v \gets$ end of an $\alpha$-RW from $s$;\\
  $H[v] \gets H[v] + 1$;
}
$V_{\mathrm{out}} \gets \{ u \in V^{\rm raw}_{\text{out}}\mid H[u] \ge \theta \}$; \\
\Return{$V_{\mathrm{out}}$}
\end{algorithm}

\begin{comment}
\begin{wrapfigure}{l}{0.4\textwidth} % r or l; width as needed
  \vspace{-\intextsep} % optional: tighten vertical gap
  \begin{minipage}{\linewidth}
    \begin{algorithm}[H]
    \caption{\algoa{}}\label{algo:discover}
    \DontPrintSemicolon
    \KwInput{Graph $G=(V,E)$, source $s\in V$, decay factor $\alpha$, threshold $\delta$ and failure probability $p_f$.}
    \KwOutput{$R_{out}\subseteq V$.}
    $T_1 \gets \tfrac{1}{\delta}\,\bigl(\log \tfrac{1}{\delta}+\log \tfrac{2}{p}\bigr)$;\\
    $T_2 \gets \tfrac{20}{\delta}\,\bigl(\log \tfrac{1}{\delta}+\log \tfrac{2}{p}\bigr)$;\\
    $V_{\mathrm{out}}\gets \varnothing$; $H(\cdot)\gets \mathbf{0}$;\\
    \For{$i \gets 1$ \KwTo $T_1$}{
      $v \gets$ end of an $\alpha$-RW from $s$;\\
      $V_{\mathrm{out}}\gets V_{\mathrm{out}}\cup \{v\}$;
    }
    \For{$i \gets 1$ \KwTo $T_2$}{
      $v \gets$ end of an $\alpha$-RW from $s$;\\
      $H[v] \gets H[v] + 1$;
    }
    $\theta \gets 10\bigl(\log \tfrac{1}{\delta}+\log \tfrac{2}{p}\bigr)$;\\
    $V_{\mathrm{out}} \gets \{ v \in V_{\mathrm{out}}\mid H[v] \ge \theta \}$;\\
    \Return{$R_{}$}
    \end{algorithm}
    \vspace{-2\intextsep}
  \end{minipage}
\end{wrapfigure}
\end{comment}

%% file: text/linear_alg.tex
\subsection{\algor{}}
\label{sec:distwalks}
Building on \algoa{}, we observe that the MC approach is highly effective for producing an initial, relatively accurate estimate. Leveraging this insight, we introduce \algor{}, which combines \algoa{} with a new decomposition of SSPPR computation based on paths. We first review the bottlenecks of existing techniques, to motivate our idea. %This novel algorithm attains \emph{linear} complexity in the number of edges and, for the first time, substantially sharpens the upper bounds for SSPPR-R queries, as formalized in Theorem~\ref{thm:ssppr-r-upper}.

%\mypara{Evaluate Personalized PageRank with from Heavy to Light}
\mypara{Current Bottlenecks.}
Current SOTA methods (e.g., FORA and SpeedPPR) combine push and walk procedures based on a key \emph{invariant} decomposition of PPR: $\pi(s,t) = \hat{\pi}(s,t) + \sum_{u \in V} r(u)\pi(u,t),$ where $\hat{\pi}(s,\cdot)$ is the current estimation vector and $r(\cdot)$ denotes the residual values, which is based on the iterative definition of PPR vectors of $\pi(s,\cdot) = (1-\alpha)\mathbf{P}\pi(s,\cdot)+ \alpha \mathbf{1}_{s}.$ These methods typically perform push operations to initialize $\hat{\pi}(s,\cdot)$, followed by $\alpha$-random walks to estimate the residual term using $r(u)$ on active nodes. 
While this hybrid design achieves notable improvements over MC, we identify two critical bottlenecks:
\begin{enumerate}[label=(\roman*)]
    \item The deterministic push phase expends substantial computation on \emph{heavy} nodes with large $\pi(s,v)$, even though their scores could be efficiently estimated with only a few random walks.
    \item The walk phase often revisits nodes whose estimates are already accurate, resulting in redundant computations and poor reuse of previously gathered information for refinement.
\end{enumerate}
These inefficiencies originate from the traditional decomposition itself, which overlooks the deeper relationship between \emph{PPR and its underlying path summations}. 
Motivated by this insight, we propose a novel \emph{path-based decomposition} that yields more accurate initial estimates while fully reusing observed information, eliminating redundant $\alpha$-RW steps and reducing computational complexity.

%We observe that above decomposition stems from the fundamental relationship between PPR and the sum over all paths in the graph. With this insight, we can find more powerful decompositions, which can significantly narrow the gap in different situations.

\mypara{A New Path-based Decomposition.} Let $\overleftrightarrow{p_G}(u,v)$ denote the set of all directed paths between any node pair $(u,v)$ in graph $G$. For any path ${\footnotesize\overleftrightarrow{p}} \in \overleftrightarrow{p_G}(u,v)$, we define its weight $w_G({\footnotesize\overleftrightarrow{p}})$ as the probability that an $\alpha$-decay random walk (RW) traverses ${\footnotesize\overleftrightarrow{p}}$. For example, for a single-edge path ${\footnotesize\overleftrightarrow{p_1}}=\{u,v\}$ with $(u,v)\in E$, we have $w_G({\footnotesize\overleftrightarrow{p_1}})=\frac{1-\alpha}{d_{\mathrm{out}}(u)}$. For two paths ${\footnotesize\overleftrightarrow{p_1}} \in \overleftrightarrow{p_G}(u_2,u_3)$ and ${\footnotesize\overleftrightarrow{p_2}} \in \overleftrightarrow{p_G}(u_1,u_2)$, their concatenation ${\footnotesize\overleftrightarrow{p_1}}\!\circ\!{\footnotesize\overleftrightarrow{p_2}} \in \overleftrightarrow{p_G}(u_1,u_3)$ satisfies $w_G({\footnotesize\overleftrightarrow{p_1}}\!\circ\!{\footnotesize\overleftrightarrow{p_2}})=w_G({\footnotesize\overleftrightarrow{p_1}})\,w_G({\footnotesize\overleftrightarrow{p_2}})$. Without loss of generality, we assume the graph contains no dangling nodes and define the weight of zero-length paths to be zero. Under this formulation, the PPR can be equivalently expressed as {$\pi(s,t)=\alpha \sum_{{\footnotesize\overleftrightarrow{p}} \in \overleftrightarrow{p_G}(s,t)} w_G({\footnotesize\overleftrightarrow{p}})$}, where the summation aggregates contributions from all feasible paths between $s$ and $t$. Building upon this, we then introduce a new path-based \emph{decomposition} that focuses on estimating PPR values of any subset of nodes. This key component plays a fundamental role in our final estimator for SSPPR-R query, replacing the traditional Push-and-Walk paradigm while allowing \emph{sufficient re-using} of intermediate estimations and avoid re-visiting redundant nodes.

Specifically, given any subset $X\subset V$ of nodes with source node $s \notin X$, we define the \emph{boundary} set \boundary{X}, to contain all nodes in $V\backslash X$ which have an out-edge to node set $X$. As such, any path ${\footnotesize\overleftrightarrow{p}}$ starting from $s$ and ending at nodes in $X$ can be linked as ${\footnotesize\overleftrightarrow{p_1}}\circ {\footnotesize\overleftrightarrow{p_2}}$, where ${\footnotesize\overleftrightarrow{p_1}}$ is a path ending in $\boundary{X}$ and ${\footnotesize\overleftrightarrow{p_2}}$ starts in $\boundary{X}$, has its first edge entering $X$, and then stays entirely within $X$ (or intuitively, it is the path ever since the last visit of outside $X$).
This partition exists and is unique for each \({\footnotesize\overleftrightarrow{p}} \in \overleftrightarrow{p_G}(s,u)\) for any source $s$ and node $u\in V$. Thereby, for any node $t \in X$ we have that 
\begin{align}\label{Eq:decom1}
\pi(s,t) &= \alpha \begin{matrix}\sum_{u\in \boundary{X}} \sum_{{\footnotesize\overleftrightarrow{p_1}} \in \overleftrightarrow{p_G}(s, u)}\end{matrix} w_G({\footnotesize\overleftrightarrow{p_1}}) \begin{matrix}\sum_{{\footnotesize\overleftrightarrow{p_2}} \in \overleftrightarrow{p_G}(u, t|X)}\end{matrix} w_G({\footnotesize\overleftrightarrow{p_2}}) \\ &= 
\begin{matrix}\sum_{u\in \boundary{X}}\end{matrix} \pi(s, u) \begin{matrix}\sum_{{\footnotesize\overleftrightarrow{p_2}} \in \overleftrightarrow{p_G}(u,t|X)}\end{matrix} w_G({\footnotesize\overleftrightarrow{p_2}}),
\end{align}
where $\overleftrightarrow{p_G}(u,t|X)$ denotes paths that only contain nodes in $X$ except the starting point. Equation~\ref{Eq:decom1} indicates that $\pi(s,t)$ can be expressed as the combination of:  
(i) intermediate contributions from the boundary set via $\pi(s,u)$ for $u\in\boundary{X}$, and  
(ii) a restricted path–weight aggregation over paths that stay entirely inside $X$ after entering it.  
Consequently, once accurate estimates of $\pi(s,u)$ are available for all $u\in\boundary{X}$, the remaining contribution to $\pi(s,t)$ for $t\in X$ can be computed using only nodes and edges within $X$, without revisiting any node outside $X$. A key observation is that the residual term  
${\sum_{p_2 \in \overleftrightarrow{p_G}(u,t\mid X)} w_G(p_2)}$  
can be equivalently expressed through PPR values on a carefully constructed auxiliary graph $G_X$.  
This equivalence is formalized below:

\begin{lemma}[Decomposition](Proof in Appendix~\ref{sec:proof:lemma:newdecom}.)
\label{lemma:newdecom} Given graph $G$, subset $X$ and source node $s \notin X$, we have \begin{align}
\label{eq:decomp-pi}
\begin{matrix}
\pi_G(s,t) = \frac{1}{\alpha}\sum_{u \in \boundary{X}}{\pi_G(s,u)\pi_{G_X}(u, t)d_{out}^{G_X}(u)}/{d_{out}^{G}(u)},\end{matrix}
\end{align} 
where $G_X$ is constructed by: \begin{enumerate}[label=(\textcolor{myblue}{\Alph*}).] \item $V_{G_X} = X \cup \boundary{X} \cup \{v_X\}$, where $v_X$ is a new auxiliary node. \label{construction_1} 
\item For nodes $u\in V$ and $t\in X$, if directed edge $(u,t)\in E$, add this edge to $E_{G_X}$. \label{construction_2} 
\item For nodes $t\in X$ and $u \in V/X$, if directed edge $(t,u)\in E$, add edge \((t, v_X)\) to $E_{G_X}$. \label{construction_3} 
\end{enumerate} \end{lemma}

Intuitively, Step \ref{construction_1} eliminates any new dangling nodes introduced by the out-edge deletion, and Step~\ref{construction_2} preserves all edges entering $X$, while Step~\ref{construction_3} redirects all edges leaving $X$ to the sink-like node $v_X$, ensuring that $d_{\mathrm{out}}^{G_X}(t)=d_{\mathrm{out}}^{G}(t)$ for all $t\in X$. This construction precisely encodes the restricted path aggregation in Equation~\ref{Eq:decom1}.  
We emphasize that $G_X$ is a conceptual object used for analysis and we avoid materialize it explicitly to maintain a bounded complexity later.

\input{floats/algo_dist_walk}
By Lemma~\ref{lemma:newdecom}, we can now construct an on-hand estimator to solve PPR estimations for any interested nodes in the given subset $X$. Based on an initially accurate estimation $\{\hat{\pi}_G(s,u)\mid u\in \boundary{X}\}$ (e.g., by running MC), the remaining part becomes $\alpha$-RW simulation for $\pi_{G_X}(u,\cdot)$ on graph $G_X$ with distributed starting nodes $u\in \boundary{X}$ valued $\mathbf{\hat{S}_r}(u):=\frac{1}{\alpha}{\hat{\pi}_G(s,u)d_{\mathrm{out}}^{G_X}(u)}/{d_{\mathrm{out}}^{G}(u)}$. This sourced $\alpha-$RW simulation can be readily achieved by first constructing an alias table with normalized probability distribution $\mathbf{\hat{\bar{S}}_r}:=\mathbf{\hat{S}_r}/\texttt{SUM}(\mathbf{\hat{S}_r}),$
$\text{where } {\small\texttt{SUM}(\mathbf{\hat{S}_r}) =\sum_{u \in \boundary{X}}\mathbf{\hat{S}_r}(u)}$. This step costs $\bigo(|\boundary{X}|)$ for construction and supports sampling in $\bigo(1)$ time and then simulate a proper amount of $N_r$ $\alpha-$RWs. This gives rise to our \emph{intermediate} Algorithm~\ref{algo:dist_walk} ($\algorr{}$), which can accurately estimate PPR values of nodes in target set $X$ larger than the pre-defined threshold $\delta_r$ (returned as node set $X_{out}$), based on some initial estimation on the boundary set $\boundary{X}$. With proper parameter settings, we guarantee the quality of algorithm $\algorr{}$ with relative error $c_r$ and failure probability $p_r$:

\begin{lemma}[\algorr{} Guarantee](Proof in Appendix~\ref{sec:proof:lemma:one_round}.)
\label{lemma:one_round}
Assume $n$ is sufficiently large and initial estimation satisfies $\forall u\in \boundary{X}, \hat\pi_G(s, u)/\pi_G(s, u) \in [1-c, 1+c]$ for constant $c\in (0,\frac{1}{2}]$. Then with probability $\geq 1 - p_r$, the output set $X_{out}$ and updated estimation $\hat{\pi}_G(s,\cdot)$ from Algorithm~\ref{algo:dist_walk} satisfy: 
\begin{enumerate}[label=(\roman*)]
\item $X_{\geq \delta_r\cdot\texttt{SUM}(\mathbf{\hat{S}_r})} \subset X_{out}$. That is---$\forall x \in X$ with $\pi_G(s,x) \geq \delta_r\cdot\texttt{SUM}(\mathbf{\hat{S}_r})$, it's included in set $X_{out}$. 
\item $\forall t \in X_{out}, \hat\pi_G(s, t)/\pi_G(s, t) \in [(1-c)(1-c_r), (1+c)(1+c_r)]$ holds for relative error guarantee. %, where $0<c_r<1/2$ is the error parameter in the algorithm.
\end{enumerate}%{\color{red}What is $c_r$?}
\end{lemma}

By Lemma~\ref{lemma:one_round}, Algorithm~\ref{algo:dist_walk} provides a basic mechanism for solving the SSPPR-R query: one may simply add an auxiliary source  $s'$ together with a directed edge $(s, s')$ to $G$, and take the initial target set $X = V\setminus\{s'\}$, with the boundary consisting solely of the source node $s$, for which $G_X =G, \hat{\pi}_G(s',s') = \pi_G(s',s')=\alpha$ is known exactly. 
Setting the round parameters to $c_r = c$, $p_r = p$, and $\delta_r = \delta$ makes \algorr{} degenerate into the raw \algoa{} procedure, which performs $\alpha$-RWs directly from $s$ similarly. However, simulating $\alpha$-RWs reveals a key structural property: nodes with larger PPR values relative to $s$ (\emph{heavy} nodes) are encountered far more frequently and therefore require significantly fewer samples to obtain accurate estimates compared to nodes with very small PPR values (\emph{light} nodes). This motivates us a progressive refinement strategy. Instead of estimating all nodes at once, we expand the estimation frontier gradually—from heavy nodes to increasingly lighter nodes—by decreasing the threshold round-by-round, eventually reaching the target $\delta$. 
This multi-round refinement greatly integrates with core of \algorr{} that estimated heavy nodes will not be re-visited in subsequent rounds (as we only utilize $G_X$), avoiding redundant costs. To implement the above refinement strategy, we introduce a shrinking factor $\Delta>1$ that exponentially decreases the detection threshold from an initial value $\delta_0\in\Theta(\frac{1}{n})$ down to the target threshold $\delta_R\texttt{SUM}(\mathbf{\hat{S}_r}) \leq \delta$ over $R\in\bigo\log_\Delta(1/(n\delta))$ rounds. This yields the main procedure of \algor{} framework (Algorithm~\ref{algo:dist_walks}) for solving the SSPPR-R query with great complexity.

\input{floats/algo_dist_walks}

We begin by identifying an initial heavy node set $V_0$ using \algoa{} (line~\(2\)), and estimating their PPR values accurately via $\alpha$-RWs (line~\(3\)). The complement $X_0 = V \setminus V_0$ then becomes the initial target set for refinement. For each round $r=1,\dots,R-1$, we invoke the single-round estimator \algorr{} on the current target set $X_{r-1}$, producing a newly resolved node set $V_r$ with certified relative-error guarantees. The remaining unresolved nodes form the next target set $X_r = X_{r-1} \setminus V_r$. As the remaining PPR scores decreases geometrically by factor $\Delta$ across rounds, the algorithm progressively expands its estimation coverage from heavy nodes toward light nodes. After $R$ rounds, all nodes with $\pi(s,x)\ge\delta$ will then be included ideally. Although chaining multiple refinement rounds may seem to introduce cumulative approximation error, we show that all such errors can be rigorously controlled. By carefully choosing the parameters $\delta_r$, $c_r$, and $p_r$ in each round (line~\(5\)), the overall estimator remains within desired SSPPR-R query guarantees, formally established as below:

\begin{lemma}[\algor{} Guarantee](Proof in Appendix~\ref{sec:proof:lemma:all_rounds}.)
\label{lemma:all_rounds}
Algorithm~\ref{algo:dist_walks} correctly solves the SSPPR-R query: it returns an estimated SSPPR vector $\hat{\pi}(s,\cdot)$ that with probability at least $1-p_f$, it holds $|\hat{\pi}(s,u)-\pi(s,u)| \le c\cdot \pi(s,u)$ for all node $u \in V$ with PPR value $\pi(s,u)\ge \delta$.
\end{lemma}

\subsection{Complexity analysis}
We next analyze the overall complexity. Recall that we avoid explicitly constructing $G_X$ during each round of estimation. We first describe the implementation details needed to derive the result.

\mypara{Implicit Construction of $G_X$.}
Instead of explicitly materializing $G_X$ in each round $r$ (denoted as $G_X^{(r)}$), we observe that it suffices to maintain and update the out-neighbors of nodes newly estimated in that round (i.e., the node set $V_r$). Initially, we construct $G_X^{(0)}$ by augmenting the original graph $G$ with a global auxiliary node $v_X$. Then, immediately after the node set $V_r$ is estimated in round $r$, we update $G_X^{(r-1)}$ to $G_X^{(r)}$ by processing all edges incident to each node $v_{\text{new}} \in V_r$ using the following steps, while maintaining the updated out-degree $d_{out}^{G_X}$.

\begin{enumerate}[label=(\roman*)]
\item For any incoming edge $e=(u,v_{\text{new}})$ where $u\in X_r$ (the remaining nodes to be estimated), we redirect $e$ to point to $v_X$, i.e., update $e$ to $(u,v_X)$.
\item For any incoming edge $e=(u,v_{\text{new}})$ where $u\notin X_r$ (nodes already estimated), we remove $e$.
\item For any outgoing edge $e=(v_{\text{new}},u)$ where $u\notin X_r$ (nodes already estimated), we remove $e$.
\item If after steps (b) or (c) an already estimated node $u\notin X_r$ satisfies $d_{out}^{G_X}(u)=0$, we remove $u$.
\end{enumerate}

These steps dynamically recover $G_X^{(r)}$ in each round without reconstructing it from scratch. Consequently, the update cost for each new node $v_{\text{new}}\in V_r$ is bounded by $\bigo(d_{in}(v_{\text{new}})+d_{out}(v_{\text{new}}))$. Since $\bigcup_r V_r$ equals the initial target set $X_0$, the total cost of all updates is bounded by $\bigo(m)$.

Finally, we note that the construction of $G_X$ may introduce parallel edges, which can be handled directly by algorithm $\algoa$. Alternatively, parallel edges can be avoided by adding $\bigo(n)$ auxiliary (padding) nodes $u'_1,\dots,u'_n$ and redirecting edges in FIFO order. This modification does not increase the number of nodes in recursive applications of the construction, since the same padding nodes can be reused. The remaining cost arises from $\alpha$-RW sampling and alias-table construction. Each round performs 
$N_r=\Theta((\log |X_r|/p_r)/(c_r^2\delta_r))$ independent $\alpha$-random walks and builds an $\bigo(n)$ alias table. Since the term $\texttt{SUM}(\mathbf{\hat{S}_r})\delta_r$ shrinks geometrically and our parameter settings satisfy $\delta_r=\Omega(1/n)$ and $c_r=\Omega(1/R)$, the algorithm terminates in $\bigo(\log(1/(n\delta)))$ rounds. This yields an overall complexity of
$\bigo\!\left(m+n\log n\,\log^3(1/(n\delta))\right)$ for \algor{}.

\mypara{Reducing $\log$ factors.}
So far, we have presented the core framework of \algor{}, which establishes a computational upper bound of
$O(m+n\log n\,\log^3(1/(n\delta)))$ for the SSPPR-R query. While this bound already demonstrates near-optimality in dense regimes, the logarithmic factor $\log^3(1/(n\delta))$ can be further improved to $\log(\frac{\log n}{m\delta})$ through additional algorithmic optimizations, which matches the bound stated in Theorem~\ref{thm:ssppr-r-upper}. To achieve this improvement, we introduce the following two technical refinements without altering the overall framework:

\begin{enumerate}[label=(\roman*)]
\item \textbf{Adaptive thresholds for $\delta_r$.}
We employ a doubling strategy to dynamically determine the threshold $\delta_r$ in each round. This allows the computational cost to be amortized across rounds, scaling as $\bigo(n/R)$ instead of the worst-case $\bigo(n)$ per round.

\item \textbf{Global error analysis.}
We revisit the accuracy guarantees from a global perspective. Instead of relying on worst-case error propagation from previous rounds (as in Lemma~\ref{lemma:one_round}), we analyze the cumulative statistics of the estimator, enabling tighter variance bounds and improved overall complexity.
\end{enumerate}

The details of above two improvements are displayed in Appendix~\ref{sec:proof:log} for better readability and completeness. Combining all these augments, we claim that Theorem~\ref{thm:ssppr-r-upper} holds. \qed

%% file: floats/algo_dist_walk.tex
\begin{algorithm}[!t]
\caption{\(\mathsf{\algorr{}}     \ (G,X,\hat{\pi}_G(s,\cdot),\alpha,c_r,\delta_r,p_r)\)}
\label{algo:dist_walk}
\DontPrintSemicolon
\KwInput{Graph \(G\); target set \(X\); initial estimation $\hat{\pi}_G(s,\cdot)$; decay factor \(\alpha\); error parameter \(c_r\); threshold \(\delta_r\); failure probability \(p_r\).}
\KwOutput{Detected set $X_{out}\in X$ and updated estimation $\hat{\pi}_G(s,\cdot).$}

Construct $G_X$ following steps \ref{construction_1}\ref{construction_2}\ref{construction_3}\\
$ \forall u\in \boundary{X}, \mathbf{\hat{S}_r}(u) \gets {d_{\mathrm{out}}^{G_X}(u)\hat{\pi}_G(s,u)}/{(\alpha\,d_{\mathrm{out}}^{G}(u))};$ $\texttt{SUM}(\mathbf{\hat{S}_r}) \gets \sum_{u \in \boundary{X}}\mathbf{\hat{S}_r}(u), \mathbf{\hat{\bar{S}}_r} \gets \mathbf{\hat{S}_r}/\texttt{SUM}(\mathbf{S_r})$;\\
Construct an alias table for \(\mathbf{\hat{\bar{S}}_r}\); \\
Use it to execute $X_{out} \gets \algoa{}(G_X, \mathbf{\hat{\bar{S}}}, \alpha, \delta_r, p_r)$;\\

$N_r\gets {256\alpha\log(n/p_r)}/\left({c_r^2\delta_r}\right), H(\cdot)\gets \mathbf{0};$ \\

\For{\(i=1,2,\dots,N_r\)}{
Draw a node \(u\sim \mathbf{\hat{\bar{S}}}\); \\
Let $v \gets$ end of i-th $\alpha$-RW from \(u\) on $G_X$;\\
\textbf{if}\(v\in X_{out}\) \textbf{then} \(H[v]\gets H[v]+\frac{1}{N_r}\);\\
}

\(\hat\pi_G(s,t) \gets H(t)\cdot \texttt{SUM}(\mathbf{\hat{S}_r}) \text{ for } \forall t \in X_{out}; \) \\

\Return{\(X_{\text{out}}\,,\,\hat{\pi}_G(s,\cdot)\)}
\end{algorithm}

\begin{comment}
\begin{algorithm}[!h]
\caption{\(\mathsf{\algor{}}\)}
\label{algo:dist_walks}
\DontPrintSemicolon
\KwInput{Graph \(G\); source set \(V_{\text{source}}\); target set \(X\); source distribution \(S\) (supported on \(V_{\text{source}}\)); lower bound \(q_w\) on \(\sum_{x\in X}\pi(S,x)\); decay \(\alpha\); error \(c\in(0,1)\); threshold \(\delta\); failure probability \(p\).}
\KwOutput{Estimates \(\hat\pi(S,\cdot)\) on \(X\) and detected set \(V_{\text{out}}\subseteq X\).}
\BlankLine
\(\displaystyle N \gets C \cdot \frac{\log|X|+\log(1/p)}{q_w\,c^2\delta}\) for a sufficiently large constant \(C\);\;
\(N_{\mathrm{eff}}\gets 0\);\;
\(W[x]\gets 0\) for all \(x\in X\);\;
\BlankLine
Preprocess an alias table for \(S\) in \(O(|V_{\text{source}}|)\) time;\;
\For{\(w=1,2,\dots,N\)}{
  Sample \(s\sim S\) via the alias table;\;
  Run an $\alpha$-RW from \(s\) to its terminal vertex \(v\), without transportation to \(V_{pad}\) at \(s\);\;
  \If{\(v\in X\)}{
    \(N_{\mathrm{eff}}\gets N_{\mathrm{eff}}+1\);\;
    \(W[v]\gets W[v]+1\);\;
  }
}
\(V_{\text{out}}\gets \{x\in X:\; W[x] \ge (\delta/32)\,N_{\mathrm{eff}}\}\);\;
\(\hat\pi(S,x) \gets \begin{cases} W[x]/N_{\mathrm{eff}}, & x\in X\ \text{and}\ N_{\mathrm{eff}}>0,\\[2pt] 0, & \text{otherwise.}\end{cases}\)\;
\Return{\(V_{\text{out}},\,\hat\pi(S,\cdot)\)}\;
\end{algorithm}

\end{comment}

%% file: floats/algo_dist_walks.tex
\begin{algorithm}[!b]
\caption{$\mathsf{\algor{}}(G,s,\alpha,c,\delta,p_f,\Delta)$}
\label{algo:dist_walks}
\DontPrintSemicolon

\KwInput{Graph $G$; source node $s$; decay factor $\alpha$; error $c$; threshold $\delta$; probability $p_f$; shrinking factor $\Delta$.}

\KwOutput{Estimation $\hat{\pi}_G(s,\cdot).$}

$\hat{\pi}_G(s,\cdot) \gets \mathbf{0},\;
\delta_0 \gets \frac{\alpha}{(1+c)n\Delta},\;
R \gets \log_\Delta(\frac{2\delta_0}{\alpha\delta});$\\

$V_{0} \gets \algoa{}(G, s, \delta_0, p_f)$;\\

Cast $\frac{320R^2\log(n/p_f)}{c^2\delta_0}$ $\alpha$-RWs from $s$
to estimate $\hat{\pi}_G(s,u)$ for all $u\in V_{0}$; \\
$X_0 \gets V \backslash V_{0}$;\\

\For{$r=1,\dots,R-1$}{
    $c_r=\frac{c}{4R},\; \delta_r=\delta_0,\; p_r=\frac{p_f}{2n}$;\\
    $V_{r}, \hat{\pi}_G(s,\cdot) \gets
    \mathsf{\algorr{}}(G, G_X, X_{r-1}, \hat \pi_G(s,\cdot), \alpha, c_r,\delta_r,p_r);$\\
    $X_r \gets X_{r-1}\backslash V_r$;
}
\Return{$\hat{\pi}_G(s,\cdot)$}
\end{algorithm}

%% file: text/ssppr_r_lower_dense.tex
\section{Lower Bound of SSPPR-R Query}
\label{sec:ssppr_r_lower}
In this section, we prove the lower bound for SSPPR-R as stated in Theorem~\ref{thm:ssppr-r}. As stated from $\algoa$, we have established an upper bound query complexity as $\bigo(\log(1/\delta)/\delta)$. Since we can use $\bigo(m)$ queries to acquire the entire graph, the total query complexity thus will not surpass $\Omega(\min(m, \log(1/\delta)/\delta))$, which is exactly the \emph{optimal in arc-centric graph access model} of our result. The overall proof proceeds in three main steps: \begin{enumerate*}[label=(\roman*).]
\item  Construct a family of hard-instance graphs with distribution;
% \item  Define a distribution over that family;
\item  Reduce the SSPPR-R query to a simpler query type, i.e. examining certain amount of edges over the distribution;
\item Derive the query complexity of the simplified problem over the distribution, based on observed known graph structures through graph oracle interactions.
%Analyze the query complexity problem on the distribution via the distribution of its answers over known queried results.
\end{enumerate*}

\mypara{Hard-instance Graph family.}
To construct a family of hard-instance graphs with accurate analysis of query complexity, we model different queries under uniform measurement. To achieve this, we keep the same node subsets and their in/out-degrees for all the graphs. In this setting, the query amount required to execute equals the known edge amount. Additionally, we need to reduce the SSPPR-R problem on this graph to a node-classification problem. To achieve this, we will construct a target node set $X$ and a ground truth split $X = X_1 \sqcup X_2$, where nodes in $X_1$ and $X_2$ have marginal PPR scores violating the SSPPR-R requirement. Thus, any algorithm must predict this split. Formally, we define the graph family $\mathcal{U}(n, D, d)$ such that, for every sufficiently large $n$, it serves as a collection of hard instances for estimating SSPPR queries. This family is constructed based on a simple yet insightful core structure $\mathbf{U}$, illustrated in Figure~\ref{fig:UUr_instance}. 

\input{floats/fig_hard_instances}

The graph is constructed from five disjoint node sets: $\{s\}$, $Y_1$, $Y_2$, $X_1$, and $X_2$. For convenience, let $Y = Y_1 \cup Y_2$ and $X = X_1 \cup X_2$. The node $s$ is the designated \emph{source} node. All nodes in $X$ are treated as \emph{target} nodes, whose PPR values with respect to $s$ we aim to estimate. Each of the sets $X_1$, $X_2$, $Y_1$, and $Y_2$ contains $\Theta(n)$ nodes. We next describe how edges are constructed. Every node $y_1 \in Y_1$ has an incoming edge from the source node $s$. For each node $x_1 \in X_1$, we add $D$ incoming edges from nodes in $Y_1$ and $D-d$ incoming edges from nodes in $Y_2$. Similarly, each node $x_2 \in X_2$ receives $D$ incoming edges from $Y_2$ and $D-d$ incoming edges from $Y_1$. In addition, every node in $X$ has a self-loop. With suitable choices of the parameters, the nodes in $X_1$ and $X_2$ will have noticeably different PPR values with respect to $s$. This separation will later be formalized in Lemma~\ref{lemma:PPR_R_const}. To allow for more flexible control over the graph's parameters like node/edge count, we next introduce a modified structure, the $r$-padded instance. Intuitively, information of the graphs is encoded by edges between $Y$ and $X$. We can add some padding edges from nodes in $Y$ (or to nodes in $X$) to confuse the algorithm. \(\forall\mathrm{U} \in \mathcal{U}(n, D, d)\), we define the \emph{\(r\)-padded} instance of \(\mathrm{U}\), denoted as \(\mathrm{U}(r)=(V_{U(r)},E_{U(r)})\) in the below by involving two independent node sets $Z_X$ and $Z_Y$:
\begin{enumerate*}
    \item $V_{U(r)}=V_U\cup Z_X\cup Z_Y$, where $Z_X,Z_Y$ are two additional node sets containing exactly $r$ nodes;
    \item $E_{U(r)}=E_U \cup \{(z_x, x)\;\mid \forall x \in X, z_x \in Z_X\} \cup \{(y, z_y)\;\mid \forall y \in Y, z_y \in Z_Y\} \cup \{(z_y, z_y)\;\mid \forall  z_y \in Z_Y\}$, which incorporates extra edges coming from node set $Z_x$ to $X$ and $Y$ to $Z_Y$, along with self-loops in $Z_Y$.
\end{enumerate*}
%An illustrative example is provided in Figure~\ref{fig:UUr_instance}. 
Intuitively, when \(r\) is sufficiently large relative to \(D\), each \query{ADJ} query originating from a node in \(X\) or \(Y\) may be misdirected to the auxiliary node sets \(Z_X\) or \(Z_Y\) with non-negligible probability. This allows us to increase the hardness for dense graphs without breaking core structure of $\mathcal{U}(n, D, d)$. Then we establish the existence of such hard-instance graph family below.

\begin{lemma}[Existence of Hard-Instance Graph Family](Proof in Appendix~\ref{sec:proof:PPR_R_const})
\label{lemma:PPR_R_const}
Choose any constant \(c\in (0,\frac{1}{2}]\) and any functions \(\delta_0(n_0)\in (0,1), m_0(n_0)\in\Omega(n_0)\cap \bigo(n^{2})\). For sufficiently large $n_0$, there exist parameters \(n, D, d, r\)  (as function of \(n_0\)), such that for $\forall \mathrm{U}$ in $\mathcal{U}(n,D,d)$, its padded $\mathrm{U}(r)$ holds:%the following conditions hold for its $r-$padded instance: 
\begin{enumerate}[label=(\roman*).]
\item  The node, edge count of \(\ \mathrm{U}(r)\) is $\bigo(n_0)$ and $\bigo(m_0)$;
\item  The edge count is in $\Theta(n(r+D)) =\Theta\left(\min(m_0, {\log(\frac{1}{\delta_0})}/{\delta_0})\right)$, and the node count is in $\Theta(r+n) = \bigo(n)$. Additionally, \(D\geq \frac{3}{2}d\) and \(d\leq \log(n)\); 
\item For any $x_1,x_2\in X_1,X_2$, the PPR scores $\pi(s,x_1), \pi(s,x_2)$ are fixed values, irrelevant to the choice of nodes or the graph instance. It holds that $\pi(s,x_1), \pi(s,x_2)\geq \delta_0, \pi(s,x_1) > \frac{1}{(1-c)^2}\cdot \pi(s,x_2).$ 
\end{enumerate}
\end{lemma}

\noindent This lemma forms the basis for our proof of the lower bound by contradiction. For any SSPPR-R algorithm with $o\left(\min(m_0, \log(1/\delta_0)/\delta_0)\right)$ queries, we will prove later that on our distribution of graphs over this family  $\mathcal{U}(n,D,d)$ with $r-$padding, the algorithm must fail with high probability. 

%Specifically, the condition (ii) ensures both the count of edges corresponds to our lower bound for known edges and prevent the algorithm from a \emph{sampling} approach. The condition (iii) guarantees that an algorithm must find the exact split of $X$ into $X_1, X_2$ to solve the SSPPR-R query. 

\mypara{Distribution.}
\label{sec:distribution}
We now formalize the generation of \(\mathrm{U} \in \mathcal{U}(n, D, d)\) with node and edge labels, forming the distribution \(\Sigma(n, D, d)\) over all label permutations (see Figure~\ref{fig:generate graph}). Intuitively, \(\Sigma(n, D, d)\) is created by randomly splitting the node set \(X\) into \(X_1\) and \(X_2\), then independently sampling all edge permutations consistent with the hard-instance template. The whole procedure is as follows:

$(i).$ \emph{Splitting \(X\).}  
We first split the indices \(\{1,\dots,2n\}\) into two equal parts, which determines whether each node \(x_1,\dots,x_{2n}\) belongs to \(X_1\) or \(X_2\). Each such split corresponds to choosing \(n\) indices for \(X_1\), and all possible choices form the set \(\binom{X}{n}\). For example, in Figure~\ref{fig:generate graph}, the split \(\mathrm{SP}=\{1,3\}\) places nodes \(x_1\) and \(x_3\) in \(X_1\), while \(x_2\) and \(x_4\) belong to \(X_2\).

\input{floats/fig_generate_graph}

$(ii).$ \emph{Permuting in-edges of \(X\).}  
Next, for each node \(x_i\in X\), we independently permute its \(2D-d\) incoming edges using a permutation \(\mathrm{P}_i^{2D-d}\in \mathrm{S}^{2D-d}\). The permuted edges are then divided into two groups, corresponding to edges coming from \(Y_1\) and \(Y_2\). If \(x_i\in X_1\), the first \(D\) edges in the permutation are assigned to \(Y_1\) and the remaining \(D-d\) edges to \(Y_2\); the assignment is reversed when \(x_i\in X_2\). We denote the collection of such permutations by \(\mathrm{S}^{2D-d}(X)\). In Figure~\ref{fig:generate graph}, applying \(\mathrm{P}_3=(1,3,2)\) assigns the first and third in-edges of \(x_2\) to \(Y_2\) (white-filled nodes) and the second to \(Y_1\) (grey-filled nodes).

$(iii).$ \emph{Permuting out-edges of \(Y_1,Y_2\).}  
We then permute the out-edges of \(Y_1\) and \(Y_2\) using two independent permutations \(\mathrm{P}_{Y_1}^{n(2D-d)}\) and \(\mathrm{P}_{Y_2}^{n(2D-d)}\). For a node \(y^{(1)}_i\in Y_1\), its \(k\)-th outgoing edge connects to a node in \(X_1\) if \(\mathrm{P}_{Y_1}^{n(2D-d)}[(2D-d)(i-1)+k]\le nD\), and to a node in \(X_2\) otherwise; the same rule applies symmetrically for nodes in \(Y_2\). Figure~\ref{fig:generate graph} illustrates this step using the permutation \(\mathrm{P}_{Y_1}^{n(2D-d)}=(1,3,2,4,5,6)\), where the first four virtual nodes in the second column correspond to edges connecting to nodes in \(X_1\).

$(iv).$ \emph{Connecting \(X\) and \(Y\).}  
Finally, we match the out-edges of nodes in \(Y\) with the in-edges of nodes in \(X\). For example, for the pair \(Y_1\) and \(X_1\), we select a random bijection \(\mathrm{S}(Y_1,X_1):\{1,\dots,nD\}\to\{1,\dots,nD\}\) that maps the out-edges of \(Y_1\) to the in-edges of \(X_1\). Similar bijections (such as \(\mathrm{S}(Y_2,X_1)\)) are used for the other pairs. Each edge is represented as a tuple \((y^{(1)}_{i_j},k_j)\) or \((x^{(1)}_{i_j'},k_j')\), and the bijection specifies how these tuples are matched. In Figure~\ref{fig:generate graph}, the bijection \(\mathrm{S}(Y_1,X_1)=(4,3,2,1)\) connects the corresponding virtual nodes (black-filled) across the columns. 

Following the above procedure, we obtain a graph instance \(\mathrm{U}\in\mathcal{U}(n,D,d)\).

\begin{lemma}
\label{lemma:legal_instance}
For all possible {\small \ \(\mathrm{SP}, \mathrm{P}_i^{2D-d}, \mathrm{P}_{Y_1}^{n(2D - d)}, \mathrm{P}_{Y_2}^{n(2D - d)}, \mathrm{P} \in \binom{X}{n}\times \mathrm{S}^{2D - d}(X)\times \mathrm{S}^{n(2D - d)}(Y_1)\times\mathrm{S}^{n(2D - d)}(Y_2)\times\Pi_{Y_a\in \{Y_1,Y_2\},X_b\in\{X_1,X_2\}}S(Y_a,X_b)\)}: the above procedure links to a legal graph instance \(\mathrm{U}\in\mathcal{U}(n, D, d)\). 
\end{lemma}

\noindent Based on above Lemma, we can then construct our distribution over splits and bijections as:

\begin{definition}
\( \Sigma(n, D, d)\) is defined as the uniform distribution over
{\small\(\binom{X}{n}\times \mathrm{S}^{2D - d}(X)\times \mathrm{S}^{n(2D - d)}(Y_1)\times\mathrm{S}^{n(2D - d)}(Y_2)\times\Pi_{Y_i\in \{Y_1,Y_2\},X_i\in\{X_1,X_2\}}S(Y_i,X_i)\)}. For \(\sigma \sim \Sigma(n, D, d)\), we use \(\mathrm{U}_\sigma\) to denote the corresponding labeled graph instance \(\mathrm{U} \in \mathcal{U}(n, D, d)\) and \(\mathrm{U} \sim \Sigma(n, D, d)\) to denote a sample of this distribution.
\end{definition}

\mypara{Establishing the Main Theorem.}\label{sec:proof_main}
Now that the distribution $\Sigma(n, D, d)$ is formally defined, we can prove the target lower bound by contradiction: no SSPPR-R algorithm can have both a bounded query complexity and a constant success probability on $\Sigma(n, D, d)$.

We begin by addressing two technical issues.  First, standard SSPPR-R algorithms are formulated for simple graphs, while instances drawn from $\Sigma(n,D,d)$ often contain parallel edges. Second, Lemma~\ref{lemma:PPR_R_const} is stated for padded graphs, whereas our hard distribution is defined on the original, unpadded instances. To handle these issues, we use a two-step reduction. We first lift an arbitrary algorithm for simple graphs so that it operates on the padded multigraph instances $\mathrm{U}(r)$. We then restrict this lifted algorithm back to the unpadded graph, thereby obtaining the desired adapted algorithm on $\Sigma(n, D,d)$. Moreover, by construction of $\mathcal{U}(n, D, d)$, the nodes in $X_1$ and $X_2$ exhibit a significant gap in their PPR scores. Consequently, any algorithm that satisfies the relative-error guarantee for SSPPR-R must be able to recover the ground truth partition of $X$ into $X_1 \sqcup X_2$. We formalize both ingredients in the following lemma.

\begin{lemma}(Proof in Appendix~\ref{sec:proof:thm:algo_adapt_tot})
\label{thm:algo_adapt_tot}
Assume that there is an SSPPR-R algorithm $\algo$ on simple directed graphs with expected query complexity $T = o(\min(m_0, \log(1/\delta_0)/\delta_0))$ with threshold $\delta$, error parameter $c$ and failure probability bound $p_f$. Then for any constant $\gamma$,  there are sufficiently large parameters $n, D, d$ along with an algorithm $\mathcal{A}_{ad}$. Using at most $\gamma nD$ arc-centric queries on each $\mathrm{U}\sim \Sigma(n, D, d)$, $\mathcal{A}_{ad}$ will output the ground truth split of $X = X_1 \sqcup X_2$ with average success probability at least $1-2p_f$.
\end{lemma}

With this lemma in hand, we show that the probability that $\mathcal{A}_{ad}$  recovers the correct partition is sufficiently small to obtain a contradiction. To analyze the query complexity limit, we must first formally measure the information any algorithm receives. Each graph in the family $\mathcal{U}(n, D, d)$ is uniquely determined by the connections between $Y$ and $X$. Each such connecting edge has two labels: the "$k^{\text{th}}$ out-edge of $y_i$" and the "$k'^{\text{th}}$ in-edge of $x_j$". This can be viewed as connecting two half-edges. An arc-centric query reveals partial information of one such connection. Formally, a known edge $e_t = [ (y_i, k), (x_j, k')]$ is an event in $\Sigma(n,D,d)$ representing that the $k^{\text{th}}$ out-edge of $y_i$ connects to the $k'^{\text{th}}$ in-edge of $x_j$. In its execution, our algorithm $\mathcal{A}_{ad}$ sequentially observes the events, in other words, knows the edges $e_1, \dots, e_T$, and its task is to infer the most probable split based on this information. The optimal strategy for any algorithm is to output the split that maximizes the posterior probability (i.e., the MAP estimator): $$\text{argmax}_{\mathrm{SP}}\pr_{\mathrm{U} \sim \Sigma(n, D,d)}[\mathrm{SP} \text{ is the correct split of } \mathrm{U} \mid e_1,\dots, e_T].$$
Therefore, the expected success probability of any algorithm using at most $T$ queries is bounded as:

\begin{lemma}[Success Probability Expression](Proof in Appendix \ref{sec:proof:lemma:data_process})
\label{lemma:data_process}
For any algorithm $\mathcal{A}_{ad}$ using at most $T = \gamma nD$ arc-centric queries on an instance $\mathrm{U}\sim \Sigma(n, D, d)$, its expected probability of outputting the correct split of $X$ is bounded by $$\max_{e_1,\cdots, e_T}\max_{\mathrm{SP}}\pr_{\mathrm{U} \sim \Sigma(n, D,d)}[\mathrm{SP} \text{ is the correct split of} \;\;\mathrm{U} \mid e_1,\cdots, e_T].$$
\end{lemma}

To estimate $\pr_{\mathrm{U} \sim \Sigma(n, D,d)}[\mathrm{SP} \text{ is the correct split of } \mathrm{U} \mid e_1,\cdots, e_T]$, we rely on the following two observations:
\begin{enumerate*}[label=(\roman*).]
\item \textbf{Traversal is Hard}: An algorithm that tends to cover the entire graph needs to know $\Omega(nD)$ edges.
\item \textbf{Random Walk is Hard}: Consider an algorithm that classifies  each node $x_i$ into $X_1$ or $X_2$ in the ground truth split, by sampling random incident edges. Observe that nodes in $X_1$ differ from those in $X_2$ because they have $d$ more edges originating from $Y_1$ among their $2D-d$ in-edges. Then, to achieve a failure probability of  $p_i$, the algorithm must sample  $\Omega(\frac{D}{d}\log_{D/(D-d)}(1/p_i)) \subset \Omega(D)$ edges. Specifically, to guarantee a uniform error bound   $p_i \in \bigo(1/n)$, , the required number of samples becomes  $ \Omega(D)$ per node.
\end{enumerate*}

With these observations, we then prove the following lemma:

\begin{lemma}[Bounded Success Probability]
(Proof in Appendix~\ref{sec:proof:lemma:targetProb})
\label{lemma:targetProb}
There is a constant \(\gamma\) such that for sufficiently large \(n\), \(d \leq \log(n)\), \(D > \frac{3d}{2}\), and \(T \leq \gamma nD\), the following holds. Given any known edges $e_1, \dots, e_T$ and any split \(\mathrm{SP}\) of \(X\), we have:
$ \pr_{\mathrm{U}\sim \Sigma(n, D, d)}[\mathrm{SP} \text{ is the correct split of } \;\;\mathrm{U} \mid e_1, \cdots ,e_T] \leq \frac{1}{n}.$
\end{lemma}

\mypara{Proof Sketch.} With $o( nD)$ known edges, for most nodes $x_i \in X$, we only know a small fraction of their in-edges, as per the ``Traversal is Hard'' intuition. For the split $\mathrm{SP}$, there exists a subset $X' = X_1' \cup X_2'$, where $X_1' \subset X_1$, $X_2' \subset X_2$, and $|X_1'|, |X_2'| \in \Omega(n)$, such that each $x_i \in X'$ appears at most $o(D)$ times in $e_1, \cdots, e_T$. We then use the ‘Random Walk is Hard’ intuition to show that, for nodes in $X'$, it is hard for the algorithm to determine their class in the split. 
For each pair $x_{i1}\in X_1', x_{i2}\in X_2'$, consider a "switched" split $\mathrm{SP}'$ created by setting $x_{i1} \in X_2$ and $x_{i2} \in X_1$, while fixing the split for all other nodes. Using a Bayesian approach, we can bound the target success probability $\pr[\dots]$ by comparing it to the probabilities of such switched splits. As split $\mathrm{SP}$ is sampled from uniform distribution, applying Bayes' rule gives $$\pr_{\mathrm{U} \sim \Sigma(n, D,d)}[\mathrm{SP} \text{ is the correct split of } \mathrm{U} \mid e_1,\cdots, e_T] = \frac{\pr_{\mathrm{U} \sim \Sigma(n, D,d)}[e_1,\cdots, e_T \mid \mathrm{SP}]}{\sum_{\mathrm{SP}'}\pr_{\mathrm{U} \sim \Sigma(n, D,d)}[e_1,\cdots, e_T \mid \mathrm{SP}']}.$$
Notice that there are $\Omega(n^2)$ such switched splits. The proof reduces to showing that for any such switched split $\mathrm{SP}'$:
$\frac{\pr_{\mathrm{U} \sim \Sigma(n, D,d)}[e_1,\cdots, e_T \mid \mathrm{SP}]}{\pr_{\mathrm{U} \sim \Sigma(n, D,d)}[e_1,\cdots, e_T \mid \mathrm{SP}']} \in \bigo(n^{1-\epsilon}).$ We can then decompose this ratio into a product of terms, one for each edge $e_t$:
$\prod_{t=1}^T \frac{\pr_{\mathrm{U} \sim \Sigma(n, D,d)}[e_t \mid \mathrm{SP}, e_1,\dots, e_{t-1}]}{\pr_{\mathrm{U} \sim \Sigma(n, D,d)}[e_t \mid \mathrm{SP}', e_1,\dots, e_{t-1}]}$.
Intuitively, when the known edge $e_t$ is not connected to the switched nodes $x_{i1}$ or $x_{i2}$, it is unrelated to the switch, and the ratio is 1. For the remaining $2\gamma D$ known edges, we can establish that each such ratio is bounded by $1+ \bigo(d/D)$. Finally, since $d \in \bigo(\log(n))$, the total ratio is bounded by $ (1+ \bigo(d/D))^{2\gamma D} \in \bigo(n^{1-\epsilon})$, which leads to the result.
By Lemma~\ref{lemma:targetProb}, we conclude that the expected success probability of $\mathcal{A}_{ad}$ is less than $\frac{1}{n}$. This contradicts the guarantee from Theorem \ref{thm:algo_adapt_tot} (which stated a success probability of at least $1-2p$). This contradiction shows that the initial assumption—the existence of an algorithm $\algo$ with $T \in o\left(\min(m_0, \log(1/\delta_0)/\delta_0)\right)$ queries—must be false. Therefore, any such algorithm must require $\Omega\left(\min(m_0, {\log(1/\delta_0)}/{\delta_0})\right)$ queries in expectation. This completes the proof of our Theorem~\ref{thm:ssppr-r}.\qed

%% file: floats/fig_hard_instances.tex
\begin{figure} 
  \centering
  \vspace{-0.5\intextsep}
  \includegraphics[width=0.88\linewidth]{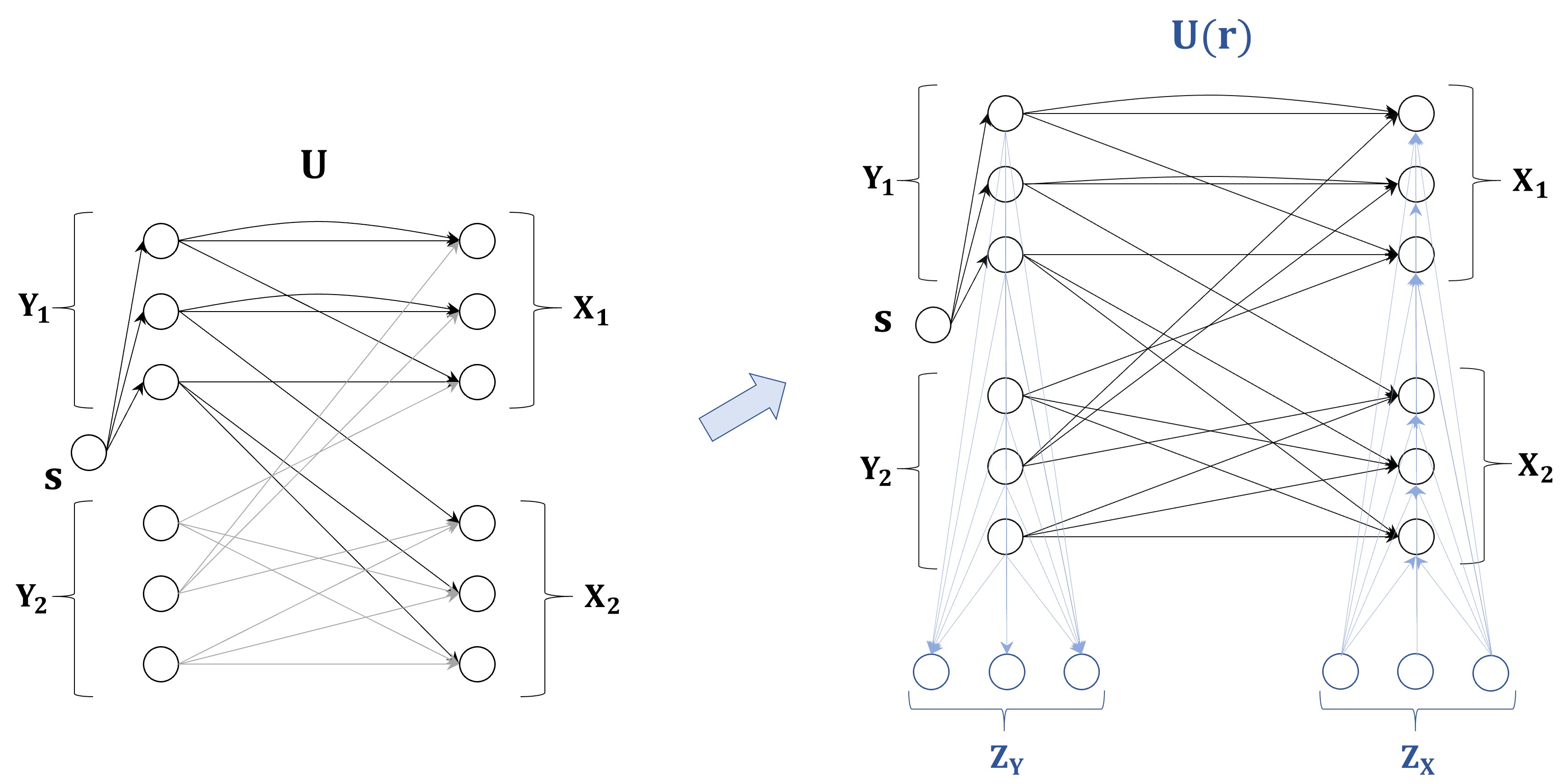}
  \vspace{-1\intextsep}
  \caption{Instance $\mathrm{U}\to r$-padded Instance $\mathrm{U}(r)$.}
  \label{fig:UUr_instance}
  %\vspace{-\intextsep}
\end{figure}

\begin{comment}

\begin{wrapfigure}{l}{0.41\textwidth} % r/l for side; set box width
  %\vspace{-\intextsep} % optional: tighten vertical gap
  \centering
  \vspace{-0.5\intextsep}
  \includegraphics[width=\linewidth]{figs/fig_UUr_instance.png}
  \vspace{-1.5\intextsep}
  \caption{Instance: $\mathrm{U}\to r$-padded $\mathrm{U}(r)$.}
  \label{fig:UUr_instance}
  \vspace{-\intextsep}
\end{wrapfigure}

\begin{figure}[!b]
    \centering
    \subcaptionbox{Instance $\mathcal{U}$\label{fig:U_instance}}%
    [0.45\linewidth]{\includegraphics[height=2.8in]{figs/fig_U_instance.jpg}}
    \hfil
    \subcaptionbox{$r$-padded instance $\mathcal{U}(r)$\label{fig:U(r)_instance}}%
    [0.45\linewidth]{\includegraphics[height=2.8in]{figs/fig_Ur_instance.jpg}}
    %\hfil
    %\subcaptionbox{Auxiliary instance $\Tilde{U}$\label{fig:UZ_instance}}%
    %[0.3\linewidth]{\includegraphics[height=2.in]{figs/fig_UZ_instance.jpg}}
%\vspace{-0.5ex}
\caption{Graph instance family $\mathcal{U}(n, D, d)=\mathcal{U}(3, 2, 1)$. (a) denotes our base graph instance family $\mathcal{U}(n, D, d)$ and (b) represents the $r$-padded instance $\mathcal{U}(r)$ based on $\mathcal{U}$.}

\label{fig:instance}
%\vspace{-2.5ex}
\end{figure}
\end{comment}

%% file: floats/fig_generate_graph.tex
\begin{figure}[!t]
\centering
\includegraphics[width=0.56\textwidth]{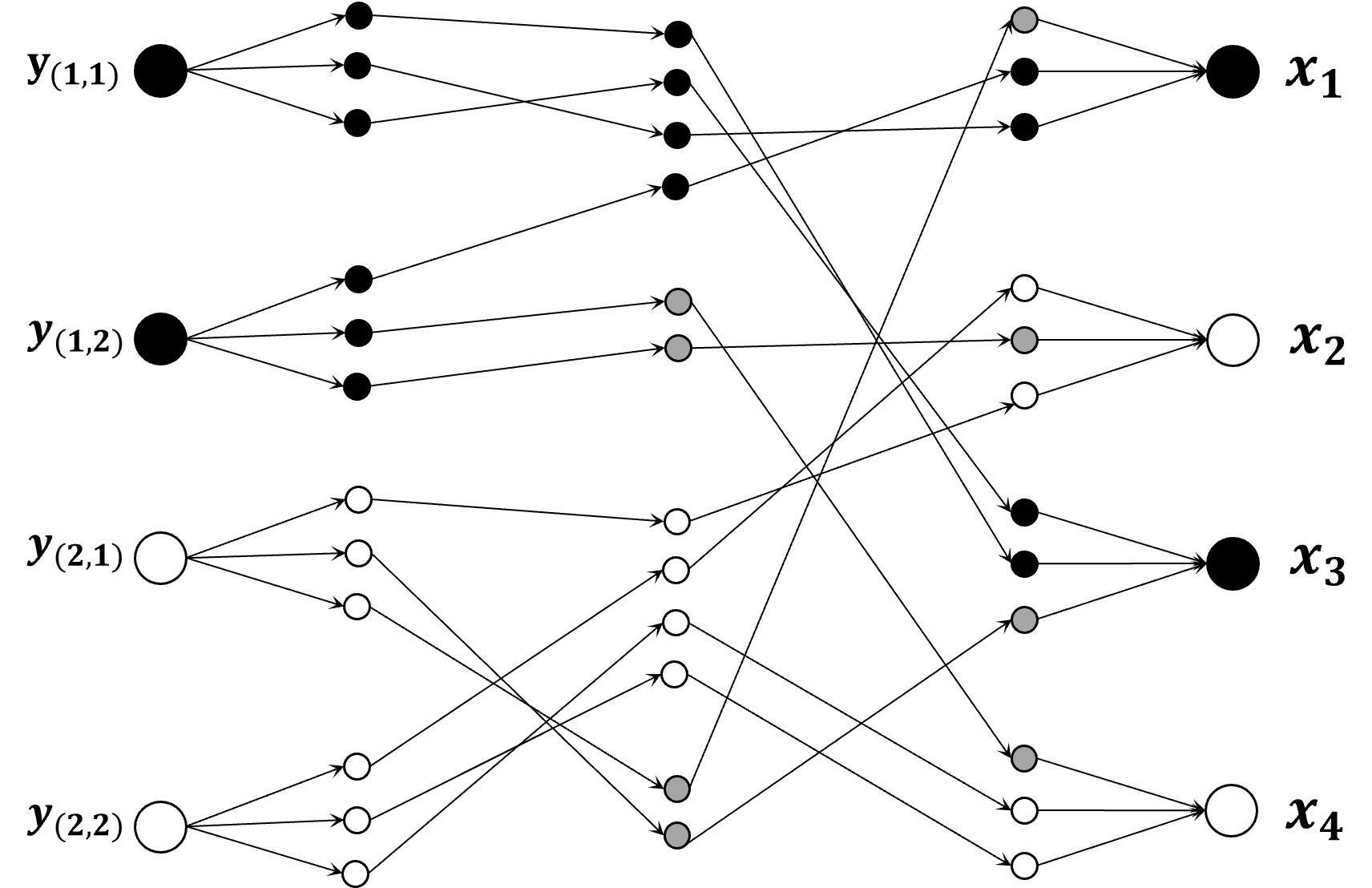}
%\vspace{1em}
\caption{An example for generating graph distribution $\Sigma(n, D, d)=\Sigma(2,2,1)$.} %\haoyu{modify the graph.}
\label{fig:generate graph}
\end{figure}

\begin{comment}
\begin{wrapfigure}{l}{0.4\textwidth} % r/l for side; set box width
  %\vspace{-0.7\intextsep} % optional: tighten vertical gap
  \centering
  \vspace{-0.7\intextsep}
  \includegraphics[width=\linewidth]{figs/fig_generate_graph.png}
  %\vspace{-1.3\intextsep}
  \caption{Output Distribution $\Sigma(2,2,1)$.}
  \label{fig:generate graph}
  \vspace{-2\intextsep}
\end{wrapfigure}
\end{comment}

%% file: text/ssppr_a_lower.tex
\section{Lower Bound of SSPPR-A Query}
In this section, we present the proof of Theorem~\ref{thm:ssppr-a-v1}. In contrast to the SSPPR-R case, the hardness in the absolute-error setting stems not from fundamental sampling variance when estimating a large PPR value $\pi(s,t)$. An estimator of $\pi(s,t)$ based on $N$ random walks behaves approximately as $\mathcal{N}(\pi(s,t),\sigma^2)$ with standard deviation $\sigma=\Theta(\pi(s,t)N^{-1/2})$.  
Thus, when $\pi(s,t)=\Theta(1)$, achieving an absolute error of $\varepsilon$ with constant success probability inherently requires $N=\Omega(1/\varepsilon^2)$ walks—an unavoidable statistical barrier. Our proof embeds a hidden binary matrix into the graph and links $\pi(s,t)$ to the count of $1$ in the matrix. Since each bit is Bernoulli, the total count of ones follows a binomial law and is well approximated by a normal distribution, paralleling the distribution of the random-walk estimator. Then, any algorithm must query a large fraction of the graph to recover this count; otherwise, its estimate of $\pi(s,t)$ incurs statistical uncertainty exceeding $\varepsilon$. % establishing the desired lower bound.

\input{floats/fig_gb_ur_1_instance}

\mypara{Graph Family and Distribution.}
We define a graph family $\mathcal{G}(D,\mathbf{b})$, where $D$ is a positive integer and $\mathbf{b}$ is a $D\times D$ binary matrix with entries $b_{ij}$.  
Each graph contains $4D+2$ vertices, partitioned into six disjoint sets $\{s\}\cup \{t\}\cup X\cup Y\cup X'\cup Y'$,
where $s$ and $t$ are the designated \emph{source} and \emph{target} nodes, respectively, and each of the sets $X,Y,X',Y'$ contains exactly $D$ nodes. Edges are added as follows. The source node $s$ connects to every node in $Y$, and every node in $X$ connects to the target node $t$. For every pair of indices $(i,j)$, the connections between the sets $(Y,Y')$ and $(X,X')$ are determined by the matrix entry $b_{ij}$:
\begin{itemize}
\item if $b_{ij}=1$, we add edges $(y_i,x_j)$ and $(y'_i,x'_j)$;
\item if $b_{ij}=0$, we add edges $(y_i,x'_j)$ and $(y'_i,x_j)$.
\end{itemize}

Thus, the matrix $\mathbf{b}$ controls how nodes in $Y$ and $Y'$ are wired to nodes in $X$ and $X'$. This $\mathbf{b}$-dependent wiring introduces structured randomness that is crucial for constructing hard instances in our lower-bound proofs. An example is illustrated in Figure~\ref{fig:gb_instance}, which shows $\mathcal{G}(D,\mathbf{b})$ with $D=2$ and 
$\mathbf{b}:=\mathbf{b}(1,\cdot)=(1,0)$ and $\mathbf{b}(2,\cdot)=(1,1)$.
  
\begin{comment}
  \begin{align*}
  \centering
  \mathbf{b}=\begin{pmatrix}
      1&0\\
      1&1
  \end{pmatrix}.
  \end{align*}
\end{comment}

\begin{lemma}[Properties of $\mathcal{G}(D, \mathbf{b})$](Proof in Appendix~\ref{sec:proof:lemma:GDB_prop}.)
\label{lemma:GDB_prop}
For any instance $\mathcal{G}(D, \mathbf{b})$, it holds:
\begin{enumerate}[label=(\roman*).]
\item The graph contains $4D+2$ nodes and $2D(D+1)$ edges.
\item The out-degree of each node in $Y$ or $Y'$ is $D$. The in-degree of each node in $X$ or $X'$ is $D$.
\item $\pi(s , t) = \alpha \left(\sum_{i,j=1}^{D}b_{ij}\right)\frac{(1-\alpha)^3}{D^2}$.
\end{enumerate}
\end{lemma}
\noindent We now define $\Sigma(D)$, a probability distribution over the graph family $\mathcal{G}(D,\mathbf{b})$, which serves as our hard distribution for the lower bound. For each fixed $D$, every entry $b_{ij}$ of the binary matrix $\mathbf{b}$ is chosen independently from a Bernoulli distribution $\operatorname{Bern}(1/2)$. Under this distribution, the total number of ones $S := \sum_{i,j=1}^{D} b_{ij}$
follows a binomial law $ \operatorname{Binomial}(D^{2}, 1/2)$. For node and edge indexing, we adopt the natural scheme:
the $i$-th outgoing edge of $s$ connects to $y_i$,
the $j$-th incoming edge of $t$ originates from $x_j$,
and for each $y_i$, its $j$-th outgoing edge goes to $x_j$ or $x_j'$ depending on the value of $b_{ij}$. Combining all the above preparations, we now proceed to prove the lower bound.  
As in the relative-error case, the proof naturally decomposes into two key lemmas:  
(i) establishing the existence of an appropriate hard distribution and its parameters, and  
(ii) analyzing the conditional distribution of outcomes under a bounded number of oracle queries. 

\begin{lemma}
(Proof in Appendix~\ref{sec:proof:lemma:SSPPR-A1}.)
\label{lemma:SSPPR-A1}
For any constants $\alpha \in (0,1)$, $p \in (0,1)$, and \(m_0 \in \Omega(n_0), m_0 \leq \binom{n_0}{2}, \varepsilon_0 \in (0,1))\) as function of \(n_0\) and assuming \(\varepsilon_0\) is sufficiently small when \(n_0\) is sufficiently large, then there exists constants $\gamma > 0$ such that for any sufficiently small absolute error tolerance $\varepsilon > 0$, we can choose an integer $D(n_0)$  satisfying:
\begin{enumerate}[label=(\roman*).]
    \item The graph $\mathcal{G}(D, \mathbf{b})$ has no more than $n_0$ nodes and no more then $m_0$ edges.
    \item If two matrices $\mathbf{b}$ and $\mathbf{b}'$ have sums $S = \sum b_{ij}$ and $S' = \sum b'_{ij}$ such that $|S - S'| \geq \gamma D$, then the corresponding PPR values satisfy $|\pi(s,t) - \pi'(s,t)| > 2\varepsilon$. 
    \item \( D^2 = \Omega(\min(1/\varepsilon_0^2, m_0))\). 
\end{enumerate}
\end{lemma}

\noindent Then we develop the next lemma to formalizes the statistical difficulty of estimating the sum $S$ with the precision required for our overall lower bound.

\begin{lemma}
(Proof in Appendix~\ref{sec:proof:lemma:statistical_hardness}.)
\label{lemma:statistical_hardness}
For any constant $p \in (0,1)$, there exists a constant $\gamma' > 0$ such that for a sufficiently large $D$, any algorithm making at most $T \leq D^2/2$ queries to a graph drawn from $\Sigma(D)$ will fail to estimate $S = \sum_{i,j=1}^{D^2} b_{ij}$ with an error less than $\gamma' D$. Specifically, if $\hat{S}$ is the algorithm's estimate, then it holds: $$\ \pr_{\mathbf{b} \sim \Sigma(D)}[|S - \hat{S}| > \gamma' D] > p.$$
\end{lemma}

\noindent Combining the constructed \(D\) in Lemma \ref{lemma:SSPPR-A1}, we conclude that using only \(o(\min(1/\varepsilon_0^2, m_0))\) queries will cause larger failure probability than \(p\) in evaluating \(\pi(s, t)\) with absolute error guarantee. \qed

%% file: floats/fig_gb_ur_1_instance.tex
\begin{figure}[!t]
    \centering
    \subcaptionbox{Distribution $\mathcal{G}(D,\mathbf{b})$\label{fig:gb_instance}}%
    [0.45\linewidth]{\includegraphics[height=2.5in]{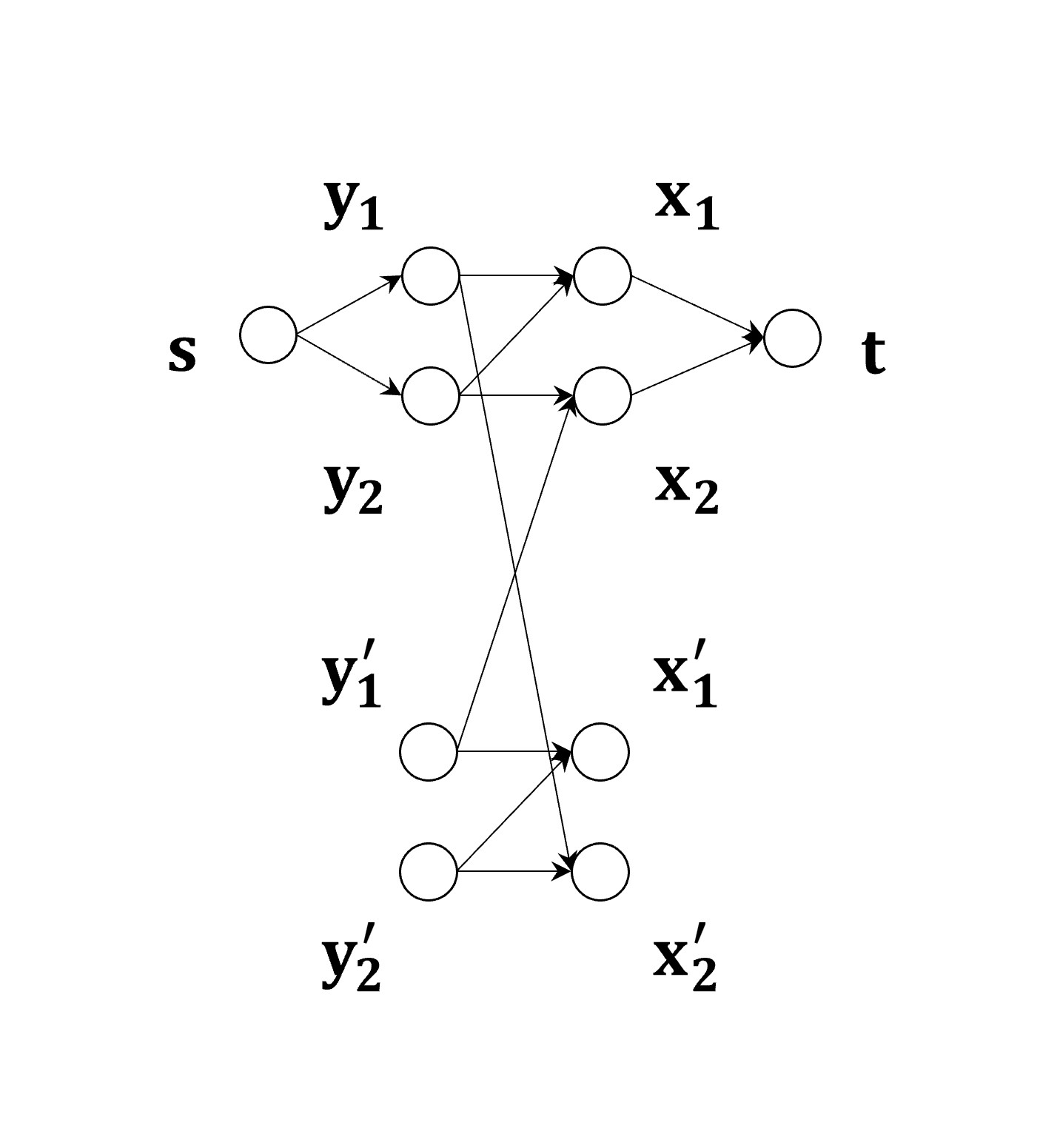}}
    \hfil
    \subcaptionbox{Instance $\mathcal{U}^{-1}(r)$\label{fig_Ur_1_instance.jpg}}%
    [0.45\linewidth]{\includegraphics[height=2.5in]{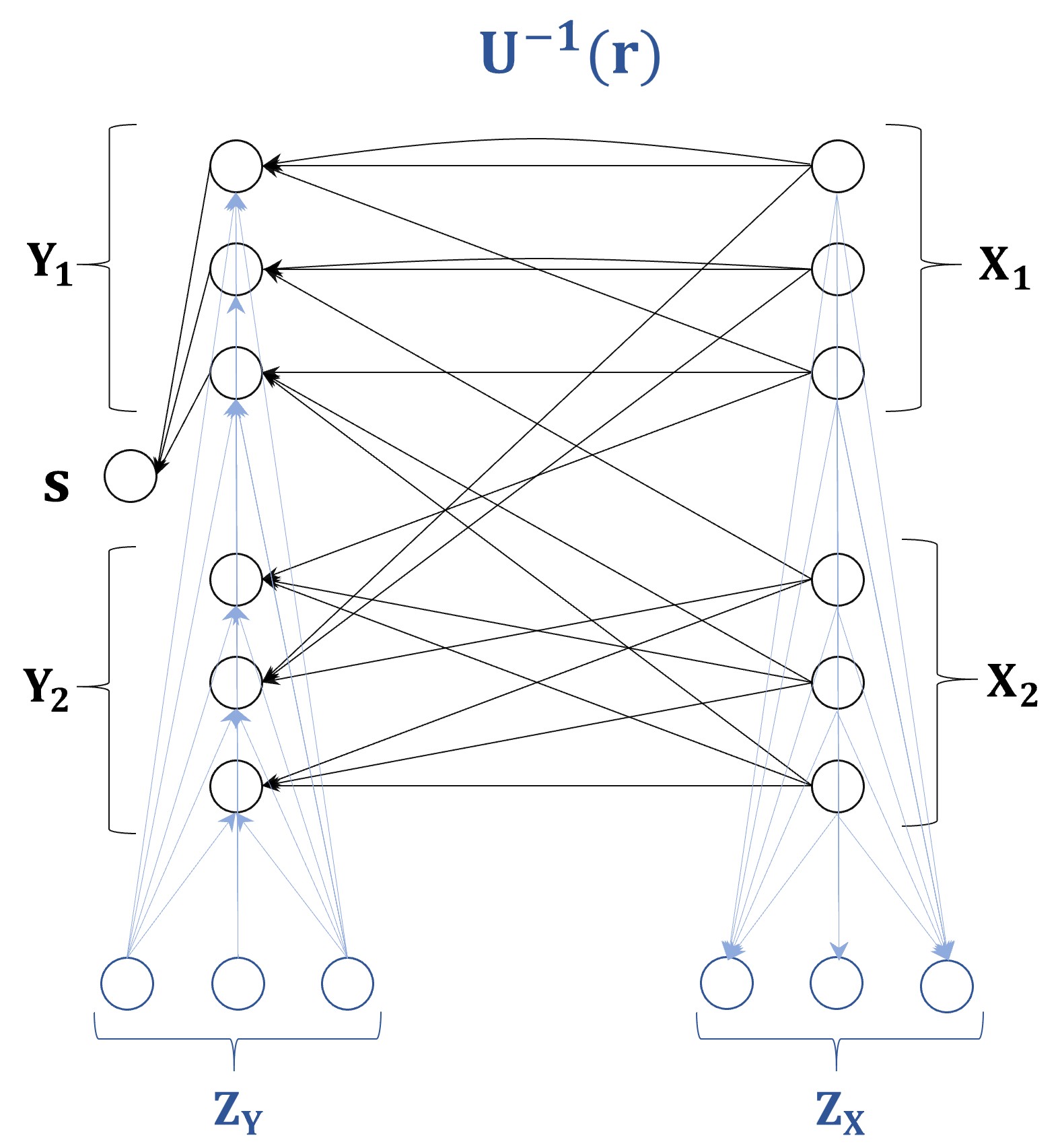}}
    %\hfil
    %\subcaptionbox{Auxiliary instance $\Tilde{U}$\label{fig:UZ_instance}}%
    %[0.3\linewidth]{\includegraphics[height=2.in]{figs/fig_UZ_instance.jpg}}
\vspace{-1ex}
\caption{Graph examples for lower bounds. (a) denotes the example of generating graph distribution $\mathcal{G}(D,\mathbf{b})$ in SSPPR-A query and (b) represents the graph instance $\mathcal{U}^{-1}(r)$ derived from $\mathcal{U}(r)$ in STPPR query.}
\end{figure}

%% file: text/stppr_lower.tex
\section{Lower Bound of STPPR Query} 
In this section, we present the proof of Theorem~\ref{thm:stppr}. The overall argument follows the same structure as the proof of Theorem~\ref{thm:ssppr-r}.  
The key difference lies in the construction of the hard-instance graph family:  
to create the required discrepancy in single target PPR values, we reverse all edge directions in the base graph~$\mathrm{U}$,  
obtaining its inverse family $\mathrm{U}^{-1}$. % as depicted in Figure~\ref{fig:u-1_instance}.  
The details of this construction and its properties are formalized in Lemma~\ref{lemma:TPPR}:

\begin{lemma}[STPPR]
(Proof in Appendix~\ref{sec:proof:lemma:TPPR}.)
\label{lemma:TPPR}
Let $\mathrm{U}^{-1}$ denote the inverse of a graph $\mathrm{U}$ obtained by reversing the direction of each edge.  
Choose any decay factor $\alpha\in (0,1)$, error parameter $c\leq \tfrac{1}{2}$, and functions $\delta_0(n_0)\in \bigl(0,\tfrac{(1-\alpha)^2}{16}\bigr)$ and $m_0(n_0)\in \Omega(n_0)\cap \bigo(n_0^2)$.  
For sufficiently large $n_0$, there exists graph parameters $n,D,d$ (as functions of $n_0$) such that for $\forall\mathrm{U}\in\mathcal{U}(n,D,d)$, the following properties hold for the inverse $\ \mathrm{U}^{-1}$ that:
\begin{enumerate}[label=(\roman*).]
    \item The node, edge count of $\ \mathrm{U}^{-1}$ is $\bigo(n_0)$, $\bigo(m_0)$.
    \item $n^{1/2}\geq D\geq 2d, d\leq \log n$ and $Dn = \Omega\!\left(\min\!\left(m_0, n_0\log n_0 \cdot \tfrac{1}{\delta_0}\right)\right)$.
    \item For all $x_1, x_2\in X_1,X_2$ satisfying $\pi(x_1,s), \pi(x_2,s) \geq \delta_0$, we have
    $ \pi(x_1,s)\;\geq\; \frac{1}{(1-c)^2}\,\pi(x_2,s).$
\end{enumerate}
\end{lemma}

\noindent As in the proof of Theorem~\ref{thm:ssppr-r}, estimating $\pi(x_i,s)$ necessarily requires identifying the correct partition of $X$ into $X_1$ and $X_2$.  
In the arc-centric query model, a directed graph and its edge-reversed version are indistinguishable; hence we may directly apply Lemma~\ref{lemma:targetProb} to the inverse family $\mathcal{U}^{-1}(n,D,d)$.  
It follows that any algorithm for STPPR must perform enough queries to recover the hidden partition, implying a lower bound of  
{\footnotesize $\Omega\!\left(\min\!\left(m_0,\; n_0 \log n_0 \cdot \tfrac{1}{\delta_0}\right)\right)$}, establishing Theorem~\ref{thm:stppr}. \qed

%% file: text/conclusion.tex
\section{Conclusion}
In this work, we significantly narrow or even close the gaps between the upper and lower bounds for the computational complexity of SSPPR queries. On the upper bound side, we tighten the complexities of SSPPR-A and SSPPR-R to {\scriptsize$\bigo(1/\epsilon^2)$} and {\scriptsize$\bigo\!\left(\min\left({\log(1/\delta)}/{\delta},\,m+n\log n\,\log(\tfrac{\log n}{m\delta})\right)\right)$}, respectively. These results improve the best known bounds by factors of {\scriptsize$\log(1/\epsilon)$} and {\scriptsize$\log(\tfrac{\log n}{m\delta})$}. On the lower bound side, under the standard arc-centric model, we prove stronger bounds of {\scriptsize$\Omega(\min(m,1/\epsilon^2))$} for SSPPR-A and {\scriptsize$\Omega(\min(m,{\log(1/\delta)}/{\delta}))$} for SSPPR-R. These improve previous results by factors of $m/n$ or $1/\epsilon$, and $m/n$ or $\log(1/\delta)$, respectively. Together, these results substantially strengthen the theoretical foundations of PPR computation. For SSPPR-R, the upper and lower bounds coincide for graphs with $m\in\Omega(n\log^2 n)$ and thresholds satisfying $1/\delta\in\bigo(\mathrm{poly}(n))$, establishing theoretical optimality under most graph regimes. For SSPPR-A, partial optimality is obtained for coarse absolute error requirements (that is, larger $\epsilon$), where the upper bound reduces to {\scriptsize$\bigo(1/\epsilon^2)$}, matching the lower bound. To the best of our knowledge, this is the first work establishing optimal algorithmic results for SSPPR queries, establishing a new theoretical foundation for the study of PPR computation. Meanwhile, we also generalize our techniques to the \textit{Single-Target Personalized PageRank} (STPPR) query, improving its lower bound from {\scriptsize$\Omega(\min(n,1/\delta))$} to {\scriptsize$\Omega(\min(m,\tfrac{n}{\delta}\cdot\log n))$} under standard arc-centric model. This new result also matches the upper bound established in~\cite{rbs}, revealing the \textit{optimality} and underscoring its theoretical versatility.

%% file: text/ack.tex
\begin{acks}
%\haoyu{tbd.}
This research / project is supported by the National Research Foundation, Singapore, under its AI Singapore Programme (AISG Award No: AISG3-RP-2024-034), and under its Frontier CRP Grant (NRF-F-CRP-2024-0005). This research is also supported by NTU SUG-NAP, and supported by the Ministry of Education, Singapore, under an AcRF Tier 2 Grant (MOE-000761-01). Any opinions, findings, and conclusions or recommendations expressed in this material are those of the author(s) and do not reflect the views of the National Research Foundation, Singapore. 
\end{acks}

%% file: text/appendix.tex
\appendix
\newpage
%\section{Appendix}
%\label{sec:appendix}

%\section{Table of Notations}
%The following table summarizes the frequently used notations in this paper, organized by sections.
\input{floats/tab_notations}

\section{Proof of Lemma~\ref{lemma:discover}}
\label{sec:proof:lemma:discover}
\begin{fact}[Chernoff bound]
\label{fact:chernoff}
Let \(X_1,\ldots,X_n\) be independent random variables taking values in \([0,1]\). Let \(X := \sum_{i=1}^n X_i\) and \(\mu := \mathbb{E}[X]\). Then, for any \(c>0\), \begin{enumerate*} \item if \(c\in[0,1]\), then \(\pr[X > (1+c)\mu] \le \exp\!\big(-\tfrac{c^2}{3}\mu\big)\); \item if \(c\in[0,1]\), then \(\pr[X < (1-c)\mu] \le \exp\!\big(-\tfrac{c^2}{2}\mu\big)\); \item if \(c\ge 1\), then \(\pr[X > (1+c)\mu] \le \exp\!\big(-\tfrac{c}{3}\mu\big)\).
\end{enumerate*}
\end{fact}
\vspace{-1em}
\begin{proof}
Let \(L := \ln \tfrac{1}{\delta} + \ln \tfrac{2}{p}\), and set
\(T_1 = L \lceil 1/\delta \rceil\), \(T_2 = 20L \lceil 1/\delta \rceil\), and \(\theta = 10L\).
We allocate failure budgets of \(p_f/2\), \(p_f/4\), and \(p_f/4\) for three independent failure events: 
\begin{enumerate}[label=(\roman*)]
\item No false negatives in discovery round.
After \(T_1\) independent walks, the miss probability
$$
\pr[u \notin V^{\rm raw}_{\text{out}}] = (1 - \pi(s,u))^{T_1}
\le (1 - \delta)^{T_1}
\le e^{-\delta T_1}
\le e^{-L}
= \delta \tfrac{p_f}{2}.
$$
By union over node set \(|V_{\ge\delta}|\), we have $\pr[V_{\ge\delta}\not\subseteq V^{\rm raw}_{\text{out}}] 
\le \tfrac{1}{\delta} \cdot \delta \tfrac{p_f}{2} = \tfrac{p_f}{2};$ \item No false negatives in estimation round. For \(u \in V_{\ge \delta}\), the counter \(H[u]\) after \(T_2\) walks satisfies
$$\mu_u := \mathbb{E}[H[u]] = T_2 \pi(s,u) \ge 20L.$$ With threshold at \(\theta = 10L = \mu_u/2\), we have $$
\pr[H[u] < \theta] \le \exp(-\mu_u/8)
\le e^{-20L/8}
= e^{-2.5L}
= \left(\delta \tfrac{p}{2}\right)^{2.5}
$$ by Chernoff bound. Again by union over node set \(V_{\ge \delta}\), we have that $$\pr[\exists u \in V_{\ge \delta} : H[u] < \theta]
\le \tfrac{1}{\delta} \left(\delta \tfrac{p_f}{2}\right)^{2.5}
= \delta^{1.5} \left(\tfrac{p_f}{2}\right)^{2.5}
< \tfrac{p_f}{4}$$ for \(\delta \le 1/2\) and any sufficiently small probability \(p_f\); 
\item No heavy false positives in estimation round. For any node \(u\) with \(\pi(s,u) < \delta/(5e)\), it holds $$\mu_u = T_2 \pi(s,u)
< 20L \Big\lceil \tfrac{1}{\delta} \Big\rceil \tfrac{\delta}{5e}
= \tfrac{4L}{e} \delta \Big\lceil \tfrac{1}{\delta} \Big\rceil
< \tfrac{4L}{e}(1+\delta)
\le \tfrac{6L}{e} (\delta \le \tfrac{1}{2}).$$
With \(\theta = 10L\), we then have $$
\pr[H[u] \ge \theta] \le \exp\!\left(-\tfrac{\theta - \mu_u}{3}\right)
\le \exp\!\left(-\tfrac{L(10 - 6/e)}{3}\right)
< e^{-2.5L}
= \left(\delta \tfrac{p}{2}\right)^{2.5}.
$$ Since discovery round generates at most \(T_1\) raw candidates, another union bound gives $$\pr[\exists u \notin V_{\ge \delta/(5e)} : u \in V_{\mathrm{out}}]
\le T_1 \cdot e^{-2.5L}
\le L\!\left(\tfrac{1}{\delta}+1\right)\! \left(\delta \tfrac{p_f}{2}\right)^{2.5}
< \tfrac{p_f}{4}$$
for \(\delta \le 1/2\) and sufficiently small probability \(p\). 
\end{enumerate}
All together yields total failure probability bounded by \(p_f\), which completes the proof.
\end{proof}

\section{SSPPR-R Upper Bound Proof}
\mypara{Model Setting.} 
By default, we adopt the following settings in handling sampling complexity. 

We adopt the standard word-RAM model in which each machine word contains \(\Theta(\log n)\) bits. Basic arithmetic operations on \(\bigo(1)\)-word numbers—such as addition, subtraction, multiplication or division by a constant, and comparisons—take \(\bigo(1)\) time. We can generate a uniform random integer in \(\{1, \dots, N\}\) in \(\bigo(1)\) time for \(n = \bigo(N)\). Probabilities are assumed to be representable with \(\bigo(\log n)\)-bit precision (e.g., all positive values are at least \(n^{-\bigo(1)}\)), which is standard in prior work. We further assume \(1/\delta = \bigo(\mathrm{poly}(n))\), implying that any number in the interval \((\delta, 1)\) can be stored in \(\bigo(1)\) words, common and well-aligned with existing PPR settings.

\subsection{Proof of Lemma~\ref{lemma:newdecom}}
\label{sec:proof:lemma:newdecom}
\begin{proof}
The decomposition in Equation~(\ref{Eq:decom1}) states $\pi_G(s,t) = \sum_{u\in \boundary{X}}\pi_G(s,u)\sum_{{\footnotesize\overleftrightarrow{p_2}} \in \overleftrightarrow{p_G}(u, t)} w_G(p_2)$. Hence, by comparing this equation with the one to be proved in Lemma~\ref{lemma:newdecom}, we only need to show 
\begin{align}\label{eq:weightdecomp}
\begin{matrix}
\sum_{{\footnotesize\overleftrightarrow{p_2}} \in \overleftrightarrow{p_G}(u, t)} w_G({\footnotesize\overleftrightarrow{p_2}})={d_{\mathrm{out}}^{G_X}(u)\pi_{G_X}(u,t)}/{ d_{\mathrm{out}}^{G}(u)}    
\end{matrix}
\end{align}
To prove this equation, we first decompose $\pi_{G_X}(u,t)$ as $\alpha \sum_{{\footnotesize\overleftrightarrow{p}} \in \overleftrightarrow{p_{G_X}}(u,t)} w_{G_X}({\footnotesize\overleftrightarrow{p}})$ by definition of expressing PPR in terms of the sum of path weights. Hence, the right hand side of Equation~\ref{eq:weightdecomp} Becomes $(d_{\mathrm{out}}^{G_X}(u)/{d_{\mathrm{out}}^{G}(u)})\sum_{{\footnotesize\overleftrightarrow{p}} \in \overleftrightarrow{p_{G_X}}(u,t)} w_{G_X}({\footnotesize\overleftrightarrow{p}})$. 
Since in $G_X$ there are no edges leading back to $\boundary{X}$, the path set $\overleftrightarrow{p_{G_X}}(u,t)$ is identical to $\overleftrightarrow{p_G}(u, t)$ by construction. Moreover, all out-degrees within $X$ remain unchanged between $G$ and $G_X$. The only discrepancy arises from the first edge of each path, which is exactly balanced by the factor ${d_{\mathrm{out}}^{G_X}(u)}/{ d_{\mathrm{out}}^{G}(u)}$. Therefore, Equation~(\ref{eq:decomp-pi}) holds. 

We further proceed to establish another inequality yet to be used:
\begin{align}
\begin{matrix}
\sum_{x\in X} \pi_G(s, x)
 &\leq (1-\alpha)\sum_{u \in \boundary{X}}  \frac{1}{\alpha}{d_{out}^{G_X}(u)\pi_G(s,u)}/{d_{out}^{G}(u)} \leq \frac{1 }{\alpha} \sum_{x\in X} \pi_G(s, x).
\end{matrix}\label{eq:decom_bound}
\end{align}
By construction of $G_X$, all edges from $u$ ends in $X$. Furthermore, since all edges from $X$ to $G\backslash X$ are redirected to the padding node $v_X$, the subset $X$ together with $v_X$ is closed with respect to out-edges. Then, the potential endpoints of a random walk 
originating from $u \in \boundary{X}$ lie within $X \cup \{u, v_X\}$, and there is no loop containing $u$. Consequently, we have $\pi_{G_X}(u,u) = \alpha$, $\sum_{t\in X \cup \{v_X\}}\pi_{G_X}(u,t)=1-\alpha$,
\begin{align*}
 \begin{matrix}
 \sum_{t\in X\cup \{v_X\}} \pi_G(s,t)    
 \end{matrix}
&=  \begin{matrix}\sum_{t\in X\cup \{v_X\}}\left (\sum_{u \in \boundary{X}}{\frac{1}{\alpha} d_{out}^{G_X}(u)\pi_G(s,u)\pi_{G_X}(u, t)}/{d_{out}^{G}(u)}    
 \right )\end{matrix}\\
&= \begin{matrix}\frac{1}{\alpha}\sum_{u \in \boundary{X}}{d_{out}^{G_X}(u)\pi_G(s,u)}/{d_{out}^{G}(u)}\sum_{t\in X\cup \{v_X\}}\pi_{G_X}(u, t).    
 \end{matrix}\\
\end{align*}
Combined with our result $\sum_{t\in X \cup \{v_X\}}\pi_{G_X}(u,t)=1-\alpha$, we then bound $\sum_{x\in X} \pi_G(s, x)$ with
\begin{align*}
   &\quad\ \begin{matrix}\sum_{x\in X} \pi_G(s, x) \end{matrix}\\
   &\leq \begin{matrix}
       \sum_{t\in X\cup \{v_X\}} \pi_G(s,t)
   \end{matrix} \\
   &= \begin{matrix}
       \frac{1}{\alpha}\sum_{u \in \boundary{X}}  {d_{out}^{G_X}(u)\pi_G(s,u)}/{d_{out}^{G}(u)}\sum_{t\in X\cup \{v_X\}}\pi_{G_X}(u, t)
   \end{matrix} \\
   &= \frac{1}{\alpha}\begin{matrix}
       \sum_{u \in \boundary{X}}{(1-\alpha)d_{out}^{G_X}(u)\pi_G(s,u)}/{d_{out}^{G}(u)}.
   \end{matrix}
\end{align*}
Hence the left side of Inequality~(\ref{eq:decom_bound}) holds. Next we establish the right side of the inequality. We will first estimate $\sum_{x\in X}\pi_{G_X}(u,x)$ for $u \in \boundary{X}$ from a lower bound perspective. Specifically, we will first prove that $\sum_{x\in X}\pi_{G_X}(u,x)\ge (1-\alpha)\alpha$. Observe that in $G_X$, every outgoing edge from $u \in \boundary{X}$ leads to a node in $X$. An $\alpha$-decay random walk starting at $u$ transitions to a neighbor $v \in X$ with probability $(1-\alpha)$. Upon reaching $v$, the walk terminates with probability $\alpha$. Thus, the probability that the walk terminates exactly at the second step (which lies in $X$) is $\alpha(1-\alpha)$. This yields the lower bound $\sum_{x\in X}\pi_{G_X}(u,x)\ge (1-\alpha)\alpha$. Be Equation~(\ref{eq:decomp-pi}), we can then express the term $\sum_{x\in X}\pi_G(s,x)$ as
$$
\begin{matrix}
\sum_{x\in X}\pi_G(s,x) = \frac{1}{\alpha}\sum_{u \in \boundary{X}}\frac{d_{\mathrm{out}}^{G_X}(u)\pi_G(s,u)}{d_{\mathrm{out}}^{G}(u)}\sum_{x\in X}\pi_{G_X}(u, x).
\end{matrix}
$$
Substituting $\sum_{x\in X}\pi_{G_X}(u,x) \ge \alpha(1-\alpha)$ into this expression, it yields
$$
\begin{matrix}
\sum_{x\in X}\pi_G(s,x) \geq (1-\alpha)\sum_{u \in \boundary{X}}\frac{d_{\mathrm{out}}^{G_X}(u)\pi_G(s,u)}{d_{\mathrm{out}}^{G}(u)},
\end{matrix}
$$
which completes the proof.

\end{proof}

\subsection{Proof of Lemma~\ref{lemma:one_round}}
\label{sec:proof:lemma:one_round}

\begin{proof}
Recall that in \algorr{}, our estimator consists of two parts: an initial accurate estimation over \boundary{X}: $\{\hat{\pi}_G(s,u)\mid u\in \boundary{X}\}$ and source distributed $\alpha-$RW probability based on it: 
\begin{align}
\label{eq:app:5}
\begin{matrix}
P_r(t)=\sum_{u\in\boundary{X}}\mathbf{\hat{\bar{S}}_r}(u)\cdot\pi_{G_X}(u,t)
\end{matrix}
\end{align}
 over graph $G_X$ targeting $\forall t\in X$, where $\mathbf{\hat{\bar{S}}_r}$ is normalized probability over value $$\begin{matrix}
\mathbf{\hat{S}_r}(u):=\frac{1}{\alpha}{\hat{\pi}_G(s,u)d_{\mathrm{out}}^{G_X}(u)}/{d_{\mathrm{out}}^{G}(u)}.
 \end{matrix}$$ This probability is later approximated by simulating sourced $\alpha-$RWs. First, since we invoke \algoa{} to fetch $X_{out}$ in $X$, we have that with probability at least \(1-p_r\), \(X_{out}\) contains all the nodes $t$ satisfying 
$$P_r(t) \geq \alpha (1-\alpha)\delta_r/4 $$
and each node $t$ in \(X_{out}\) satisfies 
$$P_r(t) \geq \alpha (1-\alpha)\delta_r/(20e)$$
by Lemma~\ref{lemma:discover} (guarantee of \algoa{}). As the final estimator in each round is multiplied with factor $\texttt{SUM}(\mathbf{\hat{S}_r})$ (Line 12 in Algorithm~\ref{algo:dist_walk}), result $(i)$ naturally holds. Meanwhile, by simulating $N_r={256\alpha\log(n/p_r)}/\left({c_r^2\delta_r}\right)$ $\alpha-$RWs to estimate $P_r(t)$ as $\hat{P}_r(t)$ (Lines 7-11 in Algorithm~\ref{algo:dist_walk}), it holds with probability $\geq 1-p_r/n$ that for $\forall t \in X$: 
\begin{align}
\label{eq:app:6}
\hat{P}_r(t)/P(t) \in [1-c_r, 1+c_r]
\end{align}
by Chernoff bound. Denote $\forall u\in\boundary{X},\mathbf{{S}_r}(u):=\frac{1}{\alpha}{{\pi}_G(s,u)d_{\mathrm{out}}^{G_X}(u)}/{d_{\mathrm{out}}^{G}(u)}$, and recall that $\texttt{SUM}(\mathbf{{S}_r})$ equals to the summation over all values in vector $\mathbf{{S}_r}$, it then holds that $\text{for }\forall t\in X,$
\begin{align}
    \frac{\hat\pi_G(s, t)}{\pi_G(s, t)}&=\frac{\hat{P}_r(t)\cdot\texttt{SUM}(\mathbf{\hat{S}_r})}{\pi_G(s, t)}\\
    &=\frac{\hat{P}_r(t)\cdot\texttt{SUM}(\mathbf{\hat{S}_r})}{P_r(t)\cdot\texttt{SUM}(\mathbf{\hat{S}_r})} \cdot \frac{P_r(t)\cdot\texttt{SUM}(\mathbf{\hat{S}_r})}{\pi_G(s, t)}\\
    &=\frac{\hat{P}_r(t)}{P_r(t)}\cdot\frac{P_r(t)\cdot\texttt{SUM}(\mathbf{\hat{S}_r})}{\pi_G(s, t)}\\
\text{{(By Equation~(\ref{eq:app:5}))}}&=\frac{\hat{P}_r(t)}{P_r(t)}\cdot\frac{\sum_{u\in\boundary{X}}\mathbf{\hat{S}_r}(u)\cdot\pi_{G_X}(u,t)}{\pi_G(s, t)}\\
    \text{{(By Equation~(\ref{eq:decomp-pi}))}}&=\frac{\hat{P}_r(t)}{P_r(t)} \cdot \frac{\sum_{u\in\boundary{X}}\mathbf{\hat{S}_r}(u)\cdot\pi_{G_X}(u,t)}{\sum_{u\in\boundary{X}}\mathbf{S_r}(u)\cdot\pi_{G_X}(u,t)}.\label{eq:app:8}
\end{align}
By requirement, for $\forall u\in \boundary{X}$, $\hat{\pi}_G(s, u)/\pi_G(s, u)$ is bounded within $[1-c, 1+c]$. Thus we have
\begin{align}
\label{eq:app:9}
    \frac{\mathbf{\hat{S}_r}(u)}{\mathbf{S_r}(u)}=\frac{\frac{1}{\alpha}{{\hat{\pi}}_G(s,u)d_{\mathrm{out}}^{G_X}(u)}/{d_{\mathrm{out}}^{G}(u)}}{\frac{1}{\alpha}{{\pi}_G(s,u)d_{\mathrm{out}}^{G_X}(u)}/{d_{\mathrm{out}}^{G}(u)}}=\frac{{\hat{\pi}}_G^0s,u)}{{\hat{\pi}}_G(s,u)} \in [1-c, 1+c].
\end{align}
Combining Equations~\ref{eq:app:5},~\ref{eq:app:8} and~\ref{eq:app:9}, the result $(ii)$ then follows. Thus, we complete the proof.
\end{proof}

\subsection{Proof of Lemma~\ref{lemma:all_rounds}}
\label{sec:proof:lemma:all_rounds}

\begin{lemmanrestate}[\algor{} Guarantee]
Algorithm~\ref{algo:dist_walks} correctly solves the SSPPR-R query: with probability at least $1-p_f$, it returns $\hat{\pi}(s,u)$ satisfying $|\hat{\pi}(s,u)-\pi(s,u)| \le c\cdot \pi(s,u)$ for $\forall u \in V$ with $\pi(s,u)\ge \delta$.
\end{lemmanrestate}
\begin{proof}
We first show that our final estimator of \algor{} fits the relative error requirement of SSPPR-R. In each round $r$, the algorithm succeeds with probability at least $1-p_r$ as guaranteed by \algorr{}, and the relative error is incurred by at most $c_r$ factor. As such, the total relative error across all $R$ rounds accumulates at most by $\prod_{r=0}^R (1-c_r)$ and $\prod_{r=0}^R (1+c_r)$, respectively. Recall that we set $c_r\leq\frac{c}{2R}$ in each round $r$, we can then bound the factor with: 
\begin{align}
&\quad \begin{matrix}
    \prod_{r=0}^R (1-c_r)
\end{matrix} \\
 (\text{By } x\geq\ln(1+x))&\geq \begin{matrix}
    1 + \sum_{r=0}^R \ln(1-c_r)
 \end{matrix} \\
 (\text{By }\ln(1-x) \geq -2x)&\geq \begin{matrix}
      1 - 2\sum_{r=0}^R c_r \geq\; 1-c.
 \end{matrix}
\end{align}
Similarly, we have $\prod_{r=0}^R (1+c_r)\leq1+c$. Meanwhile,  the total failure probability over all rounds is bounded as $\sum_{r=1}^{R} p_r = p_f/2 < p_f$ with setting $p_r=\frac{p_f}{2n}$ in each round by union bound. Thus, any node $u\in V$ with PPR value $\pi_G(s,u)\geq\delta$ are detected and estimated to satisfy $$|\pi_G(s,u)-\hat{\pi}_G(s,u)|\leq c\cdot\pi(s,u)$$ with probability at least $1-p_f$. Next, we show that $\algor$ can output all nodes surpassing the desired threshold $\delta$. Recall that by Lemma~\ref{lemma:one_round}, the last round $R$ has estimated all nodes with PPR values larger than $\delta_R\cdot\texttt{SUM}(\mathbf{\hat{S}_R})$. Thus, we then proceed to show that $\texttt{SUM}(\mathbf{\hat{S}_R})\leq\frac{\delta}{\delta_R}$. 
By the inequality in Lemma~\ref{lemma:newdecom}, we can first conclude that with relative error guarantee satisfied, it holds
\begin{align}
\texttt{SUM}(\mathbf{\hat{S}_R})&=\begin{matrix}
\sum_{u \in \boundary{X_R}}  \frac{1}{\alpha}{d_{out}^{G_X}(u)\hat{\pi}_G(s,u)}/{d_{out}^{G}(u)}
\end{matrix}\\
\text{ (By relative guarantee)}&\leq \begin{matrix}
\sum_{u \in \boundary{X_R}} \frac{1+c}{\alpha}{d_{out}^{G_X}(u){\pi}_G(s,u)}/{d_{out}^{G}(u)}
\end{matrix}\\
\text{ (By Inequality~(\ref{eq:decom_bound}))}&\leq\begin{matrix}
    \sum_{t\in X_R}\frac{1+c}{\alpha} \pi_G(s,t).
\end{matrix} \label{eq:app:12}
\end{align}
Then, we only need to show that \begin{align}
\label{eq:app:13}
\begin{matrix}
\sum_{t\in X_R}\pi_G(s,t)\leq\frac{\alpha}{1+c}\cdot\frac{\delta}{\delta_R}.
\end{matrix}
\end{align}
We will prove this by an induction on $\sum_{t\in X_r}\pi_G(s,t)\leq\frac{1}{\Delta^{r-1}}$ for each round $r\in\{1,2,...,R\}$. When $r=1$, $\sum_{t\in X_r}\pi_G(s,t)\leq 1=\frac{1}{\Delta^0}$ naturally holds. Suppose it holds for $r-1$ round with $\sum_{t\in X_{r-1}}\pi_G(s,t)\leq\frac{1}{\Delta^{r-2}}$. Then in round $r$, as all nodes with PPR value larger than $\delta_r\cdot\texttt{SUM}(\mathbf{\hat{S}_r})$ have been estimated, the remaining nodes has PPR maximum values of 
\begin{align}
&\quad\delta_r\cdot\texttt{SUM}(\mathbf{\hat{S}_r})
\\
\text{ (by Inequality (\ref{eq:app:12}))}&\begin{matrix}
\leq \delta_r\cdot\sum_{t\in X_r}\frac{1+c}{\alpha} \pi_G(s,t)
\end{matrix}\\
& \begin{matrix}
\leq \delta_r \cdot \frac{1+c}{\alpha}\cdot\frac{1}{\Delta^{r-1}}
\end{matrix} \\
& = \begin{matrix}
\frac{\alpha}{(1+c)n\Delta} \cdot \frac{1+c}{\alpha}\cdot\frac{1}{\Delta^{r-1}}
\end{matrix} \\
& = \begin{matrix}
\frac{1}{n}\cdot\frac{1}{\Delta^r}.
\end{matrix}
\end{align}
Therefore, the sum of all these PPR scores is less than $\frac{1}{\Delta^r}$ and we complete the induction. When round $r=R=\log_\Delta(\frac{2\delta_0}{\alpha\delta})=\log_\Delta(\frac{2\delta_R}{\alpha\delta})\geq  \log_\Delta(\frac{\alpha}{1+c}\frac{\delta_R}{\delta})$, Equation~(\ref{eq:app:13}) directly holds and we complete the proof.
\end{proof}

\section{Proof Detail in SSPPR-R Lower Bound}
\subsection{Proof of Lemma \ref{lemma:PPR_R_const}} \label{sec:proof:PPR_R_const}

\begin{lemmanrestate}[Existence of Hard-Instance Graph Family]
Choose any constant \(c\in (0,\frac{1}{2}]\) and any functions \(\delta_0(n_0)\in (0,1), m_0(n_0)\in\Omega(n_0)\cap \bigo(n^{2})\). For sufficiently large $n_0$, there exist parameters \(n, D, d, r\)  (as function of \(n_0\)), such that for $\forall \mathrm{U}$ in $\mathcal{U}(n,D,d)$, its $\mathrm{U}(r)$ holds:%the following conditions hold for its $r-$padded instance: 
\begin{enumerate}[label=(\roman*).]
\item  The node, edge count of \(\ \mathrm{U}(r)\) is $\bigo(n_0)$ and $\bigo(m_0)$;
\item  The edge count is in $\Theta(n(r+D)) =\Theta\left(\min(m_0, {\log(\frac{1}{\delta_0})}/{\delta_0})\right)$, and the node count is in $\Theta(r+n) = \bigo(n)$. Additionally, \(D\geq \frac{3}{2}d\) and \(d\leq \log(n)\); 
\item For any $x_1,x_2\in X_1,X_2$, the PPR scores $\pi(s,x_1), \pi(s,x_2)$ are fixed values, irrelevant to the choice of nodes or the graph instance. It holds that $\pi(s,x_1), \pi(s,x_2)\geq \delta_0, \pi(s,x_1) > \frac{1}{(1-c)^2}\cdot \pi(s,x_2).$ 
\end{enumerate}
\end{lemmanrestate}

Before the proof of this lemma, let us calculate the basic property of a $r-$padded instance $\mathrm{U}(r)$, which is formalized as the following lemma:

\begin{lemma}
\label{lemma:PPR_Cal}
For any instance \(\mathrm{U} \in \mathcal{U}(n, D, d)\), the following properties hold:
\begin{enumerate}[label=(\roman*).]
    \item The graph \(\mathrm{U}\) contains \(4n + 1\) nodes and \(2n(2D - d) + n\) edges;
    \item For each node \(x_1 \in X_1\) and \(x_2 \in X_2\), the Personalized PageRank values with respect to the source node \(s\) are
    \(\pi(s,x_1) = (1 - \alpha)^2 \cdot \frac{D}{(2D - d)n}\) and 
    \(\pi(s,x_2) = (1 - \alpha)^2 \cdot \frac{D - d}{(2D - d)n}\).
\end{enumerate}

Moreover, for the \(r\)-padded instance \(\mathrm{U}(r)\) with \(r \geq D\), we have:
\begin{enumerate}[label=(\roman*).]
    \item The graph \(\mathrm{U}(r)\) contains \(4n + 2r + 1\) nodes and \(2n(2D - d + 2r) + n\) edges;
    \item For each node \(x_1 \in X_1\) and \(x_2 \in X_2\), the PPR values are
    \(\pi(s,x_1) = (1 - \alpha)^2 \cdot \frac{D}{(2D - d + r)n}\) and 
    \(\pi(s,x_2) = (1 - \alpha)^2 \cdot \frac{D - d}{(2D - d + r)n}\), 
    and in particular, 
    \(\pi(s,x_2) = \frac{D - d}{D} \cdot \pi(s,x_1)\).
\end{enumerate}
\end{lemma}

\begin{proof}
The number of nodes and edges in both instances $\mathrm{U}$ and $\mathrm{U}(r)$ follows directly from their definitions. We now compute the PPR values from the source node $s$ on the instance $\mathrm{U}$. First, since each node $y \in Y_1$ receives one incoming edge from $s$, and there are $n$ such nodes, it holds that $\pi(s,y) = \frac{1 - \alpha}{n}$. Next, consider nodes $x_1 \in X_1$ and $x_2 \in X_2$. Each node in $X_1$ receives $D$ in-edges from $Y_1$, and each node in $X_2$ receives $D - d$ in-edges from $Y_1$. Since every $y \in Y_1$ has out-degree $2D - d$, the probability of transitioning from a given $y$ to any of its children is $\frac{1}{2D - d}$. By aggregating over all paths from $s$ to $x_1$ via $Y_1$, we obtain
$\pi(s,x_1) = (1 - \alpha)^2 \cdot \frac{D}{(2D - d)n}$ and
$\pi(s,x_2) = (1 - \alpha)^2 \cdot \frac{D - d}{(2D - d)n}$. The difference is then
$\pi(s,x_1) - \pi(s,x_2) = (1 - \alpha)^2 \cdot \frac{d}{(2D - d)n}$. For the padded instance $\mathrm{U}(r)$, the analysis is similar. Each node in $Y$ now has out-degree $2D - d + r$ due to the extra outgoing edges to $Z_Y$. Thus, the transition probabilities are updated accordingly, and we obtain
$\pi(s,x_1) = (1 - \alpha)^2 \cdot \frac{D}{(2D - d + r)n}$ and
$\pi(s,x_2) = (1 - \alpha)^2 \cdot \frac{D - d}{(2D - d + r)n}$. In particular, this yields the ratio $\pi(s,x_2) = \frac{D - d}{D} \cdot \pi(s,x_1)$. This completes the proof.
\end{proof}

Now we can carefully choose the required $n, D, d, r$ for Lemma \ref{lemma:PPR_R_const}.
\begin{proof}
    We consider two cases with an overall setting of $\delta$ such that 
$\frac{\log\frac{1}{\delta}}{\delta} = \min(c m_0, \frac{\log\frac{1}{\delta_0}}{\delta_0}).$
\paragraph{Case 1:} $
n_0 \leq \frac{1}{32} (1-\alpha)^2 \frac{1}{\delta}.$ And we here set $n = n_0, \quad d = \lfloor \log(n_0)\rfloor, \quad D = \lfloor\left(1 + \frac{1}{4c}\right) d \rfloor+ 1$ and $ r =  1+\lfloor \frac{(1-\alpha)^2d}{16 c n_0 \delta}\rfloor > 0.$
Under these parameters, conditions (i), and $d<log(n), D>3/2d$ in condition $(ii)$ naturally hold when $n$ is sufficiently large. Moreover for the evaluation of $n(r+D)$ $(ii)$, we have:
\begin{align}
n (r+D) \in \Omega\left((1-\alpha)^2 d/\delta c\right) = \Omega\left( \min\left(m_0, \frac{\log\frac{1}{\delta_0}}{\delta_0}\right)\right).
\end{align}

For the bound of $r$, we have $r \in \bigo(\log(n_0)/(\delta n_0)) \subset \bigo(cm_0/n_0) \subset \bigo(n_0)$. By Lemma~\ref{lemma:PPR_Cal}, the PPR value for any $x_2 \in X_2$ satisfies
\begin{align}
\pi(s,x_2) = (1-\alpha)^2 \frac{D - d}{(2D - d + r) n_0} 
\geq \frac{(1-\alpha)^2}{(4c + 2) n_0 + \frac{(1-\alpha)^2}{4 \delta}} \geq \delta \geq \delta_0.
\end{align}
\paragraph{Case 2:}  $n_0 > \frac{1}{32} (1-\alpha)^2 \frac{1}{\delta}.$ Under this setting, we set $r = 0$ since 
no padded instance is required here. Further, we give $n = \lfloor \frac{1}{32} (1-\alpha)^2 \frac{1}{\delta}\rfloor$ and $ D = 1+\lfloor\left(1 + \frac{1}{4c}\right) d\rfloor.$ Similarly, for sufficiently large $n$, conditions (i), (ii), and (iii) hold. Additionally, it holds that 
\begin{align}
n D &= n d \left(1 + \frac{1}{4c}\right)
\\ &= \Omega\left(\frac{\log(\frac{1}{\delta})}{\delta} \right)
\\&=\Omega\left( \min\left(m_0, \frac{\log\frac{1}{\delta_0}}{\delta_0}\right)\right),\\
&\text{and} \nonumber\\
\pi(s,x_2) &= \frac{(1-\alpha)^2}{n (4c + 2)} \geq \delta \geq \delta_0.
\end{align}
These ensure that condition $(ii)$ and condition $(iii)$ hold and thus we complete the proof. 
\end{proof}

\subsection{Proof of Lemma \ref{thm:algo_adapt_tot}}\label{sec:proof:thm:algo_adapt_tot}. We begin with the lifting part. To allow for generalized algorithms designed for graphs with parallel edges, we first prove that an instance from our distribution has bounded multiplicity w.h.p:

\begin{lemma}\label{lemma:paraBoundU}
 \(\pr_{\mathrm{U}(r)\sim \Sigma(n, D,d)}[ \mathrm{MUL}(\mathrm{U}(r)) \geq  L] \leq 4n^2 (\frac{4D}{n})^{L}\), for \(\forall L\leq D\).
\end{lemma}

\begin{proof}
Fix an arbitrary realization of \(\left(\mathrm{SP},\, \mathrm{P}^{2D-d}_1,\dots, \mathrm{P}^{2D-d}_{2n},\, \mathrm{P}^{n(2D-d)}_{Y_1},\, \mathrm{P}^{n(2D-d)}_{Y_2}\right)\), i.e., the split of \(X\) and all in-/out-edge permutations chosen in Steps~(i)–(iii) as described in Section~\ref{sec:distribution}. These choices determine the following graph information:
\begin{enumerate}[label=(\roman*).]
    \item The set of labels for out-edges starting from nodes in \(Y_1\) and pointing to nodes in \(X_1\);
    \item For every \(x_i \in X_1\), the set of labels of its in-edges originating from \(Y_1\).
\end{enumerate}

The only remaining randomness that may affect parallel edge multiplicity arises from the uniform bijection \(\mathrm{S}(Y_1, X_1)\). As described in Step~(iv), this bijection maps the set of all out-edge labels from nodes in \(Y_1\) directed toward \(X_1\) to the set of all in-edge labels of nodes in \(X_1\) originating from \(Y_1\), thereby determining the actual edge connections. Fix a pair of nodes \(x_i \in X_1\) and \(y_j \in Y_1\) (an analogous argument applies to the cases \((X_1, Y_2)\), \((X_2, Y_1)\), and \((X_2, Y_2)\)). Let \(\mathrm{LABEL}_{\text{out}}(y_j, X_1)\) denote the set of at most \(2D - d\) out-edge labels of \(y_j\) that are directed toward \(X_1\), and let \(\mathrm{LABEL}_{\text{in}}(Y_1, x_i)\) denote the set of \(D\) in-edge labels of \(x_i\) originating from nodes in \(Y_1\). We also define \(\mathrm{LABEL}_{\text{out}}(Y_1, X_1)\) as the complete set of out-edge labels from all nodes in \(Y_1\) targeting \(X_1\), where each label includes both the source node \(y_{j'}^{(1)}\) and its local out-edge label. Since \(\mathrm{S}(Y_1, X_1)\) is a uniformly random bijection over two sets of size \(nD\), for any given \(L\)-subset of labels in \(\mathrm{LABEL}_{\text{out}}(y_j, X_1)\) (i.e., any collection of \(L\) out-edges from node \(y_j\)), its image under \(\mathrm{S}^{-1}(Y_1, X_1)\) is distributed uniformly over all possible \(L\)-subsets of \(\mathrm{LABEL}_{\text{out}}(Y_1, X_1)\). Therefore, for any particular \(L\)-subset \(F\subset \mathrm{LABEL}_{in}(Y_1, x_i)\), we have that 
\begin{align}
  \pr\bigl[\mathrm{S}^{-1}(Y_1, X_1)(F)\subseteq \mathrm{LABEL}_{out}(y_j, X_1)\bigr]
  \;=\;
  \frac{\binom{D}{L}}{\binom{nD}{L}}.
\end{align}

There are at most \(\binom{2D - d}{L}\) ways to choose such an \(F\). A union bound therefore gives  
$$\pr\bigl[|\{\text{edges between } x_i, y_j \}|\ge L \,\bigm|\, \mathrm{P}^{2D-d}_1,\dots, \mathrm{P}^{2D-d}_{2n},\, \mathrm{P}^{n(2D-d)}_{Y_1},\, \mathrm{P}^{n(2D-d)}_{Y_2}\bigr] \le \binom{2D}{L} \cdot \frac{\binom{D}{L}}{\binom{nD}{L}}.$$
When \(n \ge 4\) (so that \(nD \ge 4D \ge 2L\)), we have the bounds \(\binom{D}{L} \le \frac{D^{L}}{L!}\), \(\binom{2D}{L} \le \frac{(2D)^{L}}{L!}\), and \(\binom{nD}{L} \ge \frac{(nD/2)^{L}}{L!}\). These imply
\(\binom{2D}{L} \cdot \frac{\binom{D}{L}}{\binom{nD}{L}} \le \left(\frac{2D}{n}\right)^L\). The same inequality holds for all other combinations of \((x_i, y_j)\), including \(x_i \in X_2\) or \(y_j \in Y_2\). Since there are only \(2n\) nodes in \(X\) and \(2n\) nodes in \(Y\), we conclude that
\begin{align}
\pr_{\mathrm{U}(r) \sim \Sigma(n, D, d)}\left[\mathrm{MUL}(\mathrm{U}) \geq L\right] \le 4n^2 \left(\frac{4D}{n}\right)^L,
\end{align}
which completes the proof.

\end{proof}

Following Lemma~\ref{lemma:paraBoundU}, the presence of parallel edges in our distribution is inconsequential, provided that we allow the algorithm to fail on high-multiplicity instances. Under this allowance, any algorithm designed for simple graphs can be generalized to our constructed graph family. This is achieved by ``blowing up’’ each node into \(L\)-parallel copies.

\begin{lemma}
[Handling Parallel Edge]\label{thm:lift_algo}
Choose constant parallel multiplicity bound $L$, decay factor $\alpha$, error parameter $c$, failure probability bound $p$, and query complexity bound $T(n)$. Suppose there exists an algorithm $\algo$ that for any simple directed graph (with no dangling nodes) with $\bigo(n)$ nodes and $ \bigo(m)$ edges, it solves SSPPR-R problem with decay factor $(1-(1-\alpha)^{1/3})$ and threshold $\delta(1-(1-\alpha)^{1/3})/\alpha$ with probability $1-p$ in expected query complexity $\bigo(T)$ within the arc-centric model, then there exists an algorithm $\mathcal{A}_L$ that solves SSPPR-R problem for any directed multigraph with edge multiplicity bounded by $L$ (with no dangling nodes), containing $\bigo(n)$ nodes, and $ \bigo(m)$ edges. Specifically, $\mathcal{A}_L$ solves SSPPR-R with decay factor $\alpha$ and threshold $\delta$ with probability $1-p$ and expected query complexity $\bigo(T+n)$ in the arc-centric model.
\end{lemma}

For generality, we we fix an arbitrary multigraph $\mathrm{G}$ (in our application, \(\mathrm{G} = \mathrm{U}(r)\)). Our proof begins by considering how an algorithm \(\mathcal{A}\) can operate on a multigraph $\mathrm{G}$, whose edge multiplicity is bounded by \(L\). Intuitively, we construct an \emph{\(L\)-lift} (or a lift when there is no confusion in $L$) of \(\mathrm{G}\) by creating \(L\) parallel copies of every node and edge, thereby transforming \(\mathrm{G}\) into a simple graph. This lifted graph allows \(\mathcal{A}\) to run as if it were operating on a simple graph, while its queries can be simulated through queries on the original multigraph \(\mathrm{G}\). However, this construction presents subtle challenges: since we are restricted to local queries on \(\mathrm{G}\), we might not even know the number of parallel edges between two already explored nodes. To address this, we introduce independent randomness at each local structure to probabilistically simulate a consistent global lift. 
\mypara{Lift of graph.}
\input{floats/fig_multiple_edges}
We now formalize this idea by describing the \emph{\(L\)-lift} of a multigraph, which encompasses the family of lifted simple graphs on which \(\mathcal{A}\) will be executed, and which can be generated based on local information. Let $\mathrm{G}=(V_{\mathrm{G}},E_{\mathrm{G}})$ be a multigraph with
$|V_{\mathrm{G}}|\leq n$, $|E_{\mathrm{G}}|\leq m$, and edge multiplicity $\mathrm{MUL}(V_{\mathrm{G}})\leq L$, We then describe a \(L\)-lift of \(\mathrm{G}, \mathcal{L}_{\rho}(\mathrm{G}) = (\widehat{V}_U, \widehat{E}_U)\) with new node set $\widehat{V}_\mathrm{G}$, edge set $\widehat{E}_\mathrm{G})$ and \(\rho = (\rho_{(w,v)})_{(w,v)}\) is a sequence of attributes, which we will define later.
Concerning $\widehat{V}_\mathrm{G}$, we incorporate $2L$ extra nodes for every $v\in V_{\mathrm{G}}$, : \(  v^{\mathrm{in}} _1,\dots ,v^{\mathrm{in}} _L\) ,  \(v^{\mathrm{out}}_1,\dots ,v^{\mathrm{out}}_L\), which results in $\widehat{V}_{\mathrm{G}} \;=\; V \;\cup\; \{v^{\mathrm{in}} _1,\dots ,v^{\mathrm{in}} _L, v^{\mathrm{out}}_1,\dots , v^{\mathrm{out}}_L\mid \forall v \in V_{\mathrm{G}}\}.$ For $\widehat{E}_\mathrm{G}$, We further construct it by two parts:
\begin{enumerate}[label=(\roman*).]
    \item \emph{New edges}. Add the edge $(v,\,v^{\mathrm{out}}_i)$ and $(v^{\mathrm{in}}_i,\,v)$ and $(v^{\mathrm{in}}_i,\,v)$  for every $v\in V$ and $i\in\{1,\dots,L\}$.
    \item \emph{Copy of original edges}. For any node pair $(w,v)$, and choose an arbitrary injection $\rho_{(w,v)}:\{\text{edges from \(w\) to \(v\)}\}\to\{1,\dots ,L\}$, for each edge \(e\) add below edges to \(\widehat{E}_\mathrm{G}\) for $i=1,\dots,L$:
    $$ \bigl(w^{\mathrm{out}}_{i},\;
        v^{\mathrm{in}}_{(i+\rho_{(w,v)}[e]-1)\bmod L}\bigr).$$
\end{enumerate}
Later, we use \(\mathrm{INJ}_L(w,v)\) to denote all possible \(\rho_{(v,w)}\), and define \(\mathrm{INJ}_L(\mathrm{G}) := \prod_{w,v \in V_{\mathrm{G}}} \mathrm{INJ}_L(w,v)\) as the product over all node pairs in \(\mathrm{G}\). Since by construction there are at most \(L\) edges from \(w\) to \(v\), the lifted graph \(\mathcal{L}_{\rho}(\mathrm{G})\) is guaranteed to be a \emph{simple} graph—i.e., one without parallel edges. 

This construction accounts for the magic number like $1-(1-\alpha)^{1/3}/\alpha$ in our target Lemma. Specifically, there is a correspondence of PPR scores between $\mathrm{G}$ and $\mathcal{L}_{\rho}(\mathrm{G}) =  (\widehat{V}_\mathrm{G}, \widehat{E}_\mathrm{G})$:\begin{lemma}\label{lemma:PPR_lift}
For each multigraph $\mathrm{G}$ with no dangling nodes and parallel multiplicity at most $L$, decay factor $\alpha$, and 
$\rho = (\rho_{(w,v)}) \in \mathrm{INJ}_L(\mathrm{G})$, let $\mathcal{L}_{\rho}(\mathrm{G})$ denote the corresponding  lift. 
Then $\mathcal{L}_{\rho}(\mathrm{G})$ satisfies:
\begin{enumerate*}[label=(\roman*)]
\item $|\widehat{V}_G| = (2L+1)\,|V_G|$ and $|\widehat{E}_\mathrm{G}| = L\,|E_G| + 2L\,|V_G|$;
\item For each $u \in V_G$, its out degree is the same as the out degree of $u^{out}_i$;
\item for every pair of nodes $s,t$ in $\mathrm{G}$, their personalized PageRank (PPR) scores 
$\pi_{(\mathrm{G},\alpha)}(s,t)$ in $\mathrm{G}$ with decay factor $\alpha$ and 
$\pi_{(\mathcal{L}_{\rho}(\mathrm{G}),\,1-(1-\alpha)^{1/3})}(s,t)$ in $\mathcal{L}_{\rho}(\mathrm{G})$ with decay factor $1-(1-\alpha)^{1/3}$ satisfy:
$\frac{\pi_{(\mathcal{L}_{\rho}(\mathrm{G}),\,1-(1-\alpha)^{1/3})}(s,t)}{1-(1-\alpha)^{1/3}}
=
\frac{\pi_{(\mathrm{G},\alpha)}(s,t)}{\alpha}.
$\end{enumerate*}\end{lemma}
\begin{proof}
Property (i) and (ii) follow directly from the definition of the lift. For (iii), recall the path decomposition: $\pi(s,t)
= \text{Decay Factor}\times \sum_{\overleftrightarrow{p} \in \overleftrightarrow{p_H}(s,t)}
w(\overleftrightarrow{p}),$ so that $\pi(s,t) / \text{Decay Factor}$ is exactly the total probability
of a decayed random walk going through $\overleftrightarrow{p}$.

Fix an edge $(v,u)$ in $\mathrm{G}$. In $\mathcal{L}_{\rho}(\mathrm{G})$ this edge is
replaced by $L$ different length-$3$ paths
$v \to v^{\mathrm{out}}_i \to u^{\mathrm{in}}_{\ast} \to u
\quad (i = 1,\dots,L).$ Under decay factor $1 - (1-\alpha)^{1/3}$, the probability of
traversing such a path is $\frac{1-\alpha}{L\,d_{\mathrm{out}}(v)}.$ Summing over the $L$ gadgets gives total probability $\frac{1-\alpha}{d_{\mathrm{out}}(v)}$,which is exactly the per-edge factor for $(v,u)$ in $\mathrm{G}$ under decay factor $\alpha$. Thus, for any path $p$ in $\mathrm{G}$, the total probability (without
teleport) of all lifted paths in $\mathcal{L}_{\rho}(\mathrm{G})$ that
project to $p$ is equal to the probability of $p$ in $\mathrm{G}$ as claimed.
\end{proof}

\mypara{Construction of the algorithm.}
Now that we have defined \(L\)-lift of a graph, then we can dive into the details of $\mathcal{A}_L$. In this part, we will preprocess $\mathcal{A}$ for a simple choice of label, define the inner state of $\mathcal{A}_L$ and then explain $\mathcal{A}_L$'s action corresponding to each action of $\mathcal{A}$.

We now construct $\mathcal{A}_L$ and describe how it simulates $\mathcal{A}$. To simplify the handling of edge labels, we may without loss of generality assume that whenever $\mathcal{A}$ executes a query $\query{ADJ}$ at a node $v$, it chooses an uncovered in- or out-edge of $v$ uniformly at random.
Indeed, $\mathcal{A}$ is required to work for every labeling of the edges, with uniform bounds on complexity and failure probability. Randomizing the labels and then forcing $\mathcal{A}$ to pick uniform uncovered edges simply means that these guarantees hold on average over the choice of labels. This is also the default convention in previous work on unlabeled graphs
(e.g.,~\cite{bressan2023sublinear,wang2024revisiting}).

Because $\mathcal{A}_L$ only has oracle access to the original multigraph 
$\mathrm{G}$, it must maintain both (i) a record of all information revealed 
so far about $\mathrm{G}$ and (ii) a partial description of the lift 
parameters $\rho$. Concretely, $\mathcal{A}_L$ runs $\mathcal{A}$ as an 
inner algorithm and maintains an internal state:\(< \mathrm{\mathrm{G}}_0,  \rho^0>\). We next describe the two components $\mathrm{\mathrm{G}}_0$ and $\rho^0$.
\begin{enumerate}[label=(\roman*).]
  \item $\mathrm{G}_0$ is the covered subgraph of $\mathrm{G}$ revealed by 
        the queries that $\mathcal{A}_L$ has executed on $\mathrm{G}$, 
        together with degree information. It contains all nodes of 
        $\mathrm{G}$, the degrees of all nodes whose degree has been queried,
        and all covered edges with their original labels. For example, if the 
        first query $\query{ADJ}_{\mathrm{in}}(x_i,k)$ returns $(y_j,k')$, 
        then we add an edge between $x_i$ and $y_j$ to $\mathrm{G}_0$ and 
        record that this edge is the $k$th in-edge of $x_i$ and the $k'$th 
        out-edge of $y_j$, even though there is only one edge between $x_i$ 
        and $y_j$ in $\mathrm{G}_0$. Specifically, $\mathrm{G_0}$ may contain only part of an edge's label (e.g. $e = (u,v)$as $k^{th}$ out-edge of $u$ but does not know its index at $v$).
  \item $\rho^{0} = (\rho^{0}_{(w,v)})_{(w,v)}$ is an element of 
        $\mathrm{INJ}_L(\mathrm{G}_0)$. For each ordered pair $(w,v)$ in 
        $\mathrm{G}$ it is an injection from$\{\text{covered edges from $w$ to $v$}\}
        \quad\text{to}\quad
        \{1,\dots,L\}.$
\end{enumerate}

To simplify notation, we claim that an \(\rho \in \mathrm{INJ}(\mathrm{G})\) is over \( \rho^{0}\) if and only if \(\mathcal{L}_{\rho_0}(\mathrm{\mathrm{G}_0}) \subset \mathcal{L}_{\rho}(\mathrm{G})\); Equivalently, this means \(\rho_{(w,v)}(e) = \rho^{0}_{(w,v)}(e)\), for arbitrary node pairs \(w, v \) in \(\mathrm{G}\) and \(e \) from \(w\) to \(v\) in \(\mathrm{\mathrm{G}}_0\). Intuitively, this means that the inner \(\mathcal{A}\) knows it is possible on \( \mathcal{L}_{\rho}(\mathrm{G})\). With this information stored, \(\mathcal{A}_L\) is now able to simulate \(\mathcal{A}\). Specifically, when \(\mathcal{A}\) tries to execute a query \(Q\), \(\mathcal{A}_L\) processes it according to the strategies outlined below:
\begin{enumerate}[label=(\roman*).]
    \item Whenever \(\mathcal{A}_L\) discovers a new node in via some query on $\mathrm{G}$, it immediately issues degree queries for this node and records the answers  to $\mathrm{G}_0$.
    \item \(\mathcal{A}_L\) needs not to execute any query, and can give \(\mathcal{A}\) with an inferred return, because the return value of \( Q\) is identical on \(\mathcal{L}_{\sigma}(\mathcal{\mathrm{G}})\) for all \(\mathrm{\mathrm{G}},\rho\) over \(<\mathrm{\mathrm{G}}_0, \rho^{0}>\).  This means that no matter how we modify \(\rho^{0}\) later, the returning value of that query is the same, and then \(\mathcal{A}_L\) need not do actual queries. The following queries fall into this case:
    \begin{itemize}
        \item \(Q=\query{INDEG}\) or \(\query{OUTDEG}\). This is because whatever \(\rho\), the degree is fixed for each node, the copied nodes inherit those degrees as in definition, and \(\mathcal{A}_L\) stores all the degree queries for covered nodes ;
        \item \(Q\) is an \(\query{ADJ}_{in/out}\) query with an index \(k\) exceeding the corresponding degree; its answer is likewise predetermined;
        \item \(Q = \query{ADJ}_{in/out}(v, k)\). The answer is simply the node \(v^{in/out}_k\) or "None";
        \item  \(Q=\query{ADJ}_{out}(w_i^{out}, k)\) (the case \(Q=\query{ADJ}_{in}(w_{i}^{in}, k)\) is analogous), where the label "\(k^{th}\) out-edge" of \(w\) (or \(k^{th}\) in-edge of \(w\) ) corresponds to an edge \(e\) in \(\mathrm{\mathrm{G}}_0\) with known label at $w$. Denote the other node of that edge \(e\) is \(v\) and \(e\) is \(v'\)s \(k^{'th}\)  in-edge. Following the construction of lift, we have the target label \(j = ((i - 1 + \rho^0_{w,v}(e)) \bmod L) + 1\), then the inferred answer is \(v^{in}_j\).
    \end{itemize}
    \item \(Q=\query{JUMP}\)(): just execute \(\query{JUMP}\)() on \(\mathrm{G}\) and return \(w\). Then we choose an answer for \(\mathcal{A}\) uniformly random in \(\{w,w^{in}_1,\cdots,w^{in}_L,w^{out}_1,w^{out}_L\}\).
    \item \(Q=\query{ADJ}_{out}(w_i^{out}, k)\) or \(\query{ADJ}_{in}(w_i^{in}, k)\), where the label "\(k^{th}\) out-edge" of \(w\) (or \(k^{th}\) in-edge of \(w\) ) has not appear in \(\mathrm{\mathrm{G}}_0\). In this case, we can execute query \(\query{ADJ}_{out}(w, k)\) and then it turns to Case \((i)\) in the above. 

\end{enumerate} 
Up to now, we finish the construction of \(\mathcal{A}_L\) and then we will reveal next its properties.

\mypara{Properties of the algorithm.}
Now it is time to verify the property of our constructed algorithm \(\mathcal{A}_L\).  The requirements in the theorem can be specified as \begin{enumerate*}[label=(\roman*).]
    \item $\mathcal{A}_L$ can be applied on multigraphs like $\mathrm{G}$;
\item $\mathcal{A}_L$ has expected query complexity $\bigo(T+n)$ in arc-centric model;
\item $\mathcal{A}_L$ can output the PPR scores with relative error guarantee, and has failure probability bounded by $1-p$.
\end{enumerate*}
Requirement $(i)$ follows directly from the construction. For requirements $(ii)/(iii)$, this is intuitively because $\mathcal{A}$ satisfies these property on each $\mathcal{L}_{\rho}(\mathrm{G})$ due to Lemma \ref{lemma:PPR_lift}, and $\mathcal{A}_L$ works as if $\mathcal{A}$ is applied on a $\mathcal{L}_{\rho}(\mathrm{G})$ for a uniformly random sample of $\rho = (\rho_{(w,v)}) \in \mathrm{INJ}_L(\mathrm{G})$. 
We can formally define these two execution of $\mathcal{A}$:
\begin{enumerate}[label=(\roman*).]
\item Run $\mathcal{A}_L$ on the original graph $\mathrm{G}$, $\mathcal{A}$ is executed as an inner algorithm;
\item Sample a $\rho = (\rho_{(w,v)}) $ uniformly at random in $\mathrm{INJ}_L(\mathrm{G})$, and run $\mathcal{A}$ $\mathcal{L}_{\rho}(\mathrm{G})$.
\end{enumerate}
We now argue that these two executions generate the same joint distribution
over the transcript of $\mathcal{A}$ (its sequence of queries and answers).
We prove this by induction on the number of queries of $\mathcal{A}$,
viewing $\mathcal{A}$ itself as a randomized process whose next query may
depend on all previous answers.

The inductive hypothesis is that after $t$ queries the partial mapping
$\rho^0$ maintained by $\mathcal{A}_L$ has the same distribution as the
restriction of a uniformly random $\rho\in\mathrm{INJ}_L(\mathrm{G})$ to the
edges in $\mathrm{G}_0$. This is clearly true at $t=0$.

Consider the $(t+1)$-st query $Q$ of $\mathcal{A}$. If $Q$ is of a type for
which $\mathcal{A}_L$ answers without querying $\mathrm{G}$ (the cases listed
in (ii) above), then the answer is a deterministic function of 
$\langle \mathrm{G}_0,\rho^0\rangle$, and hence has the same distribution in
both executions. For degree queries or out-of-range adjacency queries, the
answer does not depend on $\rho$ at all.

If $Q$ is an adjacency query that reveals a new edge, then $\mathcal{A}_L$
chooses a fresh label in $\{1,\dots,L\}$ uniformly among those that have not
yet been used for the corresponding ordered pair $(w,v)$. This is exactly
how a uniformly random $\rho\in\mathrm{INJ}_L(\mathrm{G})$ behaves when
restricted to new edges: conditioned on the labels already assigned to
covered edges, the label of a new edge from $w$ to $v$ is uniform over the
remaining unused labels. Thus the distribution of $\rho^0$ after processing
$Q$ still matches the restriction of a random $\rho$, which maintains the
inductive invariant.

\mypara{Establishing Lemma \ref{thm:algo_adapt_tot}}
Together, Lemma~\ref{lemma:paraBoundU} and \ref{thm:lift_algo} yields a "lifted" algorithm  $\mathcal{A}_L$ on padded instances for $\mathrm{U} \in \mathcal{U}(n, D, d)$. Our goal requires us to restrict the $\mathcal{A}_L$ to $\mathrm{U}$. Intuitively, if we randomly permute the  in/out-edge labels in $\mathrm{U}(r)$, an $\query{ADJ}$ in $\mathrm{U}(r)$ can be translated to a query in $\mathrm{U}(r)$ with probability $(2D-d)/(r+2D-d)$. In this way, the query complexity is reduced by factor of $\bigo(D/(r+D))$ in our main lemma in this section.

\begin{lemmanrestate}
Fix parameters $0<d<D, r\in \bigo(n)$. Suppose any algorithm $\mathcal{A}_L$ solves SSPPR-R problem for any directed multigraph with edge multiplicity bounded by $L$ (with no dangling nodes), containing $\bigo(n)$ nodes,  $ n(D+r)$ edges, with decay factor $\alpha$ and threshold $\delta$, probability $1-p$ and expected query complexity $\bigo(T)$ in the arc-centric model. Then, there is an algorithm $\mathcal{A}^{res}_L$ on $\mathrm{U} \sim \Sigma(n, D, d)$ with average probability at least $1-2p$ solves SSPPR-R problem with at most $\bigo(n + TD/(p(r+D)))$ queries.
\end{lemmanrestate}
\begin{proof}
Without loss of generality, we may assume that whenever $\mathcal{A}_L$ executes a query $\query{ADJ}$ at a node $v$, it chooses an uncovered in- or out-edge of $v$ uniformly at random similarly.  Due to edge and node counting in Lemma \ref{lemma:PPR_Cal} and multiplicity bound in Lemma \ref{lemma:paraBoundU}, the hypothesis in our lemma is for $\mathcal{A}_L$'s performance on $\mathrm{U}(r), \mathrm{U}\sim \Sigma(n, D, d,r)$. 

Recall that we have not fixed a specific labeling scheme for \(\mathrm{U}(r)\) in earlier steps; instead, we simply assigned labels to the padding nodes \(Z_X\) and \(Z_Y\) as \(4n + 2,\dots,4n + 2r + 1\). The out-edges from each node \(y_i \in Y\) to \(Z_Y\) are assigned labels \(2D - d + 1,\dots,2D - d + r\), and the in-edges from \(Z_X\) to each \(x_i \in X\) are labeled analogously. Given this fixed labeling, all information regarding \(Z_X\) and \(Z_Y\) is known in advance from the perspective of \(\mathcal{A}_L\). Since \(\mathcal{A}_L\) operates only on \(\mathrm{U}(r)\), these edges can be hard-coded into the algorithm. We initialize a subgraph \(G_0\) internal to \(\mathcal{A}_L\) that stores all such predetermined edges in \(\mathrm{U}(r)\)—specifically, all the nodes \(s\), \(X\), \(Y\), \(Z_X\), and \(Z_Y\), along with the edges from \(Z_X\) to \(X\) and from \(Y\) to \(Z_Y\).

We then define the modified algorithm \(\mathcal{A}_L^{res}\) as follows. The core idea is to simulate \(\mathcal{A}_L\) on the original graph \(\mathrm{U}\) by internally emulating its behavior on the padded graph \(\mathrm{U}(r)\). Whenever \(\mathcal{A}_L\) issues a degree query or an \(\mathsf{Adj}\) query involving a node in the padding sets \(Z_X \cup Z_Y\), \(\mathcal{A}_L^{res}\) returns the predetermined response encoded in \(G_0\), which is consistent across all \(\mathrm{U}(r)\) constructed from some \(\mathrm{U} \in \mathcal{U}(n, D, d)\). In the case where \(\mathcal{A}_L\) issues a query of the form \(\mathsf{Adj}^{\mathrm{in}}(x_i, k)\), we proceed as follows: if \(k > 2D - d\), the query refers to a synthetic edge from \(Z_X\), so we return the \((k - 2D + d + 1)\)-st node in \(Z_X\); otherwise, we forward the actual query \(\mathsf{Adj}^{\mathrm{in}}(x_i, k)\) to the graph \(\mathrm{U}\), and add the resulting edge to the internal subgraph \(G_0\). In this way, \(\mathcal{A}_L^{res}\) simulates the interaction between \(\mathcal{A}_L\) and \(\mathrm{U}(r)\) using queries exclusively on \(\mathrm{U}\). Therefore, with success probability at least \(p\), \(\mathcal{A}_L^{res}\) correctly outputs the split \((X_1, X_2)\).

We now bound the query complexity of \(\mathcal{A}_L^{res}\). As assumed, the algorithm \(\mathcal{A}_L\) issues queries of the form \(\query{ADJ}_{\mathrm{out}/\mathrm{in}}(v, k)\) by selecting \(k\) uniformly at random from the un-queried incident edge labels. For any node \(v \in \mathrm{U}\) where fewer than \((D + r)/2\) incident labels have been examined, the probability that a new query returns a \emph{real} (non-padding) edge is at most \((2D - d)/(2D - d + r - (D + r)/2)\). Consequently, the expected number of genuine adjacency queries issued by \(\mathcal{A}_L^{res}\) can be bounded by the sum of two terms: the first is \((2D - d)\cdot \mathbb{E}[|\{v \in \mathrm{U} \mid \text{at least } (D + r)/2 \text{ incident labels examined}\}|]\); the second is \((2D - d)/(2D - d + r - (D + r)/2)\cdot o((D + r)n)\). Since both terms are asymptotically \(o(nD)\), the overall expected query complexity of \(\mathcal{A}_L^{res}\) is \(\bigo(TD/(r+D))\).

For the exact query complexity bound in our lemma, we can use Markov bound for an additional failure probability of $p/2$. With a $2/p$ additional multiplicity factor of the original cost.
\end{proof}

\begin{lemmanrestate}
Assume there is a SSPPR-R algorithm $\algo$ on simple directed graphs with expected query complexity $T = o(\min(m_0), \log(1/\delta_0)/\delta_0))$ with threshold $\delta$, error parameter $c$ and failure probability bound $p_f$. Then for any constant $\gamma$,  there are sufficiently large parameters $n, D, d$ along with an algorithm $\mathcal{A}_{ad}$. Using at most $\gamma nD$ arc-centric queries on each $\mathrm{U}\sim \Sigma(n, D, d)$, $\mathcal{A}_{ad}$ will output the ground truth split of $X = X_1 \sqcup X_2$ with average success probability at least $1-2p_f$.
\end{lemmanrestate}
\begin{proof}From Lemma \ref{lemma:PPR_R_const}, we choose sufficiently large $n, D, d, r$. The algorithm $\algo$ is adapted to $\mathcal{A}_L$ (Theorem \ref{thm:lift_algo}) and then to $\mathcal{A}^{res}_L$ on $\mathrm{U} \sim \Sigma(n, D, d)$. Condition (iv) of Lemma \ref{lemma:PPR_R_const} present that for $x_{i_1}\in X_1, x_{i_2}\in X_2$, $\pi(s, x_{i_1}), \pi(s, x_{i_2})$ are fixed values with $\pi(s, x_{i_1})> \pi(s, x_{i_2})/(1-c)^2$. Then, we add a post-processing step to $\mathcal{A}^{res}_L$: using its predicted PPR value $\hat\pi(s, x_i)$, we predict the split by assigning $x_i$ to $X_1$ if $\hat\pi(s, x_i) \geq \sqrt{\pi(s, x_{i_1}), \pi(s, x_{i_2})}$ (per  condition (iv) of Lemma \ref{lemma:PPR_R_const}). This composite algorithm is $\mathcal{A}_{ad}$. With average probability $1-2p_f$, $\mathcal{A}^{res}_L$ satisfies the SSPPR-R error guarantee. If it succeeds, for any $x_i \in X_1$, $\hat\pi(s, x_i) \geq (1-c) \pi(s, x_i) \geq \sqrt{\pi(s, x_{i_1}), \pi(s, x_{i_2})}$, ensuring correct assignment. Similarly, nodes in $X_2$ are assigned correctly. This completes the proof of correctively and now we can conclude that the complexity is $\bigo(n + TD/(p(r+D)))$. Given $T = o(\min(m_0), \log(1/\delta_0)/\delta_0))$, this is $o(n + \min(m_0, \log(1/\delta_0)/\delta_0) \cdot D/(r+D))$. By condition (ii) of Lemma \ref{lemma:PPR_R_const}, the second term is $o(nD)$. The total complexity is therefore $o(nD)$.
\end{proof}

\subsection{Proof of Lemma \ref{lemma:data_process}}\label{sec:proof:lemma:data_process}
\begin{lemmanrestate}[Formularize Success Probability]
For any algorithm $\mathcal{A}_{ad}$ using at most $T = \gamma nD$ arc-centric queries on an instance $\mathrm{U}\sim \Sigma(n, D, d)$, its expected probability of outputting the correct split of $X$ is bounded by: $\max_{e_1,\cdots, e_T}\max_{\mathrm{SP}}\pr_{\mathrm{U} \sim \Sigma(n, D,d)}[\mathrm{SP} \text{ is the correct split of} \;\;\mathrm{U} \mid e_1,\cdots, e_T].$
\end{lemmanrestate}
\begin{proof}
Let $\mathrm{U} \sim \Sigma(n,D,d)$ and let $\mathrm{SP}$ denote the (random) true split of $X$ induced by $\mathrm{U}$.  
Fix any algorithm $\mathcal{A}_{ad}$ that uses at most $T=\gamma nD$ arc–centric queries.  Naturally, we decompose \(\algo\) into two components: the first part, denoted \(\mathcal{A}_1\), is an \emph{exploration} procedure that performs queries on the graph; the second part, \(\mathcal{A}_2\), processes the information collected and outputs a split of \(X\) into subsets \(X_1\) and \(X_2\). The algorithm \(\algo\) succeeds if and only if \(\mathcal{A}_2\) outputs a correct split. Since we are working under the RAM model, we can assume that \(\mathcal{A}_1\) and \(\mathcal{A}_2\) operate independently—that is, \(\mathcal{A}_2\) is a separate algorithm whose input is the data written into memory by \(\mathcal{A}_1\), and whose output is the inferred split of \(X\).Given that RAM is finite (though arbitrarily large) and that the only inputs to \(\mathcal{A}_1\) are the query responses, which are part of the known edges. Thus, we may treat the final RAM state as a function of the known edges sequence $(e_1,\cdots, e_T)$, along with any internal randomness, where $e_t$ is the $t$-th revealed edge (we can pad with dummy symbols if the algorithm stops early).
The expected success probability of \(\algo\) can thus be formulated as:

\begin{align}
\nonumber
&\mathbb{E}_{\mathrm{U} \sim  \Sigma(n, D, d)}\bigl[\pr[\algo_{ad} \text{ succeeds on } U]\bigr]\leq \sum_{e_1, \cdots, e_T} \pr_{\algo, {\mathrm{U}\sim \Sigma(n, D, d)}}[e_1, \cdots, e_T] \\
&\times \max_{\mathrm{\mathrm{SP}} \in \binom{X}{n}} \pr_{\mathcal{U}\sim \Sigma(n, D, d)}[\mathrm{SP} \text{ is the correct split of }  \mathrm{U} \mid e_1, \cdots, e_T]\\
&\leq \max_{e_1,\cdots, e_T}\max_{\mathrm{SP}}\pr_{\mathrm{U} \sim \Sigma(n, D,d)}[\mathrm{SP} \text{ is the correct split of} \;\;\mathrm{U} \mid e_1,\cdots, e_T].
\end{align}
This achieves our previous claim.
\end{proof}

\subsection{Proof of Lemma~\ref{lemma:targetProb}}
\label{sec:proof:lemma:targetProb}

\begin{lemmanrestate}
%\label{lemma:targetProb}
There is a constant \(\gamma\), such that: For \(n\) sufficiently large, \(d \leq log(n)\), \(D>\frac{3d}{2}\),  \(T\leq \gamma nD\), and given any $T$ known edges $e_1, \cdots, e_T$, and a split of \(X\), \(\mathrm{SP}\), we have:
$$\pr_{\mathrm{U}\sim \Sigma(n, D, d)}[\mathrm{SP} \text{ is the correct split of }  \mathrm{U} \mid e_1, \cdots, e_T] \leq \frac{1}{n}.$$
\end{lemmanrestate}

\begin{proof}
Our goal is to bound
\begin{align}
\pr_{\mathrm{U}\sim\Sigma(n,D,d)}
\!\Bigl[\,\mathrm{SP}\text{ is the true split}\;\Bigm|\;e_1, \cdots, e_T\Bigr].
\end{align}
The distribution \(\Sigma(n,D,d)\) is generated in two independent steps:
\begin{enumerate}[label=(\roman*).]
  \item \emph{Choose the split.}\;
        Pick \(\mathrm{SP}\in\binom{X}{n}\) \emph{uniformly at random}.
  \item \emph{Choose all permutations.} We choose all $\bigl(\mathrm{P}_1^{2D-d},\cdots,\mathrm{P}_{2n}^{2D-d},\,\mathrm{P}_{Y_1}^{n(2D-d)},\,\mathrm{P}_{Y_2}^{n(2D-d)},\,\mathrm{S}(Y_a, X_b)\bigr)$, where \(\mathrm{S}(Y_a, X_b)\in\{\mathrm{S}(Y_{1}, X_{1}),\mathrm{S}(Y_{1}, X_{2}),\mathrm{S}(Y_{2}, X_{1}),\mathrm{S}(Y_{2}, X_{2})\}\), randomly and independently.
\end{enumerate}
By Bayes’ rule, the posterior probability mass on the event:
\(\,\mathrm{SP}\text{ is the true split}\,\), equals to
\begin{equation}
\label{eq:bayes}
  \frac{\pr_{\mathrm{U}\sim \Sigma(n,D,d)|_{\mathrm{SP}}}\!\bigl[e_1,\cdots, e_T\bigr]}
{\displaystyle\sum_{\mathrm{SP}^{'}\in \binom{X}{n}}
        \pr_{\mathrm{U}\sim \Sigma(n,D,d)|_{\mathrm{SP}^{'}}}\!\bigl[e_1,\cdots, e_T\bigr]},
\end{equation}
where the sum runs over any
\(\sigma=(\mathrm{SP},\mathrm{P}_1^{2D-d},\cdots,\mathrm{P}_{2n}^{2D-d},\mathrm{P}_{Y_1}^{n(2D-d)},\mathrm{P}_{Y_2}^{n(2D-d)},\mathrm{S}(Y_a, X_b)), \forall a,b\in\{1,2\}\).
Because all splits are a-prior equiprobable, the posterior on any
particular \(\mathrm{SP}\) is proportional to the likelihood of the observed
transcript. 
Now we will give a systematic evaluation of each term in this fraction. Since queries from \(s\) or to \(s\) is the same on all the \(\mathrm{U}(r)\in\mathcal{U}(n, D, d)\), we assume all the queries \(Q_t, R_t\) is only between \(X\) and \(Y\). We now derive Lemma~\ref{lemma:condPQR} to compute the overall probability. For ease of representation, we include the following additional notations:
\begin{enumerate}[label=(\roman*).]
    \item  $e_t = [(y_{i_t},k_t), (x_{j_t},l_t)]$ ,$1\leq t \leq T_0$.  represents an edges as the $k_t^{th}$ out-edge of  $y_{i_t} \in Y$ and  the $l_t^{th}$ in-edge of  $x_{j_t} \in X$.
    \item For node sets (or single node labels) \(V_1, V_2\), we use \(E_t(V_1\rightarrow V_2)\) for the count of \(\{1\leq t'< t, e_{t'} \text{ originates from } V_1, \text{pointing to } V_2\}\), e.g.,  \(E_{t}(Y_1, x_{i_t})\) is the count of edges from \(Y_1\) to \(x_{i_t}\)among know edges.
    %\item Since we will no longer use the threshold in this section, we use the standard notation \(\delta_{ab} = \begin{cases}1, \; a=b\\ 0, \; a\neq b\end{cases}\) for the constant for \(a,b \in \{1,2\}\).
\end{enumerate}

%\begin{definition}
%We here introduce a sequence of notations for simplification in follow content (\(1\leq t \leq T_0\) for integer $t$):
%\end{definition}
We now formulate the calculation in the following lemma.

%\mathbf{1}_{a=b}=\begin{cases}1, a=b\\0, a\neq b\end{cases}
\begin{lemma}(Conditional probability of queries)\label{lemma:condPQR} For any $\; x_{i_{T_0 + 1}} \in X_b,\; y_{j_{T_0 + 1}} \in Y_a,\; a,b\in\{1,2\}:$
\begin{align}
\pr_{U\sim \Sigma(n,D,d)|_{\mathrm{SP}}}\bigl[e_{T_0 + 1}\mid e_1,\cdots, e_{T_0}\bigr] =\frac{D-(1-\mathbf{1}_{a=b})d-E_{T_0+1}(Y_{a}\rightarrow {x_{T_0 + 1}})}{\left(2D-d-E_{T_0+1}({Y\rightarrow {x_{T_0 + 1}}})\right)\left(n(2D-d) - E_{T_0+1}({Y_{a}\rightarrow X)}\right)}.
\end{align}
\end{lemma}
We proceed by induction; our induction hypothesis is:
In the distribution \(\Sigma(n,D,d)|_{SP}\) conditioned on \(Q_1=R_1,\dots Q_t=R_t\), \(\mathrm{P}_1^{2D-d},\cdots,\mathrm{P}_{2n}^{2D-d}, \mathrm{P}_{Y_1}^{n(2D-d)}, \mathrm{P}_{Y_2}^{n(2D-d)}\) are still independent of each other, and that \(\mathrm{S}(Y_{1}, X_{1}),\mathrm{S}(Y_{1}, X_{2}),\mathrm{S}(Y_{2}, X_{1}),\mathrm{S}(Y_{2}, X_{2})\) are also independent of each other. 
\begin{comment}
Specifically, the distribution is isomorphic to the uniform distribution of the following set:
\begin{align*}
&\Bigl(\prod_{i=1}^{2n} SYM(\{1,..,2D-d-E_{Y_{1}\rightarrow {x_{i_0}}}-E_{Y_{2}\rightarrow {x_{i_0}}}\})\Bigr)\\
&\times SYM(\{1,\cdots,n(2D-d)-E_{Y_{1}\rightarrow X_{1}} - E_{Y_{1}\rightarrow X_{2}}\})
\times SYM(\{1,\cdots,n(2D-d)-E_{Y_{2}\rightarrow X_{1}} - E_{Y_{2}\rightarrow X_{2}}\})\\
&\rtimes \Pi\begin{cases}
SYM(\{1,\cdots,nD - E_{Y_{1}\rightarrow X_{1}}\})\\
SYM(\{1,\cdots,n(D-d) - E_{Y_{1}\rightarrow X_{2}}\})\\
SYM(\{1,\cdots,n(D-d) - E_{Y_{2}\rightarrow X_{1}}\})\\
SYM(\{1,\cdots,nD - E_{T_0+1}({Y_{2}\rightarrow X_{2}}\}))
\end{cases}
\end{align*}
Particularly, the first term means for each \(x_i\) we still need to decide the source of in-edges except those covered in \(Q_i, R_i\); the second two terms stands for deciding the target of out-edges in \(Y_1, Y_2\). the final factor specifies the individual connections, as in Procedure \ref{sec:distribution}.
\end{comment}
Without loss of generality, we prove the case \(x_{i_{T_0+1}} \in X_1, y_{j_{T_0+1}} \in Y_1\).  As in the single query:
\begin{enumerate}[label=(\roman*).]
    \item By definition, \(\mathrm{SP}\) determines whether \(x_{i_{T0+1}}\) is in \(X_1\) or \(X_2\). Here we assume that it lies in \(X_1\).
    \item \(\mathrm{P}^{2D-d}_{i_{T_0+1}}\) decides whether the \(k_{T_0+1}^{th}\) edge of \(x_{i_{T_0+1}}\) is sourced from \(Y_1\). This occurs with probability 
    \begin{align}
    \mathrm{P}^{2D-d}_{i_{T_0+1}}=\frac{D-  E_{T_0+1}({Y_1\rightarrow {x_{i_{T_0+1}}}})}{2D-d - E_{T_0+1}({Y\rightarrow {x_{i_{T_0+1}}}})}.
    \end{align}
    
    \item\(\mathrm{P}_{Y_1}^{n(2D-d)}\) decides whether the \(l_{T_0+1}^{th}\) edge of \(y_{j_{T_0 + 1}}\) targets to \(X_1\). This occurs with probability  
    \begin{align}
    \mathrm{P}_{Y_1}^{n(2D-d)}=\frac{nD- E_{T_0+1}{(Y_{1}\rightarrow X_{1}})}{n(2D-d) - E_{T_0+1}({Y_{1}\rightarrow X})}.
    \end{align} 
    \item Since in the composition, the \(\mathrm{P}_1^{2D-d},\cdots,\mathrm{P}_{2n}^{2D-d}, \mathrm{P}_{Y_1}^{n(2D-d)}, \mathrm{P}_{Y_2}^{n(2D-d)}\) terms are irrelevant to each other by induction hypothesis: For simplification of notations, we call the event for \(\mathrm{P}_1^{2D-d},\cdots,\mathrm{P}_{2n}^{2D-d}, \mathrm{P}_{Y_1}^{n(2D-d)}, \mathrm{P}_{Y_2}^{n(2D-d)}\) suitable for \(e_t\): When 
    \begin{equation}
            x_{i_t}\in X_1,y_{j_t} \in Y_2, 
    \end{equation}
    the event holds if and only if the $k_{Q_{t}}$ in-edge of $x_{i_t}$, i.e. \((x_{i_t},  k_{Q_{t}})\) originates from nodes in \(Y_2\); and the out-edge  \((y_{j_{t}}, k_{R_{t}})\) points to nodes in \(X_1\). We analogously define the same event for the other seven cases \(x_{i_t} \in X_2\) or  \(y_{j_t} \in Y_1\) or \(e_t = [(y_{j_{t}}, k_{t}), (x_{i_{t}}, l_{t})]\). Combined with induction hypothesis, the probability of
    \begin{align}
&\pr_{\mathrm{P}_1^{2D-d},\cdots,\mathrm{P}_{2n}^{2D-d}, \mathrm{P}_{Y_1}^{n(2D-d)}, \mathrm{P}_{Y_2}^{n(2D-d)}} \bigl[\text{suitable for }e_{T_0+1} \mid \text{suitable for }e_t, t=1,\cdots,T_0\bigr] \\
&=\frac{D- E_{T_0+1}({Y_{ 1}\rightarrow {x_{i_{T_0+1}}}})}{2D-d - E_{T_0+1}({Y\rightarrow {x_{i_{T_0+1}}}}) } \times \frac{nD- E_{T_0+1}({Y_{1}\rightarrow X_{1}})}{n(2D-d) - E_{T_0+1}({Y_{1}\rightarrow X})}.
    \end{align}
    \item Conditioned on arbitrary possible \(\mathrm{P}^{n(2D-d)}_i, \mathrm{P}^{n(2D-d)}_{Y_1}, \mathrm{P}^{n(2D-d)}_{Y_2}\), the bijection \(\mathrm{S}(Y_{1}, X_{1})\)  of \(nD\) elements decides the particular out-edge index from the reverse of the \(k_{Q_{T_0+1}}^{th}\) in-edge of \(x_{i_{T_0+1}}\) among \(nD-E_{Y_{1}\rightarrow X_{1}}\) choices, The probability of exactly finding it is the \(k_{R_{T_0+1}}^{th}\) out-edge of \(y_{j_{T_0+1}}\) is \(\frac{1}{nD-E_{T_0+1}({Y_{1}\rightarrow X_{1}})}\). 
\end{enumerate}

Consequently, the resulting probability matches the expression in Lemma~\ref{lemma:condPQR}. Since we only set some restriction independently for \(\mathrm{P}_{Y_1}^{n(2D-d)},\mathrm{P}^{2D-d}_{i_{T_0+1}}\), the \(\mathrm{P}_1^{2D-d},\cdots,\mathrm{P}_{2n}^{2D-d}, \mathrm{P}_{Y_1}^{n(2D-d)}, \mathrm{P}_{Y_2}^{n(2D-d)}\) are still independent conditioned on \(Q_{t+1} = R_{t+1}\), and also for the independence of $$\mathrm{S}(Y_{1},X_{1}),\;\mathrm{S}(Y_{1},\; X_{2}),\;\mathrm{S}(Y_{2}, X_{1}),\;\mathrm{S}(Y_{2}, X_{2}).$$  Thus the induction hypothesis remains valid for \(T_0+1\).

Combining the analysis above, it remains to show the left hand side expression is smaller than \(\frac{1}{n}\) for Lemma~\ref{lemma:targetProb}.
When \(T\leq \gamma nD\), there are at most \(\frac{\gamma}{\gamma_0} n\) nodes \(i\) such that \(E_{T}({Y \rightarrow x_i}) \geq \gamma_0 D\).  Then under the split \(\mathrm{SP}\), there are at least \((1 - \frac{\gamma}{\gamma_0})\cdot n\) nodes in \(X_1\), and at least \((1 - \frac{\gamma}{\gamma_0})\cdot n\) nodes in \(X_2\).  For each pair of \(x_i, x_{i^{'}}\) such that \(E_T({Y \rightarrow x_i}), E_T({Y \rightarrow x_{i^{'}})} < \gamma_0 D\) and in \(\mathrm{SP}\) \(x_i \in X_1, x_{i^{'}} \in X_2\), we have another split \(\mathrm{SP}^{'}\), that \(\mathrm{SP}^{'}\) sends \(x_i\) to \(X_2\), \(x_{i^{'}}\) to \(X_1\), while other nodes are the same as \(\mathrm{SP}\); Further, there are at least \((1 - \frac{\gamma}{\gamma_0})^2 n^2\) such pairs.
Intuitively, \(\mathrm{SP}^{'}\) is very close to \(\mathrm{SP}\), and the success probability should be very close to that of \(\mathrm{SP}\)  as well if we have only known insufficient data. Specifically, we calculate the proportion of their success probability as follows:

\begin{align}
  &\frac{\pr_{U\sim \Sigma(n,D,d)|_{SP}}\!\bigl[e_1,\dots,e_T\bigr]}
  {\pr_{U\sim \Sigma(n,D,d)|_{SP^{'}}}\!\bigl[e_1,\dots,e_T\bigr]}= \Pi_{t=0,\cdots,T-1}\frac{\pr_{U\sim \Sigma(n,D,d)|_{SP}}\!\bigl[e_{t+1}\mid e_1,\dots,e_T\bigr]}
  {\pr_{U\sim \Sigma(n,D,d)|_{SP^{'}}}\!\bigl[e_{t+1}\mid e_1,\dots,e_T\bigr]}.\\
\end{align}
Denote \(Q_t, R_t\) by either \(\query{ADJ}_{in}(x_{i_t}, k_{Q_{t}}),(y_{j_t}, k_{R_{t}}) \) or \(\query{ADJ}_{out}(y_{j_t}, k_{Q_{t}}),(x_{i_t}, k_{R_{t}})\ \text{for } t=1,..,T\). By Lemma \ref{lemma:condPQR}, we have that for any $x_{i_{t + 1}} \in X_b, y_{j_{t + 1}} \in Y_a, a,b\in\{1,2\}$, 
\begin{align}
&\pr_{U\sim \Sigma(n,D,d)|_{\mathrm{SP}}}\!\bigl[e_{t+1}\mid e_1,\dots,e_T\bigr] =\frac{D - (1-\mathbf{1}_{a=b})d-E_{t+1}(Y_{a}\rightarrow {x_{i_{t + 1}}})}{\left(2D-d-E_{t+1}({Y\rightarrow {x_{i_{t + 1}}}})\right)\left(n(2D-d) - E_{t+1}({Y_{a}\rightarrow X}\right)}.
\end{align}
Notice that for both \(\mathrm{SP}^{'}\) and \(\mathrm{SP}\), \(a\) is a fixed number, and if \(x_{i_{t+1}} \notin \{x_i, x_{i^{'}}\}\), \(b\) is also fixed. Then
\begin{align}
E_{t+1}({Y_a\rightarrow x_{i_{t+1}}}), E_{t+1}(Y\rightarrow x_{i_{t+1}}), E_{t+1}(Y_1\rightarrow x_{i_{t+1}})
\end{align} are the same in both cases . When\(x_{i_{t+1}} \notin \{x_i, x_{i^{'}}\}\), the target fraction is \(1\), otherwise the fraction is at most 
\begin{align}
1+\frac{d}{2D-d- E_{Y\rightarrow x_{i_{t+1}}}} \leq 1+\frac{d}{2D-d-\gamma_0 D} \leq 1+\frac{2}{4-3\gamma_0}\frac{d}{D}.
\end{align}
Thus, it holds that 
\begin{align}
\frac{\pr_{U\sim \Sigma(n,D,d)|_{\mathrm{SP}}}\!\bigl[e_1,\dots,e_T\bigr]}
  {\pr_{U\sim \Sigma(n,D,d)|_{\mathrm{SP}^{'}}}\!\bigl[e_1,\dots,e_T\bigr]} \leq \left(1+\frac{2}{4-3\gamma_0}\frac{d}{D}\right)^{\gamma_0 D}.
\end{align}
Since \(d\leq \log(n)\), this expression is at most \(n^{\frac{2\gamma_0}{3-4\gamma_0}}\). Notice that there are at least \((n - \frac{\gamma}{\gamma_0} n)^2\) such pairs, which creates  \((n - \frac{\gamma}{\gamma_0} n)^2\) terms of \(\pr_{U\sim\Sigma(n,D,d)|_{\mathrm{SP}^{'}}}\!\bigl[e_1,\dots,e_T\bigr]\) in the lower part of Equation~\ref{eq:bayes}, and they are at least \(n^{-\frac{2\gamma_0}{3-4\gamma_0}}\cdot\pr_{U\sim\Sigma(n,D,d)|_{\mathrm{SP}^{'}}}\!\bigl[e_1,\dots,e_T\bigr]\). Thereby the expression is no more than \((1-\frac{\gamma}{\gamma_0})n^{-2+\frac{2\gamma_0}{3-4\gamma_0}}\). When \(\gamma\) is a small constant and \(n\) is sufficiently large, the probability is smaller than \(\frac{1}{n}\).
\end{proof}

\section{Proof Detail in STPPR Lower Bound}
\subsection{Proof of Lemma~\ref{lemma:TPPR}}
\label{sec:proof:lemma:TPPR}
\begin{lemmanrestate}[STPPR]
Let $\mathrm{U}^{-1}$ denote the inverse graph of a graph $\mathrm{U}$ obtained by reversing the direction of each edge.  
Choose any decay factor $\alpha\in (0,1)$, error parameter $c\leq \tfrac{1}{2}$, and functions $\delta_0(n_0)\in \bigl(0,\tfrac{(1-\alpha)^2}{16}\bigr)$ and $m_0(n_0)\in \Omega(n_0)\cap \bigo(n_0^2)$.  
For sufficiently large $n_0$, there exist graph parameters $n,D,d$ (as functions of $n_0$) such that for all graph instances $\mathrm{U}\in\mathcal{U}(n,D,d)$, following conditions hold for $\mathrm{U}^{-1}$:
\begin{enumerate}[label=(\roman*).]
    \item The node count of $\mathrm{U}$ is in $\bigo(n_0)$, and the edge count of $\mathrm{U}$ is in $\bigo(m_0)$.
    \item $n^{1/2}\geq D\geq 2d$ and $d\leq \log n$.
    \item $Dn = \Omega\!\left(\min\!\left(m_0, n_0\log n_0 \cdot \tfrac{1}{\delta_0}\right)\right)$.
    \item For all $x_1\in X_1$ and $x_2\in X_2$ with PPR values satisfying $\pi(x_1,s)\geq \delta_0$ and $\pi(x_2,s)\geq \delta_0$, we have $\pi(x_1,s)\;\geq\; \frac{1}{(1-c)^2}\,\pi(x_2,s).$
\end{enumerate}
\end{lemmanrestate}

\begin{proof}
As in the original hard instance, in $\mathrm{U}^{-1}$ we have for $x_i\in X_1$:
$
\pi(x_i,s)=(1-\alpha)^2\frac{D}{2D-d+r},
$
and for $x_{i'}\in X_2$:
$
\pi(x_{i'},s)=(1-\alpha)^2\frac{D-d}{2D-d+r},
$
so that their ratio remains $\frac{D-d}{D}$. Consider first the case $m_0 > n_0 \log n_0 \cdot \tfrac{1}{\delta_0}$.  
Set $n = n_0$, $d = \log n_0$, $D = \left(1 + \frac{1}{4c}\right) d$, and 
$
r = d\left(\frac{(1-\alpha)^2}{4c\delta_0} - 1 - \frac{1}{4c}\right).$
The edge count satisfies $\bigo(n(D+r)) = \bigo(n_0\log n_0 \tfrac{1}{\delta_0}) \subset \bigo(m_0)$; the node count is $\bigo(\max(n,r)) = \bigo(\max(n,m/n)) = \bigo(n_0)$.  
For each $x_2 \in X_2$, $\pi(x_2,s) = \delta_0$, and the ratio condition holds as in the Lemma~. For the second case, where $\tfrac{1}{\delta_0}\le m_0 < n_0 \log n_0 \cdot \tfrac{1}{\delta_0}$, by continuity we can find positive $n,\delta$ such that
$n \log n \cdot \tfrac{1}{\delta} = m_0,$
with $n\le n_0$, $\delta \ge \delta_0$, and $\binom{n}{2}\ge m_0$.  
Setting $d=\log n$, $D=\left(1+\frac{1}{4c}\right)d$, and $r = d\left(\frac{(1-\alpha)^2}{4c\delta} - 1 - \frac{1}{4c}\right),$
we similarly ensure that $\pi(x_2,s)=\delta \ge \delta_0$, and the ratio condition again holds.  
Moreover, $n(r+D) \in \Omega\!\left(\tfrac{nd}{\delta}\right) = \Omega(m_0),$ the edge count is then $$\bigo(n(r+D))=\bigo(\tfrac{nd}{\delta})=\bigo(m_0)$$ and the node count satisfies $$\bigo(\max(n,r)) = \bigo(n,m/n)\subset \bigo(n_0).$$ This completes the proof.
\end{proof}

\section{SSPPR-A Lower Bound Proof}

\subsection{Proof of Lemma~\ref{lemma:GDB_prop}}
\label{sec:proof:lemma:GDB_prop}
\begin{lemmanrestate}[Properties of $\mathcal{G}(D, \mathbf{b})$]
For any graph instance $\mathcal{G}(D, \mathbf{b})$, the following properties hold:
\begin{enumerate}[label=(\roman*).]
\item The graph contains $4D+2$ nodes and $2D(D+1)$ edges.
\item The out-degree of each node in $Y$ or $Y'$ is $D$. The in-degree of each node in $X$ or $X'$ is $D$.
\item $\pi(s , t) = \alpha \left(\sum_{i,j=1}^{D}b_{ij}\right)\frac{(1-\alpha)^3}{D^2}$.
\end{enumerate}
\end{lemmanrestate}

\begin{proof}
The number of nodes follows directly from the construction:
$2$ special nodes ($s$ and $t$) plus $4D$ nodes in the sets $X,Y,X',Y'$,  
for a total of $2 + 4D$ nodes. For the edges, there are
$D$ edges from $s$ to $Y$,
$D$ edges from $X$ to $t$, and
$2D^{2}$ edges between $\{Y,Y'\}$ and $\{X,X'\}$,
giving a total of $2D^{2}+2D \;=\; 2D(D+1).$ For property (ii), take any node $y_i\in Y$ (or $y_i'\in Y'$).  
For every $j\in\{1,\dots,D\}$ there is exactly one outgoing edge—either to $x_j$ or to $x_j'$.  
Hence each $y_i$ (and each $y_i'$) has out-degree exactly $D$. Similarly, for any node $x_j\in X$ (or $x_j'\in X'$), each $i\in\{1,\dots,D\}$ contributes exactly one incoming edge (from either $y_i$ or $y_i'$), so every $x_j$ (and $x_j'$) has in-degree exactly $D$. Now we compute the PPR value $\pi(s,t)$ by tracing the probability mass of an $\alpha$-decay random walk starting at $s$.  
Initially the mass at $s$ is $\alpha$.

\begin{enumerate}[label=(\roman*).]
\item With probability $(1-\alpha)$ the walk leaves $s$ and spreads uniformly to its $D$ out-neighbors in $Y$.  
Hence each $y_i\in Y$ receives $\frac{\alpha(1-\alpha)}{D}.$

\item Each $y_i$ has out-degree $D$ and forwards its mass to nodes in $X$ or $X'$ according to $b_{ij}$.  
The total contribution to a node $x_j\in X$ is
$\sum_{i=1}^D b_{ij}\cdot \frac{\alpha(1-\alpha)}{D}\cdot \frac{1-\alpha}{D}
= \left(\sum_{i=1}^D b_{ij}\right)\frac{\alpha(1-\alpha)^2}{D^2}.$

\item Each $x_j$ has a single outgoing edge to $t$, passing on all of its mass with an additional $(1-\alpha)$ factor.  
Summing over $j$ gives
\(
\pi(s,t)
=\sum_{j=1}^D \left[\left(\sum_{i=1}^D b_{ij}\right)\frac{\alpha(1-\alpha)^2}{D^2}\right](1-\alpha)
=\alpha\left(\sum_{i,j=1}^D b_{ij}\right)\frac{(1-\alpha)^3}{D^2}.
\)
\end{enumerate}
This establishes the desired formula for computing PPR value $\pi(s,t)$.
\end{proof}

\subsection{Proof of Lemma~\ref{lemma:SSPPR-A1}}
\label{sec:proof:lemma:SSPPR-A1}
\begin{lemmanrestate}
Given constants $\alpha \in (0,1)$, $p \in (0,1)$, and \(m_0 \in \Omega(n_0), m_0 \leq \binom{n_0}{2}, \varepsilon_0 \in (0,1))\) as function of \(n_0\) and assuming \(\varepsilon_0\) is sufficiently small when \(n_0\) is sufficiently large. There exist constants $\gamma > 0$ such that for any sufficiently small absolute error tolerance $\varepsilon > 0$, we can choose an integer $D(n_0)$  satisfying:
\begin{enumerate}[label=(\roman*).]
    \item The graph $\mathcal{G}(D, \mathbf{b})$ has no more than $n_0$ nodes and no more than $m_0$ edges.
    \item If two matrices $\mathbf{b}$ and $\mathbf{b}'$ have sums $S = \sum b_{ij}$ and $S' = \sum b'_{ij}$ such that $|S - S'| \geq \gamma D$, then the corresponding PPR values satisfy $|\pi(s,t) - \pi'(s,t)| > 2\varepsilon$.
    \item \(D^2 = \Omega(\min(1/\varepsilon_0^2, m_0))\). 
\end{enumerate}
\end{lemmanrestate}

\begin{proof}
Choose \(D = \sqrt{min(1/\varepsilon^2, m_0)}\frac{(1-\alpha)^3}{8}\). Condition \((i)\) follows directly from Lemma \ref{lemma:GDB_prop}, since \(n_0 \geq  \sqrt{m_0}/2\geq D\). Condition \((iii)\) follows by definition. Now our goal is to show condition \((ii)\).
From the properties of our graph family, we have the change in PPR value:
$$|\pi(s,t) - \pi'(s,t)| = \alpha \frac{(1-\alpha)^3}{D^2} |S - S'|$$
Given that $|S - S'| \geq \gamma D.$ for some constant $\gamma$ we will choose, this implies:
$$|\pi(s,t) - \pi'(s,t)| \geq \alpha (1-\alpha)^3 \gamma/ D > 2\varepsilon.$$
\end{proof}

\subsection{Proof of Lemma~\ref{lemma:statistical_hardness}}
\label{sec:proof:lemma:statistical_hardness}
\begin{lemmanrestate}
For any constant $p \in (0,1)$, there exists a constant $\gamma' > 0$ such that for a sufficiently large $D$, any algorithm making at most $T \leq D^2/2$ queries to a graph drawn from $\Sigma(D)$ will fail to estimate $S = \sum_{i,j=1}^{D^2} b_{ij}$ with an error less than $\gamma' D$. Specifically, if $\hat{S}$ is the algorithm's estimate, then:
$$\pr_{\mathbf{b} \sim \Sigma(D)}[|S - \hat{S}| > \gamma' D] > p.$$
\end{lemmanrestate}

\begin{proof}
Let $\mathcal{A}$ be an arbitrary algorithm operating within our query model. We first analyze the information gained from its queries.
The graph structure is highly regular. The node set is fixed, as are the degrees of all nodes. Queries like \(\query{DEG}\) are uninformative about the matrix $\mathbf{b}$. Furthermore, the neighbors of $s$ (the set $Y$) and the nodes pointing to $t$ (the set $X$) are also fixed. Thus, \(\query{ADJ}\)  queries on $s$ and $t$ yield no information about $\mathbf{b}$.

The only queries that reveal the structure of $\mathbf{b}$ are \(\query{ADJ}\) queries on nodes in $\{Y, Y', X, X'\}$. For instance, querying the $j$-th neighbor of $y_i$ reveals whether the edge is $(y_i, x_j)$ or $(y_i, x_j')$, which is equivalent to learning the value of $b_{ij}$. Each such query reveals exactly one entry of the matrix $\mathbf{b}$.

The algorithm $\mathcal{A}$ makes at most $T \leq D^2/2$ queries. Let $K$ be the set of $T$ indices $(i,j)$ corresponding to the entries of $\mathbf{b}$ that the algorithm learns. The algorithm knows the values $\{b_k\}_{k \in K}$. The remaining $M = D^2 - T$ entries, where $k \notin K$, remain unknown. Since $T \leq D^2/2$, the number of unknown entries is $M \geq D^2/2$. Let $S$ be the true sum of all entries in $\mathbf{b}$. We can write $S = S_{obs} + S_{unobs}$, where $S_{obs} = \sum_{k \in K} b_k$ is the sum of observed entries and $S_{unobs} = \sum_{k \notin K} b_k$ is the sum of unobserved entries. The algorithm knows $S_{obs}$ exactly. The unobserved sum, $S_{unobs}$, is a sum of $M$ independent and identically distributed $Bern(1/2)$ random variables. Thus, $S_{unobs}$ follows a binomial distribution, $Binomial(M, 1/2)$, with mean $\mu=M/2$ and standard deviation $\sigma = \sqrt{M}/2$.

An optimal algorithm's best estimate for this unknown sum is its mean, $\hat{S}_{unobs} = M/2$. The algorithm's estimation error is therefore $|S_{unobs} - M/2|$. We want to show that $$\pr[|S_{unobs} - M/2| > \gamma' D] > p.$$ Let's standardize the variable: $Z_M = (S_{unobs} - \mu)/\sigma$. The failure probability then becomes $$\pr[|Z_M| > \gamma'D/\sigma].$$ Since $M \geq D^2/2$, the standard deviation $\sigma \geq D/(2\sqrt{2})$, which means the normalized error threshold is bounded by a constant: $\gamma'D/\sigma \leq 2\sqrt{2}\gamma' =: c_\gamma$. Thus, the failure probability is at least $\pr[|Z_M| > c_\gamma]$.

By the Berry-Esseen theorem, the cumulative distribution function $F_M(z)$ of $Z_M$ is pointwise close to the standard normal CDF $\Phi(z)$, with $|F_M(z) - \Phi(z)| \leq C/\sqrt{M}$ for a constant $C$. This allows us to lower-bound the failure probability:
$$\pr[|Z_M| > c_\gamma] \geq 2(1 - \Phi(c_\gamma)) - \frac{2C}{\sqrt{M}}$$
For any target probability $p \in (0, 1)$, we can choose the constant $\gamma' > 0$ to be small enough such that $2(1 - \Phi(c_\gamma)) > p$ (since $\Phi(c_\gamma) \to 0.5$ as $\gamma' \to 0$). For a sufficiently large $D$, the error term $2C/\sqrt{M}$ becomes negligible, guaranteeing the failure probability $\geq p$. This completes the proof.
\end{proof}

\input{text/upper_bound_analysis_new}

%% file: floats/tab_notations.tex
\begin{table}[H]
  \centering
  \caption{Frequently used notations in this paper.}
  \label{tab:notations}
  \vspace{-1em}
  % --- CONTROL 1: ROW SPACING ---
  \renewcommand{\arraystretch}{1.1}
  % --- CONTROL 2: COLUMN WIDTH ---
  \begin{tabularx}{\textwidth}{p{2.8cm} X}
  \toprule
  \textbf{Notation} & \textbf{Description} \\
  \midrule
  
  % --- SECTION 1 ---
  $G = (V_G, E_E)$ & Underlying directed graph with node set $V$, edge set $E$; or simply $(V, E)$ \\
  
  $n, m$ & Number of nodes, edges in $G$ \\
  
  $\alpha$ & Constant decay factor defining PageRank and PPR; $\alpha \in (0, 1)$\\

  $p$ & Constant failure probability bound\\
  
  $\delta, c$ & Relative error setting threshold; constant relative error parameter \\
  
  $\epsilon$ & Absolute error setting parameter \\
  
  $\pi(s, t),\pi_G(s,t)$ & Personalized PageRank scores between nodes $s, t$ in $G$\\

  \query{ADJ}, \query{DEG}, \query{JUMP} & Queries in arc-centric model\\

  \midrule
  
  $e = (u, v)$ & Directed edge from node $u$ to $v$ \\
  
  $d_{out}^G(u), d_{in}^G(u)$ & Out-degree, in-degree of node $u$ in $G$ \\

  $V_{\ge \delta}, X_{\geq \delta}$ & Subset of $V$ with PPR scores $\ge \delta$\\

  $V_{\mathrm{out}}, X_{\mathrm{out}}$ & Candidate node set output by algorithm \algoa\\

  $\overleftrightarrow{p_G}(u,v)$ & Set of all directed paths between node pair $(u,v)$ in $G$\\
  
  $\footnotesize\overleftrightarrow{p}, w_G({\footnotesize\overleftrightarrow{p}})$ & Path in $G$; path weight regarding $\alpha$-decay random walk\\

  $\boundary{X}$ & Nodes in $V \setminus X$ having out-edge to set $X$\\

  \midrule

  $\mathbf{{S}_r}(u)$ & Weight in Lemma \ref{lemma:newdecom}; $\mathbf{{S}_r}(u)=\frac{1}{\alpha}{{\pi}_G(s,u)d_{\mathrm{out}}^{G_X}(u)}/{d_{\mathrm{out}}^{G}(u)}$\\

  $\hat\pi_G(s,v)$ & Estimated PPR scores\\

  $\mathbf{{\hat S}_r}(u)$ & Empirical weight in Lemma \ref{lemma:newdecom}; $\mathbf{{\hat S}_r}(u)= \frac{1}{\alpha}{\hat{\pi}_G(s,u)d_{\mathrm{out}}^{G_X}(u)}/{d_{\mathrm{out}}^{G}(u)}$\\

  $\texttt{SUM}(\mathbf{\hat{S}_r}) $ & Total score summation; $\texttt{SUM}(\mathbf{\hat{S}_r}) = \sum_{u\in\boundary{X}}\mathbf{{\hat S}_r}(u) $\\

  $\mathbf{\hat{\bar{S}}_r}$ & Normalized weight distribution; $\mathbf{\hat{\bar{S}}_r}:=\mathbf{\hat{S}_r}/\texttt{SUM}(\mathbf{\hat{S}_r})$\\

  $c_r, p_r,\delta_r$ & Parameters in Algorithm \algorr\\

  $N_r, H[v]$ & Total walk count in \algorr; walks ending at node $v$\\

  $X_r$ & Remaining target (non-estimated) subset after round $r$\\
  
  $V_r$ & Newly estimated subset in round $r$; $V_r = X_{r-1} \setminus X_r$\\

  $R$ & Round bound in \algor; last execution round if context clear\\

  $\Delta$ & Constant shrink factor in \algor; $\Delta>1$\\

  \midrule

  $ \mathcal{U}(n, D, d)$ & Hard-instance multigraph family with parameters $D, d$\\

  $\mathrm{U}, \mathrm{U}(r)$ & (Multi)graph in family $ \mathcal{U}(n, D, d)$; its $r$-padded instance\\

  $\Sigma(n, D, d)$ & Distribution over family $ \mathcal{U}$ \\

  $\mathrm{SP}\in \binom{X}{2}$ & Split of set $X$ into $X_1, X_2$; generator part for $\Sigma(n, D, d)$\\

  $\mathrm{MUL}(\mathrm{U})$ & Edge multiplicity of multigraph $\mathrm{U}$ (Def. \ref{def:multigraph})\\

  $\mathcal{A}_{(\cdot)}$ & Algorithms in Lemmas \ref{lemma:PPR_lift}, \ref{thm:algo_adapt_tot} \\

  $e_t$ & Event in $\Sigma(n, D, d)$: $k^{\text{th}}$ out-edge of $y_i$ connects to $k'^{\text{th}}$ in-edge of $x_j$\\

  \midrule
% --- APPENDIX / ADVANCED NOTATION ---
  $\tau$ & Constant for geometric sequence $\delta_0, \delta_0\tau, \delta_0\tau^2, \dots$ (Appendix ~\ref{sec:adaptive_thre})\\

  $K_{\mathrm{RW}}, K_{\mathrm{err}}, K_{\mathrm{prob}}$ & Constants in \newalgor{} related to $N_r, c_r, p_r$\\

  $\mathbf{1}_A$ & Indicator function of event $A$\\
  
  $\mathrm{All}_{r-1}$ & Random variable for information available before round~$r$ estimation\\
  
  $ \hat\pi_G^{(r)}(s,t)$ & Estimate of $\pi_G(s,t)$ based on prior estimates (Eq. ~\eqref{eq:def_coherent_estimation})\\

  $D_r(t)$ & Martingale-difference sequence of estimation error at node $t$(Eq. ~\eqref{eq:def_diff_free})\\
  
  \bottomrule
  
  \end{tabularx}
\end{table}

%% file: floats/fig_multiple_edges.tex
\begin{figure}[!t]
\includegraphics[width=0.78\textwidth]{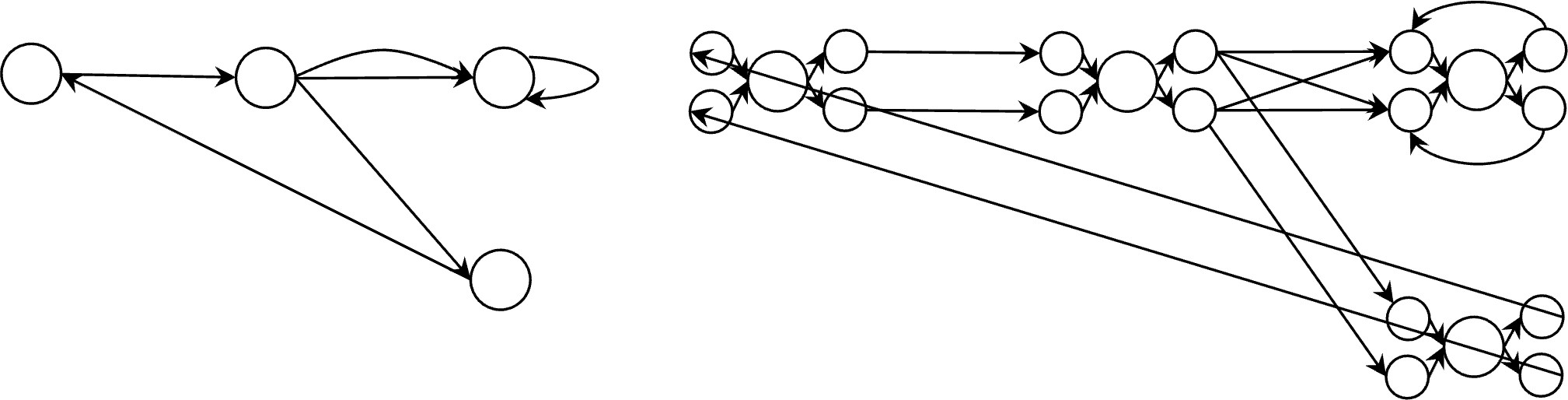}
\centering
\caption{Illustration of a \(2\)-Lift.} 
\label{fig:multiple_edges}
%\vspace{-0.5ex}
\end{figure}

%% file: text/upper_bound_analysis_new.tex
\section{\algor{}: Reducing $\log$ Factors}
\label{sec:proof:log}

In the main body, we proved that \algor{} achieves a complexity of $\bigo(m + n \log n \log^3(1/(n\delta)))$ for SSPPR-R. While this already approaches optimality in dense graph settings, the logarithmic factors admit further tightening. In this appendix, we refine the complexity analysis and introduce two auxiliary techniques—adaptive threshold scheduling and hyperparameter tuning—to optimize walk usage and variance control. Importantly, these refinements do not modify the algorithmic framework, which continues to rely on the same path-based decomposition and iterative refinement process. Rather, they provide a sharper accounting of sampling behavior across rounds, yielding the improved bound stated in Theorem~\ref{thm:ssppr-r-upper}.

\begin{theorem}[Improved Upper Bound for SSPPR-R]\label{thm:ssppr-r-upper-refined} 
\newalgor{} (Algorithm~\ref{algo:dist_walk_adjusted}) answers the SSPPR-R query within
$
\bigo\Bigl(m + n \log n R\Bigr)
$
queries in expectation, where $R=\log\Bigl(\frac{\log n}{m\delta}\Bigr)$.
\end{theorem}

As noted at the end of Section~\ref{sec:distwalks}, the tighter complexity bound follows from two refinements already introduced in the main text: (i) adaptive threshold selection via a doubling schedule, which amortizes the per-round cost and avoids dependence on the worst-case $1/\delta_r$; and (ii) a global error analysis that treats the estimator jointly across rounds rather than applying round-by-round worst-case bounds, yielding tighter variance. We next introduce the details of these two techniques.
 
\subsection{Adaptive Threshold Determination}\label{sec:adaptive_thre}
The standard \algor{} (Line~6) employs a conservative static threshold 
$\delta_r = \delta_0 \in \Theta(1/n)$, which guarantees that the remaining PPR mass
$\sum_{x \in X_r}\pi_G(s, x)$ shrinks by a factor of $\Delta$ per round.
To optimize this, we implement an \emph{adaptive mechanism} that selects
a near-maximal $\delta_r$ satisfying the shrinking requirement. This is
achieved by iterating through a geometric progression
$\delta_r = \tau^k\delta_0$ (for $k=0, 1, \dots$ and fixed $\tau \in (0,1)$)
until a convergence condition is met.

The adaptive procedure replaces the fixed assignment with the following
steps:
\begin{enumerate}[label=\Alph*.]
    \item Fix a constant $\tau\in (0,1)$, and rescale $c$ as
    \[
        c := \min\bigl(c, \tfrac{1}{8\Delta}\bigr).
    \]
    \item In each round of \newalgor{}, we first preprocess the alias table for sources and edges in $G_X$.
We then iteratively invoke
    $
       \mathsf{\newalgorr{}}(G, \allowbreak G_X, X_{r-1}, \hat \pi_G(s, \cdot), \alpha, c_r, \delta_r, p_r)
    $
    with fixed $c_r$, $p_r$ (established in the next part on accuracy), and a threshold $\delta_r$ taken from the geometric sequence.
    $\frac{R}{n}, \frac{R\tau}{n}, \frac{R\tau^2}{n}, \dots$. We start from  $\frac{R}{n}$ to guarantee that the threshold is small enough while keeping the total work bounded: even in the extreme case, if every round ends with $\delta_r=R/n$,  $\sum \frac{1}{\delta_r}$ still scales as $\bigo(n)$.
    \item During the execution of \newalgorr{}, we track the distribution of
    walk terminations $H[\cdot]$. In particular, we compute the
    proportion of walks landing in the non-estimated subset
    $X \setminus X_{\mathrm{out}}$:
    $$  \mathrm{Ratio}  =\frac{\sum_{x\in X\setminus X_{\mathrm{out}}}H[x]}{\sum_{x\in X}H[x]}=\frac{\sum_{x\in X_{r}}H[x]}{\sum_{x\in X_{r-1}}H[x]}.
    $$
    \item The search for $\delta_r$ halts when $\mathrm{Ratio}$ drops
    below $1/(2\Delta)$. We then retain the resulting
    $X_{\mathrm{out}}$ (equivalently, $V_r$) and the updated estimate
    $\hat{\pi}_G(s,\cdot)$.
    \item Whenever \(\texttt{SUM}(\mathbf{\hat{S}_r}) \le \frac{\log n}{\delta m}\), we switch to the final round: we set \(c_r := c/4\) and \(\delta_r := \frac{\log n}{\delta m}\), run \algorr{} once more, and then terminate the algorithm. Since \(\texttt{SUM}(\mathbf{\hat{S}_r})\) decreases by at least a factor \(\Delta>1\) in each preceding round, this early stopping condition is met within at most $R \le \log\Bigl(\frac{\log n}{\delta m}\Bigr)$ rounds.
\end{enumerate}

\input{floats/algo_round_est_simple}

\input{floats/algo_dist_walk_adjusted}
We now turn to the complexity. Note that, although we must perform supporting operations such as
initializing \(H[\cdot]\), storing \(I_r\), and computing the ratio
\(\frac{\sum_{x\in X\setminus X_{\mathrm{out}}} H[x]}
        {\sum_{x\in X} H[x]}\)
for each \(\delta_r\) in every round, each such operation costs \(\bigo(n)\).
Since there are only \(\bigo(R \log R)\) candidate thresholds in total, these
chores contribute an additional \(\bigo(n R \log R)\) time, which is
dominated by our expected complexity $\bigo(m+n R\log(n))$.
For \newalgorr{} and \newalgor{}, analogously to Lemma~\ref{lemma:one_round}, we obtain the following preliminary bounds. 
\begin{lemma}\label{lemma:algorrr_raw}
The time complexity of algorithm \newalgor{} scales as 
$$
\bigo\left(m+nR\log(R)+\sum_{r=2}^{R}\frac{\log(n/p_r)}{c_r^2 \delta_r}\right) = \bigo\left(m+nR\log(R) + R\log(n)\sum_{r=2}^{R}\frac{1}{\delta_r} \right).
$$
For accuracy, we have\begin{enumerate}[label=(\roman*)]
    \item In round $1$ (Line $1$ to $3$), $\pr_{u \in V_1}[|\hat\pi_G(s,u)/\pi_G(s,u)-1| \geq \frac{c}{4}] \leq \frac{p_f}{n}$.
    \item For large enough absolute constants $K_{err}, K_{RW}, K_{prob}$, assume up to round $r-1$, $|\hat\pi_G(s,u)/\pi_G(s,u)-1| \leq \frac{c}{2}, u \in X\setminus X_{r-1}$, we have with probability $1-\frac{p_f}{n}$ for each execution with $\delta_r$ at Line $13$, the returned $\hat\pi_G(s,t)$ satisfies:\[
    \pr_{t \in V_r}\left[|\hat\pi_G(s,t)/\pi_G(s,t)-1| \geq \frac{c}{2}\right] \leq \frac{p_f}{n}\]
\end{enumerate}
\end{lemma}

The expected running time is bounded by the number of simulated walks (excluding the potential early stop round with $\bigo(m)$ cost), and this number scales as
\[
\Theta\!\left(\frac{\log(n/p_r)}{c_r^2 \delta_r}\right),
\]
which is proportional to $1/\delta_r$. Since the guarantees for accuracy
and overall failure probability will be established in the subsequent
analysis, we focus solely on the quantity $\sum_r 1/\delta_r$, as in the
next claim.  

\begin{claim}\label{claim:thre_sum}
Assuming that the PPR estimates for all nodes $x$ in the current output
set $X_{\mathrm{out}}$ are already relatively accurate (i.e.,
$\hat \pi_G(s, x)/\pi_G(s, x) \in [1-c, 1+c]$), the total sum of inverse
thresholds $\sum_r 1/\delta_r$ is bounded by $\bigo(n)$ with high
probability, for $n$ sufficiently large.
\end{claim}

\subsection{Proof of Claim~\ref{claim:thre_sum}}
The intuition underlying this $\bigo(n)$ bound is that one can upper bound
$1/\delta_r$ by the number of nodes estimated in that round. In an
idealized setting, assume that all estimated nodes have PPR scores equal
to $\delta_r\sum_{x \in X_{r-1}}\pi_G(s, x)$, and furthermore the shrinkage is exactly $1/\Delta$, i.e.,
\[
\sum_{x \in X_{r}}\pi_G(s, x) = \frac{1}{\Delta}\sum_{x \in X_{r-1}}\pi_G(s, x).
\]
Together with the relation  $X_{r-1} = X_r \cup V_r$, we can establish that
\begin{align}
\sum_{t\in V_r}\pi_G(s,t) =&\;\delta_r \sum_{x \in X_{r-1}}\pi_G(s, x) \cdot |V_r|
 = \left(1-\frac{1}{\Delta}\right)\sum_{x \in X_{r-1}}\pi_G(s, x) \\
\end{align}
and
\begin{align}
\frac{1}{\delta_r} & \leq |V_r|\cdot \frac{1}{1-\frac{1}{\Delta}};\qquad \frac{1}{\delta_r}\in \bigo(|V_r|).
\end{align}
To establish a formal proof, however, we must account for the fact that the estimated nodes may have varying PPR scores within a single round and that the shrinkage can be arbitrarily small. We begin the proof by characterizing the shrinkage
$\sum_{x \in X_r}\pi_G(s, x)$ in each round, and then split into two
cases according to the rate of decrease. In the proof, we will assume without loss of generality that $\delta_r < R/n$. For those rounds with the first choice of $\delta_r$ in Step ~(B), the sum of their costs is bounded by $R \cdot n/R$.

\mypara{ Shrinkage of the Remaining Scores: The Upper Bound.}
Although we target that
$\sum_{x \in X_r}\pi_G(s, x)$ shrinks exponentially across rounds, we
have not proven this yet based on the new algorithm. Recall that, in our adjusted algorithm
(Step~(D)), we terminate the search for $\delta_r$ in round $r$ once
the ratio falls below $1/(2\Delta)$, i.e.,
\[
\frac{\sum_{x\in X_r} H[x]}{\sum_{x\in X_{r-1}} H[x]} \leq \frac{1}{2\Delta}.
\]Intuitively, this empirical ratio can be used
to estimate the quantity
\[
\frac{\sum_{x \in X_r}\pi_G(s, x)}
         {\sum_{x \in X_{r-1}}\pi_G(s, x)}. 
\]
We now claim and will prove that 
\[
    \sum_{x \in X_r}\pi_G(s, x)
    \;\leq\; \frac{1}{\Delta}\sum_{x \in X_{r-1}}\pi_G(s, x).
\]
Recall that in \newalgorr{} we sample $\alpha$-RW from the source distribution
$\mathbf{\hat{\bar{S}}_r}$. By Lemma~\ref{lemma:newdecom} and our
assumption in the claim on the relative accuracy, we have
\begin{equation}\label{eq:by_assum_claim}
  \frac{
  \sum_{u\in \boundary{X}}\mathbf{\hat{\bar{S}}}_r(u)\,\pi_{G_X}(u, x)
}{
  \pi_G(s, x)/\texttt{SUM}(\mathbf{\hat{S}}_r)
} \in [1-c,1+c].
\end{equation}
For the denominator $\sum_{x\in X_{r-1}} H[x]$, we can then interpret it as a sampling procedure from the \textit{Bernoulli} distribution with probability
\[
\sum_{x\in X_{r-1}}\sum_{u\in \boundary{X}}\mathbf{\hat{\bar{S}}_r}\,\pi_{G_X}(u, x).
\]
With Equation~\eqref{eq:by_assum_claim}, we also have
\begin{equation}\label{eq:appendix_eval_err}
\frac{
  \sum_{x\in X_{r-1}}\sum_{u\in \boundary{X}}\mathbf{\hat{\bar{S}}}_r(u)\,\pi_{G_X}(u, x)
}{
  \sum_{x\in X_{r-1}}\pi_G(s, x) / \texttt{SUM}(\mathbf{\hat{S}}_r)
} \in [1-c,1+c].
\end{equation}
By Lemma \ref{lemma:newdecom}, it holds that the denominator $$\sum_{x \in X_{r-1}}\pi_G(s, x)/\texttt{SUM}(\mathbf{\hat{S}_r}) \in \Theta(1).$$ Thus, the probability of our Bernoulli distribution scales as $\Theta(1)$. Since we have performed $\Omega(n/R)\subset \tilde{\Omega}(n)$ independent samples for this estimation, Chernoff's bound guarantees that, for any $K>1$ and $n$ sufficiently large
$$\pr\left[|\frac{\sum_{x\in X_{r-1}} H[x]}{N_r} - \sum_{x\in X_{r-1}}\sum_{u\in \boundary{X}}\mathbf{\hat{\bar{S}}_r}\,\pi_{G_X}(u, x)| \geq \frac{1}{n^{1/4}}\right] \leq \frac{1}{n^K}.
$$
Equivalently with high probability, it holds that
\[
\frac{\sum_{x\in X_{r-1}} H[x]}{N_r\sum_{x \in X_{r-1}}\pi_G(s, x)/\texttt{SUM}(\mathbf{\hat{S}_r})}
\in \left[(1-c)(1-\frac{1}{n^{1/4}}), (1+c)(1+\frac{1}{n^{1/4}})\right].
\]
However, we do not know the scale of the shrunken total scores $\sum_{x \in X_{r}}\pi_G(s, x)/\texttt{SUM}(\mathbf{\hat{S}_r})$. Indeed, our role is to show this expected value is relatively small. Therefore, we prove this by contradiction and assume that with probability larger than $1/n^{K-1}$, it holds that
\begin{equation}\label{eq:simple_contra}
    \sum_{x \in X_{r}}\pi_G(s, x) > \frac{1}{\Delta}\sum_{x \in X_{r-1}}\pi_G(s, x).
\end{equation}
As such, we can use a similar estimate for the numerator, as
$$\pr\left[|\frac{\sum_{x\in X_{r}} H[x]}{N_r} - \sum_{x\in X_{r}}\sum_{u\in \boundary{X}}\mathbf{\hat{\bar{S}}_r}\,\pi_{G_X}(u, x)| \geq \frac{1}{n^{1/4}}\right] \leq \frac{1}{n^K}.
$$
Then, with probability at least $1-\frac{2}{n^K}$, it holds that
\begin{align*}
& \quad\; \frac{\sum_{v \in X_r}\pi_G(s, v)}
             {\sum_{x \in X_{r-1}}\pi_G(s, x)}\\
&\leq  \frac{(1+c)(1+\frac{1}{n^{1/4}})\sum_{v \in X_r} H[v]}
            {(1-c)\bigl(1-\frac{1}{n^{1/4}}\bigr)\sum_{x \in X_{r-1}}H[x]}\\
\left(n\text{ sufficiently large}\right)&\leq (1+2c)  \frac{\sum_{v\in X_r} H[v]}
            {\sum_{v\in X_{r-1}} H[v]}\\
\left(\text{Terminating}\right)& \leq (1+2c)\frac{1}{2\Delta}
\end{align*}
Then it holds that $$ \frac{\sum_{x \in X_r}\pi_G(s, x)}
         {\sum_{x \in X_{r-1}}\pi_G(s, x)} \leq \frac{1}{\Delta},$$ which contradicts Eq.~\eqref{eq:simple_contra}.

\mypara{Fast-shrink case.} Although this natural stopping rule gives a clean upper bound $\frac{1}{\Delta}$ on the
shrinkage of the remaining scores, there is no corresponding lower bound.
We first discuss the case that
\[
    \frac{\sum_{x \in X_r}\pi_G(s, x)}
         {\sum_{x \in X_{r-1}}\pi_G(s, x)}
    < \frac{1}{\Delta'},
\]
where the fixed constant $\Delta'$ satisfying
\begin{equation}\label{eq:def_delta_prime}
    1/(32\Delta)<1/\Delta' < 1/(16\Delta).
\end{equation}
In this regime, the total remaining PPR mass drops very rapidly. However, our strategy of exponentially adapting the threshold ensures that we do not overshoot the target excessively. Specifically, the fact that the immediately preceding candidate $\delta_r/\tau$ failed to satisfy the stopping condition implies that the mass drop is bounded from below, which allows us to derive an upper bound on the inverse threshold $1/\delta_r$.
More precisely, given that $\mathrm{Ratio} > 1/(2\Delta)$ holds during the execution with threshold $\delta_r/\tau$, we can employ a similar proof-by-contradiction argument as in the proof of the upper bound on the shrinkage. Specifically, we assume it holds with probability at least $1-n^{K-1}$:
\begin{equation}\label{eq:simple_by_contra2}
    \frac{\sum_{x \in X'_r}\pi_G(s, x)}
         {\sum_{x \in X_{r-1}}\pi_G(s, x)} < \frac{1}{8\Delta},
\end{equation}
where $X'_r$ is the result obtained by adopting the threshold $\delta_r/\tau$.  In this setting, $\sum_{x \in V'_r}\pi_G(s, x)$ is large enough. Then, setting $V_r', H', N_r'$ to be the variables with threshold $\delta_r/\tau$,  we have 
$$\sum_{x \in V'_r}\pi_G(s, x)/\texttt{SUM}(\mathbf{\hat{S}_r})  \geq (1-1/(8\Delta))\sum_{x \in X_{r-1}}\pi_G(s, x)/\texttt{SUM}(\mathbf{\hat{S}_r}) \in \Theta(1).
$$
By standard Chernoff's bound, we obtain that
$$\pr\left[|\frac{\sum_{x\in V'_{r}} H'[x]}{N'_r} - \sum_{x\in V'_{r}}\sum_{u\in \boundary{X}}\mathbf{\hat{\bar{S}}_r}\,\pi_{G_X}(u, x)| \geq \frac{1}{n^{1/4}}\right] \leq \frac{1}{n^K}.
$$
Combining above arguments, it then holds with probability at least $1-2/n^K$ that 
\begin{align*}
    &\quad\;\frac{\sum_{x \in X'_r}\pi_G(s, x)}
         {\sum_{x \in X_{r-1}}\pi_G(s, x)}  \\
         \bigl(X'_r \sqcup V'_{r} = X_{r-1}\bigr)&= 1 - \frac{\sum_{x \in V'_r}\pi_G(s, x)}
         {\sum_{x \in X_{r-1}}\pi_G(s, x)}\\
    \bigl(n \text{ is sufficiently large})&\geq 1 - \frac{1}{1-2c}\left(
        1 - \frac{\sum_{x\in V'_r} H'[x]}
                 {\sum_{x\in X_{r-1}} H'[x]}
    \right)\\
    (c \leq 1/(8\Delta),\ \Delta>1)&\geq 1 -\frac{1}{1-2c}\left(1 - \frac{1}{2\Delta}\right)\geq \frac{1}{8\Delta}.\label{eq:lower_shrink}
\end{align*}
This establishes that the coarser threshold $\delta_r/\tau$ would keep a significant remaining total score. Since the selected threshold $\delta_r$ satisfies the fast-shrink assumption (implying $\sum_{x \in X_r}\pi_G(s, x)$ is small), the probability mass captured in the incremental set $V_r \setminus V'_r = X'_r \setminus X_r$ must be substantial.

Notice that for those $v \in V_r \setminus V_r'$, i.e. nodes that have not yet been selected by the
threshold $\delta_r/\tau$ but selected when using $\delta_r$, together with the accuracy argument like in Eq. \eqref{eq:appendix_eval_err}, we have 
\[
\pi_G(s,v) \leq \frac{\delta_r}{\tau(1-c)}\texttt{SUM}(\mathbf{\hat{S}_r})
\leq \frac{\delta_r}{\tau(1-c)\alpha(1-\alpha)}\sum_{x\in X_{r-1}}\pi_G(s, x),
\]
by Lemma~\ref{lemma:one_round}. Hence we can derive that 
\begin{align*}
|V_r|\cdot \frac{\delta_r}{\tau(1-c)\alpha(1-\alpha)}
   \sum_{x \in X_{r-1}}\pi_G(s, x)
  &\geq \sum_{x \in X'_r}\pi_G(s, x)
         - \sum_{x \in X_r}\pi_G(s, x).
\end{align*}

The right side of the above equation is exactly the PPR sum in
\(V_r \setminus V_r' = X'_r \setminus X_r\). By applying the two previous
estimations, we then obtain that

\begin{align*}
|V_r|\cdot \frac{\delta_r}{\tau(1-c)\alpha(1-\alpha)}
   \sum_{x \in X_{r-1}}\pi_G(s, x)
  &\geq  \left(\frac{1}{8\Delta} - \frac{1}{\Delta'}\right)
          \sum_{x \in X_{r-1}}\pi_G(s, x).
\end{align*}
Equivalently, we have
\[
 \frac{1}{\delta_r} \leq \frac{16\Delta}{\tau(1-c)\alpha(1-\alpha)}|V_r|.
\]
Therefore, the total cost over all fast-shrink rounds satisfies
\[
\sum_{r} \frac{1}{\delta_r} \in \bigo\Bigl(\sum_r |V_r|\Bigr) = \bigo(n).
\]

\mypara{Slow-shrink case.}
In this case we have
\[
    \frac{1}{\Delta} \geq \frac{\sum_{x \in X_r}\pi_G(s, x)}
         {\sum_{x \in X_{r-1}}\pi_G(s, x)}
    \geq \frac{1}{\Delta'}.
\]
This case can be bounded in terms of $|V_{r+1}|$. Formally, for $r<R$, since nodes with larger scores have been filtered in round $r$, at the end of round $r$, each node $x$ in $X_{r}$ has PPR score at most 
\[
\pi_G(s,x) \leq \frac{\delta_r}{1-c}\mathtt{SUM}(\mathbf{\hat{S}_r}) \leq
\frac{\delta_r}{(1-c)\alpha(1-\alpha)}\sum_{x\in X_{r-1}}\pi_G(s, x)
\]
by Lemma~\ref{lemma:one_round}. Then it holds that 
\begin{align*}
&\quad\;|V_{r+1}| \cdot \frac{\delta_r}{(1-c)\alpha(1-\alpha)}\sum_{x\in X_{r-1}}\pi_G(s, x) \\
&\geq \sum_{x\in V_{r+1}}\pi_G(s, x)\\
&=\sum_{x\in X_{r}}\pi_G(s, x) -  \sum_{x\in X_{r+1}}\pi_G(s, x)\\
\end{align*}
Since in round $r+1$,
 $\sum_{x\in X_{r+1}}\pi_G(s, x)\big/\sum_{x\in X_{r}}\pi_G(s, x) \leq \tfrac{1}{\Delta}$, it then follows that
\begin{align*}
&\quad\;\sum_{x\in X_{r}}\pi_G(s, x) -  \sum_{x\in X_{r+1}}\pi_G(s, x)\\
&\geq \left(1-\frac{1}{\Delta}\right) \sum_{x\in X_{r}}\pi_G(s, x)
\\
&\geq \frac{1}{\Delta'}\left(1-\frac{1}{\Delta}\right)\sum_{x\in X_{r-1}}\pi_G(s, x).
\end{align*}
Equivalently, we have
\[
 \frac{1}{\delta_r} \leq \frac{32\Delta}{(1-c)\alpha(1-\alpha)\left(1-\frac{1}{\Delta}\right)}|V_{r+1}|.
\]
Therefore, the total cost among slow-shrink rounds satisfies
\[
\sum_{r<R} \frac{1}{\delta_r} \in \bigo\Bigl(\sum_r |V_r|\Bigr) = \bigo(n).
\]
Finally, we consider the case where the last round (the last round before early stop round) is a \textit{slow-shrink round}. We only
need to ensure that $1/\delta_R \in \bigo(n)$. Similarly for any $v\in X_R$ we have that
\[
\pi_G(s, v) \leq \frac{1+c}{(1-\alpha)\alpha}\delta_R\sum_{x\in X_{R-1}}\pi_G(s, x).
\]
Hence, it holds that 
\[
\sum_{x\in X_{R}}\pi_G(s, x)
\leq  |X_R|\frac{1+c}{(1-\alpha)\alpha}\delta_R\sum_{x\in X_{R-1}}\pi_G(s, x).
\]
Comparing with our assumption
\[
\frac{\sum_{x \in X_R}\pi_G(s, x)}
         {\sum_{x \in X_{R-1}}\pi_G(s, x)}
    \geq \frac{1}{\Delta'},
\]
we obtain $$\frac{1+c}{(1-\alpha)\alpha}\delta_R|X_R| \geq 1/\Delta'.$$ 
Therefore,
$1/\delta_R \in \bigo(|X_R|) \in \bigo(n)$ holds and we complete the proof.
\qed

\subsection{Global Error Analysis}
In the previous subsection we refined the complexity analysis by reducing the total sampling factor
\(\sum_{r=1}^{R} 1/\delta_r\) from \(nR\) to \(n\), leading to a $\bigo(m+n\log(n)R)$ bound as expected. However, we have not shown the accuracy and correctness of our algorithm.

In this subsection we revisit the accuracy guarantees and, from a global perspective, show how to control the accumulated estimation error while setting $c_r \in \Omega(1/\sqrt{R})$ rather than the direct $c_r \in \Omega(1/R)$ setting in the original \algor.
In short, our goal is to show that for all nodes \(t\) with \(\pi_G(s,t)\ge \delta\) we can ensure
\[
  \bigl|\hat\pi_G(s,t)/\pi_G(s,t) - 1\bigr| \;\le\; c
\]
with high probability, while using parameters \(c_r \in \Omega(1/\sqrt{R})\) in each round \(r\).

Since the final error aggregates contributions from all rounds, we decompose the total error
\(\hat \pi_G(s,t) - \pi_G(s,t)\) into a sequence of sub-Gaussian errors, i.e., random errors with Gaussian-type tail decay.
We carefully bound the tail of each error term by conditioning on high-probability “good’’ events.
Moreover, although the per-round errors may be dependent, they admit uniform bounds conditional on the past, so we can organize them into a martingale-difference sequence. Intuitively, it makes our total error behaves like the sum of $R$ independent Gaussian noise terms with standard deviation $\bigo(1/\sqrt{\log(1/p_f)R})$, leading to a total deviation $\bigo(1/\sqrt{\log(1/p_f)})$.
These are standard tools for establishing concentration in multi-round randomized algorithms, and we keep the presentation self-contained.

\mypara{Sub-Gaussian random variables.}
We first recall a basic notion of sub-Gaussian random variables.
\begin{definition}\label{def:subG}
    A real-valued random variable \(D\) is \emph{sub-Gaussian with parameter} \(\sigma > 0\) if
\begin{equation}\label{eq:subgaussian-def}
   \pr\bigl(|D| \geq  a\bigr) \;\leq\; 2\exp\bigl(-a^2/\sigma^2\bigr)
   \qquad\text{for all } a>0.
\end{equation}
Our definition does not require \(D\) to be centered.
\end{definition}
It is well known that convex combinations of sub-Gaussian random variables are again sub-Gaussian; see, e.g., \cite[Section~2.5]{vershynin2018high}.

\begin{fact}\label{fact:subgaussian-convex-combination}
There exists an absolute constant \(C_1 > 0\) such that the following holds.
Let \(D_i\) be sub-Gaussian random variables with parameters \(\sigma_i\), and let weights \(w_i \geq 0\) satisfy \(\sum_i w_i = 1\).
Then the weighted sum \(\sum_i w_i D_i\) is sub-Gaussian with parameter at most \(C_1 \max_i \sigma_i\).
\end{fact}

Although our definition does not require sub-Gaussian random variables to be centered, the next lemma shows that centering preserves sub-Gaussianity up to absolute constants.

\begin{lemma}[cf.\ Lemma~2.6.8 in~\cite{vershynin2018high}]\label{lemma:centering-subgaussian}
There exists an absolute constant \(C_2 > 0\) such that if \(D\) is sub-Gaussian with parameter \(\sigma\), then \(D - \E[D]\) is also sub-Gaussian with parameter at most \(C_2 \sigma\).
\end{lemma}

\mypara{Martingales and an Azuma--Hoeffding type inequality.}
As explained at the beginning of this subsection, the total error can be written as a sum
$
  \sum_{r=1}^R D_r,
$
where the random variable \(D_r\) represents the contribution of round \(r\).
Each \(D_r\) may depend on the randomness revealed in previous rounds, such as the estimated subset of nodes \(\boundary{X}\), but, conditional on the past, it satisfies a uniform sub-Gaussian tail bound.

To formalize this, we reveal the randomness round by round and let \(\mathrm{All}_{r-1}\) denote all information available before the estimation in round \(r\).
For instance, \(\mathrm{All}_{r-1}\) determines the target set \(V_r\) and the current estimates \(\hat\pi_G(s,t)\) for all \(t \in\boundary{X}_{r-1}\).
Formally, the random variable sequence \((\mathrm{All}_r)_{r=0}^R\) plays the role of a filtration \((\mathcal{F}_r)_{r=0}^R\), but we avoid measure-theoretic notation and simply write \(\mathrm{All}_r\).

\begin{definition}[Martingale-difference sequence]\label{def:mds}
We say that \((D_r)_{r=1}^R\) is a \emph{martingale-difference sequence adapted to} \((\mathrm{All}_r)_{r=0}^R\) if, for every \(1 \le r \le R\):
\begin{enumerate}[label=(\roman*).]
  \item \(\mathrm{All}_r\) determines \(D_r\) (that is, \(D_r\) is a function of \(\mathrm{All}_r\));
  \item \(\E\bigl[D_r \,\big|\, \mathrm{All}_0, \dots, \mathrm{All}_{r-1}\bigr] = 0\);
\end{enumerate}
\end{definition}

In this language, the sequence of error terms \((D_r)_{r=1}^R\) that we obtain from our algorithm will form a martingale-difference sequence adapted to \((\mathrm{All}_r)_{r=0}^R\), which allows us to apply a standard Azuma--Hoeffding type inequality. See \cite{shamir2011variant} for an elementary proof.

\begin{lemma}[Azuma for sub-Gaussian increments]\label{lemma:azuma-subgaussian}
There exists an absolute constant \(C_3 > 0\) with the following property.
Let \((D_r)_{r=1}^R\) be a martingale-difference sequence adapted to \((\mathrm{All}_r)_{r=0}^R\).
Suppose that for each \(1 \le r \le R\) and all \(a>0\),
\[
  \pr\bigl(|D_r| \geq  a \,\big|\, \mathrm{All}_0, \dots, \mathrm{All}_{r-1}\bigr)
    \;\le\; 2\exp\bigl(-a^2/\sigma^2\bigr).
\]
Then, 
for every \(q \in (0,1)\), with probability at least \(1-q\),
\[
  \Bigl|\sum_{r=1}^R D_r\Bigr|
    \;\le\; \sqrt{\,C_3 R\sigma^2 \log(1/q)}.
\]
\end{lemma}

This lemma is the only martingale tool we will need: it shows that the total error behaves like the sum of \(R\) sub-Gaussian increments, leading to a global error bound that scales as \(\sqrt{R}\) rather than \(R\).

\mypara{Notations and the estimation sequence $\hat\pi_G^{(r)}(s,t)$.}
Recall that $\mathrm{All}_{r-1}$ denotes all information available before the estimation step in round~$r$, i.e., the internal state of our algorithm before walking (cf.\ Line $13$ in Algorithm\ref{algo:dist_walk_adjusted}).
In particular, the collection $(\mathrm{All}_0,\dots,\mathrm{All}_{r-1})$ determines the following data:
\begin{enumerate*}[label=(\roman*)]
  \item the set of not-yet-estimated nodes $X_{r-1}$, and the subset $V_r \subseteq X_{r-1}$ to be estimated in round~$r$;
  \item the parameters $c_r, \delta_r, p_r$;
  \item the previous estimates $\hat\pi_G(s,t)$ for all $t \in X \setminus X_{r-1}$;
  \item in particular, the current sample set $V_r$ and the induced subgraph $G_{X_r}$.
\end{enumerate*}

We now define a sequence of estimators $\hat\pi_G^{(r)}(s,t)$ that is convenient for the global error analysis, as it explicitly tracks how the estimate at each node $t$ evolves across rounds.
Conditioned on $\mathrm{All}_0,\dots,\mathrm{All}_{r-1}$, the “raw’’ estimator $\hat\pi_G(s,t)$ satisfies
\[
  \hat\pi_G(s,t) \;=\;
  \begin{cases}
    \text{not yet defined}, & t \in X_{r-1},\\[2pt]
    \text{fixed (already estimated)}, & t \in V \setminus X_{r-1}.
  \end{cases}
\]
Thus, to obtain a globally defined estimator at the end of round $r-1$, we extend $\hat\pi_G(s,t)$ to all $t \in V$ by setting
\begin{equation}\label{eq:def_coherent_estimation}
  \hat\pi_G^{(r-1)}(s,t) \;=\;
  \begin{cases}
    \displaystyle
    \sum_{u \in X \setminus X_r}
      \frac{\hat\pi_G(s,u)\,\pi_{G_X}(u,t)\,d^{G_X}_{\mathrm{out}}(u)}
           {\alpha\,d^{G}_{\mathrm{out}}(u)},
      & t \in X_{r-1},\\[10pt]
    \hat\pi_G(s,t), & t \in V \setminus X_{r-1}.
  \end{cases}
\end{equation}
We also set $\hat\pi_G^{(0)}(s,t) := \pi_G(s,t)$.
Heuristically, $\hat\pi_G^{(r-1)}(s,t)$ can be viewed as the mean of a random-walk–based estimator of $\pi_G(s,t)$ that uses the current PPR estimates.
To avoid confusion, from now on we work exclusively with the coherent sequence $(\hat\pi_G^{(r)}(s,t))_{r=0}^R$ and no longer use the bare notation $\hat\pi_G(s,t)$. We remark that although in this definition, we use $u\in X\setminus X_r$ rather than $\boundary{X_r}$ in Lemma \ref{lemma:newdecom}. This distinction is purely stylistic, since $\pi_{G_X}(u,t)=0$ (i.e. there is no path from $u$ to $t$ in $G_X$). Note that for every $u \in X \setminus X_r$ we have $\hat\pi_G(s,u) = \hat\pi_G^{(r-1)}(s,u)$, so the first case in~\eqref{eq:def_coherent_estimation} can be rewritten as
\begin{equation}\label{eq:def_coherent_estimation_1}
  \hat\pi_G^{(r-1)}(s,t)
  \;=\;
  \sum_{u \in X \setminus X_r}
    \frac{\hat\pi_G^{(r-1)}(s,u)\,\pi_{G_X}(u,t)\,d^{G_X}_{\mathrm{out}}(u)}
         {\alpha\,d^{G}_{\mathrm{out}}(u)},
  \qquad t \in X_{r-1}.
\end{equation}
Consequently, $(\mathrm{All}_0,\dots,\mathrm{All}_{r-1})$ determines $\hat\pi_G^{(r-1)}(s,t)$ for all $t \in V$.
As $r$ increases, the estimation error from earlier rounds, encoded in the values $\hat\pi_G^{(r-1)}(s,u)$, propagates to the nodes $t \in X_{r-1}$ through the recursion~\eqref{eq:def_coherent_estimation_1}.
In the next subsection we will see that the resulting error sequence fits automatically into our martingale framework.

\mypara{Per-round sampling error as a martingale difference.}
Fix a node $t \in V$.
We now study how the coherent estimator $\hat\pi_G^{(r)}(s,t)$ evolves across rounds and show that its per-round update forms a martingale-difference sequence (we will derive sub-Gaussian tail bounds for these increments in the next subsection). Without loss of generality, we assume that for all $ u\in V, \pi_G(s, u) >0$, or there is a path from $s$ to $u$.
Specifically, we define
\begin{equation}\label{eq:def_diff_free}
  D_r(t)
  \;:=\;
  \frac{\hat\pi_G^{(r)}(s,t) - \hat\pi_G^{(r-1)}(s,t)}{\pi_G(s,t)},
  \qquad 1 \le r \le R.
\end{equation}
To prove that $\bigl(D_r(t)\bigr)_{r=1}^R$ is a martingale-difference sequence adapted to
$\bigl(\mathrm{All}_r\bigr)_{r=0}^R$, it suffices to show
\[
\E\bigl[D_r(t)\,\big|\,\mathrm{All}_0,\dots,\mathrm{All}_{r-1}\bigr] = 0
  \qquad\text{for all } 1 \le r \le R.
\]
We first treat the case $t \in V_r$, i.e., $t$ is newly estimated in round~$r$, and then reduce the general case to this one.

\mypara{Newly Estimated Nodes.}
In round~$r$, \newalgorr{} launches $N_r$ independent $\alpha$-random walks on $G_{X_{r-1}}$, where the start node of each walk is drawn from the distribution $\hat{\bar S}$.
Conditioned on $\mathrm{All}_0,\dots,\mathrm{All}_{r-1}$, the probability that a single round-$r$ walk ends at $t$ is
\begin{equation*}
\begin{aligned}
  \pr\bigl[\text{a single round-$r$ walk ends at $t$}\,\bigm|\mathrm{All}_0,\dots,\mathrm{All}_{r-1}\bigr]
  &= \sum_{u \in X \setminus X_r} \mathbf{\hat{\bar S}}_r[u] \,\pi_{G_{X_{r-1}}}(u,t)\\
  &= \texttt{SUM}(\mathbf{\hat S}_r)^{-1} \,\hat\pi_G^{(r-1)}(s,t),
\end{aligned}
\end{equation*}
where we used the definition of $\hat{\bar S}$ and~\eqref{eq:def_coherent_estimation_1} for round $r-1$.

Let $H[t]$ denote the number of round-$r$ walks that end at node $t$.
Conditioned on $\mathrm{All}_0,\dots,\mathrm{All}_{r-1}$, we have
\[
  H[t] \sim \mathrm{Binomial}\Bigl(N_r,\,
    \texttt{SUM}(\mathbf{\hat S}_r)^{-1}\,\hat\pi_G^{(r-1)}(s,t)\Bigr).
\]
As in the algorithm, the new estimate of $\hat\pi_G(s,t)$ in round~$r$ is
\[
  \hat\pi_G^{(r)}(s,t)
  \;=\;
  \frac{H[t]}{N_r}\cdot \texttt{SUM}(\mathbf{\hat S}_r),
  \qquad t \in V_r.
\]
Therefore,
\begin{equation*}
\begin{aligned}
  \E\bigl[\hat\pi_G^{(r)}(s,t)\,\big|\,\mathrm{All}_0,\dots,\mathrm{All}_{r-1}\bigr]
  &= \frac{\texttt{SUM}(\mathbf{\hat S}_r)}{N_r}\,
     \E\bigl[H[t]\,\big|\,\mathrm{All}_0,\dots,\mathrm{All}_{r-1}\bigr]\\
  &= \frac{\texttt{SUM}(\mathbf{\hat S}_r)}{N_r}\,
     N_r \cdot
     \texttt{SUM}(\mathbf{\hat S}_r)^{-1}\,\hat\pi_G^{(r-1)}(s,t)\\
  &= \hat\pi_G^{(r-1)}(s,t).
\end{aligned}
\end{equation*}
Dividing by $\pi_G(s,t)$, we obtain
\[
  \E\bigl[D_r(t)\,\big|\,\mathrm{All}_0,\dots,\mathrm{All}_{r-1}\bigr]
  \;=\;
  \frac{1}{\pi_G(s,t)}
  \E\bigl[\hat\pi_G^{(r)}(s,t) - \hat\pi_G^{(r-1)}(s,t)\,\big|\,\mathrm{All}_0,\dots,\mathrm{All}_{r-1}\bigr]
  \;=\; 0,
\]
as desired.

\mypara{Previously Estimated or not Yet Estimated Nodes.}
If $t \in V \setminus X_{r-1}$, then $t$ has already been estimated before round~$r$, and by case~(ii) in~\eqref{eq:def_coherent_estimation} we have
\[
  \hat\pi_G^{(r)}(s,t) = \hat\pi_G^{(r-1)}(s,t),
\]
so $D_r(t) = 0$ deterministically and hence
\(
  \E\bigl[D_r(t)\,\big|\,\mathrm{All}_0,\dots,\mathrm{All}_{r-1}\bigr] = 0.
\)

It remains to consider $t \in X_{r-1}$.
Using the recursion~\eqref{eq:def_coherent_estimation_1} at rounds $r$ and $r-1$ and subtracting, we obtain
\begin{equation*}
\begin{aligned}
  \hat\pi_G^{(r)}(s,t) - \hat\pi_G^{(r-1)}(s,t)
  \;=\;
  \sum_{u \in X \setminus X_r}
    \frac{\pi_{G_X}(u,t)\,d^{G_X}_{\mathrm{out}}(u)}
         {\alpha\,d^{G}_{\mathrm{out}}(u)}
    \bigl(\hat\pi_G^{(r)}(s,u) - \hat\pi_G^{(r-1)}(s,u)\bigr).
\end{aligned}
\end{equation*}
We now have obtained that
\[
  \hat\pi_G^{(r)}(s,u) = \hat\pi_G^{(r-1)}(s,u)
  \qquad\text{for all } u \in X \setminus X_{r-1},
\]
and $X_{r-1} = X_r \sqcup V_r$ is a disjoint union.
Hence the above sum reduces to
\begin{equation}\label{eq:Ir_to_generalnode_reduced}
  \hat\pi_G^{(r)}(s,t) - \hat\pi_G^{(r-1)}(s,t)
  \;=\;
  \sum_{u \in V_r}
    \frac{\pi_{G_X}(u,t)\,d^{G_X}_{\mathrm{out}}(u)}
         {\alpha\,d^{G}_{\mathrm{out}}(u)}
    \bigl(\hat\pi_G^{(r)}(s,u) - \hat\pi_G^{(r-1)}(s,u)\bigr),
  \qquad t \in X_{r-1}.
\end{equation}
Using the definition of $D_r(u)$ from~\eqref{eq:def_diff_free}, we can rewrite~\eqref{eq:Ir_to_generalnode_reduced} as
\begin{equation}\label{eq:Ir_to_generalnode}
\begin{aligned}
  D_r(t)
  &= \frac{1}{\pi_G(s,t)}
     \sum_{u \in V_r}
       \frac{\pi_{G_X}(u,t)\,d^{G_X}_{\mathrm{out}}(u)}
            {\alpha\,d^{G}_{\mathrm{out}}(u)}
       \bigl(\hat\pi_G^{(r)}(s,u) - \hat\pi_G^{(r-1)}(s,u)\bigr)\\
  &= \sum_{u \in V_r}
       w_{r,t}(u)\,D_r(u),
\end{aligned}
\end{equation}
where
\[
  w_{r,t}(u)
  \;:=\;
  \frac{\pi_{G_X}(u,t)\,d^{G_X}_{\mathrm{out}}(u)\,\pi_G(s,u)}
       {\alpha\,d^{G}_{\mathrm{out}}(u)\,\pi_G(s,t)}.
\]
By the corresponding identity for the true PPR values $\pi_G(s,\cdot)$ in Lemma \ref{lemma:newdecom}, we have $w_{r,t}(u) \ge 0$ and
\[
  \sum_{u \in V_r} w_{r,t}(u) \;=\; 1,
\]
with the remaining mass accounted for by nodes in $X \setminus X_{r-1}$, whose increments vanish.
In particular, the coefficients $w_{r,t}(u)$ are deterministic functions of $\mathrm{All}_0,\dots,\mathrm{All}_{r-1}$. Taking conditional expectations in~\eqref{eq:Ir_to_generalnode} and using Case~1, we get
\begin{equation}\label{eq:centered_noise}
\begin{aligned}
  \E\bigl[D_r(t)\,\big|\,\mathrm{All}_0,\dots,\mathrm{All}_{r-1}\bigr]
  &= \sum_{u \in V_r}
       w_{r,t}(u)\,
       \E\bigl[D_r(u)\,\big|\,\mathrm{All}_0,\dots,\mathrm{All}_{r-1}\bigr] = 0.
\end{aligned}
\end{equation}
Combining the two cases, we conclude that for every fixed node $t \in V$, the sequence
$\bigl(D_r(t)\bigr)_{r=1}^R$ is a martingale-difference sequence adapted to
$\bigl(\mathrm{All}_r\bigr)_{r=0}^R$.
In the next subsection we will establish sub-Gaussian tail bounds for these increments and apply Lemma~\ref{lemma:azuma-subgaussian} to control the global error.
\subsection{Establish the Main Theorem}
We are now ready to prove Theorem~\ref{thm:ssppr-r-upper-refined}.
Recall Lemma~\ref{lemma:algorrr_raw}, which shows that the running time
of \newalgor{} is
\[
  \bigo\!\left(m + \sum_{r=2}^{R}\frac{\log(n/p_r)}{c_r^2\delta_r}\right)
  = \bigo\!\left(m + R\log n\cdot \sum_{r=2}^{R}\frac{1}{\delta_r}\right).
\]
In Claim~\ref{claim:thre_sum} we proved that, as long as the PPR
estimates for the already-estimated nodes remain relatively accurate,
the sum of inverse thresholds satisfies
$\sum_{r=2}^{R} 1/\delta_r \in \bigo(n)$ with high probability.
Consequently, once we ensure that this accuracy condition holds in every
round, Lemma~\ref{lemma:algorrr_raw} and Claim~\ref{claim:thre_sum}
imply the overall complexity bound of 
\[
  \bigo\!\left(m + n\log n\,R\right)
\]
for any choice of per-round parameters $c_r \in \Omega(1/\sqrt{R})$.

Within a single round $r$, our adaptive procedure may try several
candidate thresholds $\delta_r$ along a geometric sequence.
However, only the \emph{final} accepted threshold in round $r$
affects the subsequent rounds.
Thus it suffices to guarantee the desired accuracy for the estimator
associated with the last (accepted) threshold in each round.

With this in mind, we fix absolute constants $K_{\mathrm{RW}},
K_{\mathrm{err}}, K_{\mathrm{prob}}$ and choose parameters
$c_r \in \Omega(1/\sqrt{R})$ and $p_r = 1/\mathrm{poly}(n)$ so that,
with high probability, for every round $r \le R$ we have
\begin{equation}\label{eq:prob_target}
\begin{aligned}
  &\forall t \in V,\quad
   \left|\frac{\hat\pi_G^{(r)}(s,t)}{\pi_G(s,t)} - 1 \right|
   \leq \frac{c}{2},\\
  &\text{equivalently, } \forall t \in V,\quad
   \left|\sum_{k=1}^{r} D_k(t)\right| \leq \frac{c}{2},
\end{aligned}
\end{equation}
where $D_k(t)$ denotes the normalized error increment in round $k$
(see~\eqref{eq:def_diff_free}).

\mypara{Good events}
Note that the guarantee in~\eqref{eq:prob_target} is stated in terms of
the number of rounds $R$, rather than the final threshold $\delta$ of
the SSPPR-R query.
In the proof, we refer to the event that~\eqref{eq:prob_target} holds
(together with several additional technical conditions) as the
\emph{good event} up to round~$r$.
On this good event we will show that each per-round increment $D_r(t)$
is sub-Gaussian, which allows us to apply the Azuma-type inequality for
sub-Gaussian martingale differences (Lemma~\ref{lemma:azuma-subgaussian}).

We now formalize the good event via two requirements.

\begin{lemma}\label{lemma:good-events}
There exist absolute constants
$K_{\mathrm{RW}}, K_{\mathrm{err}}, K_{\mathrm{prob}}, \lambda > 0$
such that the following holds.
Assume that $n$ is sufficiently large, and in round $r$ we choose parameters (as in Line $10$ of \newalgor{}): 
\[
  N_r \;=\; K_{\mathrm{RW}}\,
    \frac{\log(1/p_r)}{c_r^2\delta_r},\qquad
  c_r \;=\; \frac{1}{K_{\mathrm{err}}\sqrt{R}},\qquad
  p_r \;=\; \frac{1}{n^{K_{\mathrm{prob}}}}.
\]
Then, for any graph $G$ with $n$ nodes and $m$ edges, there exists a sequence of events $A_0 \subset A_1 \subset \cdots \subset A_R$. 
Each event $A_r$ is defined on $\mathrm{All}_0, \dots, \mathrm{All}_r$ for an execution of \newalgor. $A$ implies that for each $r_1 \leq r$:
\begin{enumerate}[label=(\roman*).]
  \item For any $ t \in V$, it holds that
  \[\left|\sum_{k=1}^{r_1} D_{k}(t)\right|\;\leq\; \frac{c}{2}.\]
  \item For any $ t \in I_{r_1}$,
    \[
      \frac{\hat\pi_G^{(r_1-1)}(s,t)}{\texttt{SUM}(\mathbf{\hat{S}_{r_1}})}
      \;\geq \; \frac{\delta_{r_1}}{\lambda}.
    \]
    \item $\pr[A_{r}] \geq 1-\frac{r_1}{n^2}$.
\end{enumerate}
\end{lemma}
We briefly explain these three conditions.
Condition~(i) is precisely the global accuracy guarantee
in~\eqref{eq:prob_target} for all estimates up to round~$r$, and condition (iii) indicates that this happens with high probability.
\mypara{Condition (ii) for Good Event.}
Condition~(ii) encodes the success of the subroutine $\algoa$ used in
the adaptive threshold search.
In our algorithm, we invoke $\algoa$ at each attempt within round~$r$
for every candidate threshold $\delta_r$.
In a single call, $\algoa$ repeatedly samples nodes independently, where
node $t$ is selected with probability
\[
  \frac{\hat\pi_G^{(r-1)}(s,t)}{\texttt{SUM}(\mathbf{\hat{S}_r})},
\]
and checks whether $t$ belongs to $X_{\mathrm{out}}$.
By Lemma~\ref{lemma:discover}, if we set $p_r = 1/\mathrm{poly}(n)$,
then with probability at least $1-p_r$ the returned set
$X_{\mathrm{out}}$ contains no nodes with too small PPR scores. In particular, $\forall t\in I_r$
\begin{equation}\label{eq:good_algoa}
  \frac{\hat\pi_G^{(r-1)}(s,t)}{\texttt{SUM}(\mathbf{\hat{S}_r})}
    \;\ge\; \frac{\delta_r }{\lambda_0}
\end{equation}
for some absolute constant $\lambda_0 > 1$. Indeed, we may take $\lambda_0 = 5e/(1-c)$. Note that the factor $(1-c)$ arises from the threshold $(1-c)\delta_r$ used in \newalgorr's call to \algoa; this choice ensures that $\pi_G(s,t) \ge \delta_r$ for all $t \in I_r$ in Lemma~\ref{lemma:algorrr_raw}.

Note that $\algoa$ depends only on the intermediate scores
$\hat\pi_G^{(r)}(s,t)$ and therefore operates correctly even when
$\hat\pi_G^{(r)}(s,t)$ is still far from the true value $\pi_G(s,t)$. Each invocation of $\algoa$ fails with probability at most $p_r$, and
there are at most $\bigo(R\log n)$ such invocations over the entire run of
the algorithm (one for each candidate threshold in each round). Thus the total failure probability is at most $\bigo(p_r R\log n)$.
By choosing $K_{\mathrm{prob}}$ sufficiently large and applying a union
bound, we may condition on the event that \emph{all} calls to $\algoa$
succeed, and include this as part of our good event $A_r$ for all
$r \le R$. Under this conditioning, condition (ii) holds automatically throughout the analysis.

\mypara{Truncating the tail of $D_r$.}
We now construct, for each round $r$, a high-probability event on which
$D_r(t)$ is sub-Gaussian, and later combine these events inductively to
obtain Lemma~\ref{lemma:good-events}(i).
Throughout this subsection we fix a round $r$ and condition on
$\mathrm{All}_0,\dots,\mathrm{All}_{r-1}$.

Recall that for each $t \in V_r$ the estimator produced in round $r$ is
\[
  \hat\pi_G^{(r)}(s,t)
  \;=\;
  \frac{H[t]}{N_r}\cdot \texttt{SUM}(\mathbf{\hat S}),
  \quad \text{where }
  H[t] \sim \mathrm{Binomial}\Bigl(
    N_r,\,
    \texttt{SUM}(\mathbf{\hat S})^{-1}\hat\pi_G^{(r-1)}(s,t)
  \Bigr).
\]
By the standard Chernoff bound, for every $a \in (0,1]$ we have
\begin{equation}\label{eq:chernoff-one-round}
  \pr\Bigl[
    \Bigl|
      \frac{H[t]\;\texttt{SUM}(\mathbf{\hat S})}
           {N_r\hat\pi_G^{(r-1)}(s,t)}
      - 1
    \Bigr|
    \;\geq\; a
    \,\Bigm|\,\mathrm{All}_0,\dots,\mathrm{All}_{r-1}
  \Bigr]
  ~\le~
  2\exp\Bigl(
    -\frac{a^2 N_r \hat\pi_G^{(r-1)}(s,t)}
           {3\,\texttt{SUM}(\mathbf{\hat S})}
  \Bigr).
\end{equation}

We will repeatedly restrict random variables to a given event using its
indicator function.
For an event $A$, its indicator $\mathbf{1}_A$ is the random variable
\begin{equation}\label{eq:def_char_func}
  \mathbf{1}_A \;=\;
  \begin{cases}
    1, & \text{if } A \text{ occurs},\\
    0, & \text{otherwise}.
  \end{cases}
\end{equation}

Fix the good event $A_{r-1}$ from Lemma~\ref{lemma:good-events}, and
consider the truncated increment $D_r(t)\mathbf{1}_{A_{r-1}}$.
For $t \in V_r$ and any $a>0$ we have
\begin{align*}
   &\pr\Bigl[
      |D_r(t)\,\mathbf{1}_{A_{r-1}}| \,\geq\, a
      \,\Bigm|\,\mathrm{All}_0,\dots,\mathrm{All}_{r-1}
   \Bigr]\\
   &\qquad=
   \pr\Bigl[
     \Bigl|
       \frac{\hat\pi_G^{(r)}(s,t)}{\pi_G(s,t)}
       -\frac{\hat\pi_G^{(r-1)}(s,t)}{\pi_G(s,t)}
     \Bigr|
     \mathbf{1}_{A_{r-1}}
     \,\geq\, a
     \,\Bigm|\,\mathrm{All}_0,\dots,\mathrm{All}_{r-1}
   \Bigr]\\
   &\qquad\le
   \pr\Bigl[
     \Bigl|
       \frac{\hat\pi_G^{(r)}(s,t)}{\hat\pi_G^{(r-1)}(s,t)} - 1
     \Bigr|
     \,\geq\,
     a\,\frac{\pi_G(s,t)}{\hat\pi_G^{(r-1)}(s,t)}
     \,\Bigm|\,\mathrm{All}_0,\dots,\mathrm{All}_{r-1}
   \Bigr]\\
   &\qquad=
   \pr\Bigl[
     \Bigl|
       \frac{H[t]\;\texttt{SUM}(\mathbf{\hat S})}
            {N_r\hat\pi_G^{(r-1)}(s,t)}
       -1
     \Bigr|
     \,\geq\,
     a\,\frac{\pi_G(s,t)}{\hat\pi_G^{(r-1)}(s,t)}
     \,\Bigm|\,\mathrm{All}_0,\dots,\mathrm{All}_{r-1}
   \Bigr].
\end{align*}
Applying~\eqref{eq:chernoff-one-round} with
$a' := a\,\pi_G(s,t)/\hat\pi_G^{(r-1)}(s,t)$ yields
\begin{align*}
   \pr\Bigl[
     |D_r(t)\,\mathbf{1}_{A_{r-1}}| \,\geq\, a
     \,\Bigm|\,\mathrm{All}_0,\dots,\mathrm{All}_{r-1}
   \Bigr]
   &\le
   2\exp\Bigl(
     -\frac{a^2}{3\,\texttt{SUM}(\mathbf{\hat S})}
      \Bigl(\frac{\pi_G(s,t)}{\hat\pi_G^{(r-1)}(s,t)}\Bigr)^2
      N_r\hat\pi_G^{(r-1)}(s,t)
   \Bigr).
\end{align*}

We now substitute the lower bound on $N_r$ and exploit the properties of
the good event $A_{r-1}$.
Recall that
\[
  N_r \;\ge\;
  K_{RW}\,\frac{\log(1/p_r)}{c_r^2\delta_r},
\]
and that on $A_{r-1}$ we have, for all nodes in the current output set,
\[
  \frac{\pi_G(s,t)}{\texttt{SUM}(\mathbf{\hat S})}
  \;\ge\; \frac{\delta_r}{\lambda}
  \qquad\text{(Lemma~\ref{lemma:good-events}(ii))},
\]
while the relative error bound from Lemma~\ref{lemma:good-events}(i)
implies
\[
  \frac{\hat\pi_G^{(r-1)}(s,t)}{\pi_G(s,t)} \in [1-c,1+c].
\]
Combining these inequalities and using $\hat\pi_G^{(r-1)}(s,t)\le(1+c)\pi_G(s,t)$, we obtain, for all
$a \le 1-c$,
\begin{align*}
   &\quad\qquad\frac{a^2}{3\,\texttt{SUM}(\mathbf{\hat S})}
    \Bigl(\frac{\pi_G(s,t)}{\hat\pi_G^{(r-1)}(s,t)}\Bigr)^2
    N_r\hat\pi_G^{(r-1)}(s,t) \\
   &\qquad\ge
   \frac{a^2\,
   K_{RW}}{3\,\texttt{SUM}(\mathbf{\hat S})(1+c)}
   \frac{\log(1/p_r)}{c_r^2\delta_r}\,
   \pi_G(s,t)\\
   &\qquad\ge
   \frac{a^2\,
   K_{RW}}{3\,\texttt{SUM}(\mathbf{\hat S})(1+c)}
   \frac{\log(1/p_r)}{c_r^2\delta_r}\,
   \frac{\texttt{SUM}(\mathbf{\hat S})\delta_r}{\lambda}\\
   &\qquad=
   \frac{K_{RW}a^2}{3\lambda(1+c)}\cdot
   \frac{\log(1/p_r)}{c_r^2}.
\end{align*}
Hence, for all $a \le 1-c$,
\[
  \pr\Bigl[
    |D_r(t)\,\mathbf{1}_{A_{r-1}}| \,\geq\, a
    \,\Bigm|\,\mathrm{All}_0,\dots,\mathrm{All}_{r-1}
  \Bigr]
  \;\le\;
  2\exp\Bigl(
    -\frac{K_{RW}a^2}{3\lambda(1+c)}\cdot
     \frac{\log(1/p_r)}{c_r^2}
  \Bigr).
\]

To explicitly truncate the very rare large deviations, define
\begin{equation}\label{eq:Brt}
  B_r(t)
  := \Bigl\{
       \mathrm{All}_r :
       |D_r(t)\,\mathbf{1}_{A_{r-1}}| \le 1-c
     \Bigr\}.
\end{equation}
By the above tail bound with $a = 1-c$, we obtain that
\begin{align*}
  \pr\bigl[B_r(t)\cap A_{r-1}\bigr]
    &\ge
    \pr[A_{r-1}]
    - \pr\Bigl[
        |D_r(t)\,\mathbf{1}_{A_{r-1}}|
        \,\ge\,1-c
        \,\Bigm|\,\mathrm{All}_0,\dots,\mathrm{All}_{r-1}
      \Bigr]\\
    &\ge
    \pr[A_{r-1}]
    - 2\exp(-\frac{K_{RW}(1-c)^2}{3\lambda(1+c)}\cdot
      \frac{\log(1/p_r)}{c_r^2}).
\end{align*}
Moreover, for $a>0$ we have
\[
  \pr\Bigl[
    |D_r(t)\,\mathbf{1}_{A_{r-1}\cap B_r(t)}|
      \,\geq\, a
      \,\Bigm|\,\mathrm{All}_0,\dots,\mathrm{All}_{r-1}
  \Bigr]
  \;\le\;
  \begin{cases}
    2\exp\!\Bigl(
      -\dfrac{K_{RW}a^2}{3\lambda(1+c)}\cdot
       \dfrac{\log(1/p_r)}{c_r^2}
    \Bigr),
      & a \le 1-c,\\[4pt]
    0, & a > 1-c.
  \end{cases}
\]
Therefore, in the sense of Definition~\ref{def:subG}, for any event
$B \subseteq A_{r-1}\cap B_r(t)$, the truncated increment
$D_r(t)\mathbf{1}_B$ (conditioned on $\mathrm{All}_0,\dots,\mathrm{All}_{r-1}$)
is sub-Gaussian with parameter
\[
  \sigma_r
  \;:=\;
  c_r \cdot \sqrt{\frac{3\lambda(1+c)}{K_{RW}\log(1/p_r)}}.
\]

Finally, we extend this bound from $t\in V_r$ to all $t \in V$.
Recall from~\eqref{eq:Ir_to_generalnode} that for every
$t \in X_{r-1}$ we can write
\[
  D_r(t) \;=\; \sum_{u\in V_r} w_{r,t}(u)\,D_r(u),
\]
where $w_{r,t}(u)\ge 0$ and $\sum_{u\in V_r} w_{r,t}(u)\le 1$, and that
$D_r(t)=0$ for $t\notin X_{r-1}\cup V_r$.
Thus $D_r(t)$ is a convex combination of the increments $\{D_r(u)\}_{u\in V_r}$.
By Fact~\ref{fact:subgaussian-convex-combination}, for any event
\[
  B \;\subseteq\; A_{r-1} \cap \bigcap_{u\in V_r} B_r(u)
\]
and any node $t\in V$, the truncated increment
$D_r(t)\mathbf{1}_B$ (conditioned on $\mathrm{All}_0,\dots,\mathrm{All}_{r-1}$)
is sub-Gaussian with parameter
\begin{equation}\label{eq:one-round-subG}
  C_1\,\sigma_r
  \;=\;
  c_r \cdot \sqrt{\frac{3C_1^2\lambda}{K_{RW}\log(1/p_r)}}.
\end{equation}
This is the per-round sub-Gaussian bound we will use in the inductive
construction of the good events $A_r$.

\mypara{Construction of $A_r$.}
We now construct the good events $A_r$ inductively and bound their probabilities.

Fix $1 \le r \le R$ and assume that $A_{r_1}, r_1< r$ has already been defined. To prepare for an application of Lemma~\ref{lemma:azuma-subgaussian}, we first intersect all high-probability events obtained up to round~$r$ and set
\[
  A_r^\circ
  \;:=\;
  A_{r-1}
  \;\cap\;
  \bigcap_{r_1 \le r}\;\bigcap_{u \in V_{r_1}} B_{r_1}(u),
\]
where the events $B_{r_1}(u)$ are defined in~\eqref{eq:Brt}. For every $t \in V$ and every $r_1 \le r$, the truncated increment
$D_{r_1}(t)\mathbf{1}_{A_r^\circ}$ is, conditionally on
$\mathrm{All}_0,\dots,\mathrm{All}_{r_1-1}$, sub-Gaussian with parameter given in~\eqref{eq:one-round-subG}. 
We define $A_{r_1}^{\circ} \subset A_{r_1-1}$ for all $r_1< r$ analogously. To construct $A_r$, we need to bound the probability of $\pr[|\sum_{r_1\leq r} D_{r_1}(t)| \geq c/2] $ using the truncated $D_{r_1}(t)$. However, our truncation can introduce bias. For a fixed node $t \in V$ we now center these increments:
\[
  \bar D_{r_1}(t)
  \;:=\;
  D_{r_1}(t)\mathbf{1}_{A_{r_1}^\circ}
  \;-\;
  \mathbb{E}\!\left[
    D_{r_1}(t)\mathbf{1}_{A_{r_1}^\circ}
    \,\middle|\, \mathrm{All}_0,\dots,\mathrm{All}_{r_1-1}
  \right],
  \qquad r_1 \le r.
\]
Then $(\bar D_{r_1}(t))_{r_1 \le r}$ is a martingale-difference sequence adapted to $(\mathrm{All}_{r_1})_{ r_1 \leq r}$. By Lemma~\ref{lemma:centering-subgaussian} and~\eqref{eq:one-round-subG}, each $\bar D_{r_1}(t)$ is sub-Gaussian, conditionally on the past, with parameter
\[
  c_{r_1} \cdot
  \sqrt{\frac{3C_1^2 C_2^2 (1+c)\lambda}{K_{RW}\log(1/p_{r_1})}}.
\]

We now fix the absolute constants
\begin{equation}\label{eq:constant_values_logfac}
  K_{RW}
  \;=\;
  3(1+c)\max(1,C_1)^2\max(1,C_2)^2\max(1,C_3)^2\max(1,\lambda),
  \;
  K_{\mathrm{err}} \;=\; 16/c,
  \;
  K_{\mathrm{prob}} \;=\; 5.
\end{equation}
Using $c_{r_1} =1/(K_{\mathrm{err}}\sqrt{R})$ and $p_{r_1} = 1/n^{K_{\mathrm{prob}}}$, the above sub-Gaussian parameter is at most
$c/(16\sqrt{R})$ for all $r_1 \le r$. Applying Lemma~\ref{lemma:azuma-subgaussian} with $q = n^{-5}$, we have
\begin{equation}\label{eq:Azuma_centered_total}
  \pr\Biggl[ A_{r-1} \cap \{
    \Biggl|
      \sum_{r_1 \le r} \bar D_{r_1}(t)
    \Biggr|
    \geq \frac{c}{4}\}
  \Biggr]
  \;\le\; \frac{2}{n^5}.
\end{equation}
Similarly, since $1-c \geq c/4$, we can bound, for every $r_1 \le r$ and $u \in V_{r_1}$, the probability $\pr[A_{r-1} \cap \neg B_{r_1}(u)]$, by\[
\pr[A_{r-1} \cap \neg B_{r_1}(u)] \leq \pr[A_{r-1} \cap  \{
| \bar D_{r_1}(u)|
    \geq 1-c\}] \leq \frac{2}{n^5}
\]
It remains to bound the bias introduced by truncation. Specifically, we choose a round $r_1 \leq r$. Using~\eqref{eq:centered_noise} and the definition of $B_{r_1}(t)$, we have
\begin{align*}
  \left|
    \mathbb{E}\!\left[
      D_{r_1}(t)\mathbf{1}_{A_{r_1}^\circ}
      \,\middle|\,\mathrm{All}_0,\dots,\mathrm{All}_{r_1-1}
    \right]
  \right|
  &=
  \left|
    \mathbb{E}\!\left[
      D_{r_1}(t)\mathbf{1}_{\neg A_{r_1}^\circ}
      \,\middle|\,\mathrm{All}_0,\dots,\mathrm{All}_{r_1-1}
    \right]
  \right|\\
  &\leq \pr[\neg A_{r_1}^\circ ](1+c) +   \left|
    \mathbb{E}\!\left[
      D_{r_1}(t)\mathbf{1}_{\neg A_{r_1}^\circ} \mathbf{1}_{|D_{r_1}(t)| \geq 1+c}
      \,\middle|\,\mathrm{All}_0,\dots,\mathrm{All}_{r_1-1}
    \right]
  \right|.
\end{align*}
Now we split the bias into two parts by $|D_{r_1}(t)|$. The first part is well-controlled with $ \pr[\neg A_{r_1}^\circ ]$, whereas for the second part, we have to use a Chernoff bound for large deviations (~(3) in Fact \ref{fact:chernoff}). Specifically, conditioned on earlier events in round $r_1$, we have
\[
\pr[H[v] \geq (1+a)\mathbb E[H[v]] ]\leq \exp(-\frac{aE[H[v]] }{3}),\forall a\geq 1.
\]
With an analogous argument in the small-deviation case, we have for all $  a \geq 1+c \geq \frac{\hat\pi_G^{(r)}(s,t)}{\pi_G(s,t)}$,

\begin{align*}
  \pr\Bigl[
     |D_{r_1}(t)\,\mathbf{1}_{\neg A_{r_1}^\circ} \mathbf{1}_{|D_{r_1}(t)|\geq 1+c}
      \,\geq\, a
      \,\Bigm|\,\mathrm{All}_0,\dots,\mathrm{All}_{r_1-1}
  \Bigr]&
  \;\le\; \exp\Bigl(
    -\frac{K_{RW}a}{3\lambda}\cdot
     \frac{\log(1/p_r)}{c_r^2}
  \Bigr)\leq  n^{-5 a}.\\
  \end{align*}
 Note that in this equation, we use $a$ in the exponent instead of $a^2$ in the small-deviation case.
We can ensure that 
  \begin{align*}
  \left|\mathbb{E}\left[
          D_{r_1}(t)\,\mathbf{1}_{\neg A_{r_1}^\circ} \mathbf{1}_{|D_{r_1}(t)|\geq 1+c}  
          \middle|\,\mathrm{All}_0,\dots,\mathrm{All}_{r_1-1}
    \right]\right|&
    \;\le \;\mathbb{E}\left[
          |D_{r_1}(t)|\, \mathbf{1}_{|D_{r_1}(t)|\geq 1+c}  
          \middle|\,\mathrm{All}_0,\dots,\mathrm{All}_{r_1-1}
    \right]\\
    &= \int_{a\geq 1+c}  \pr\left[
          |D_{r_1}(t)|\, \mathbf{1}_{|D_{r_1}(t)|\geq 1+c}  
          \middle|\,\mathrm{All}_0,\dots,\mathrm{All}_{r_1-1} 
    \right]\operatorname{d}a\\
   & \le\int_{a\geq 1+c} n^{-5 a} \operatorname{d}a \;\leq\; \frac{1}{n^5}.
\end{align*}

Now that we have estimated the bias, we are able to bound the probabilities. By the choice of $K_{\mathrm{prob}}$ and the tail bound for $\neg B_{r_1}(u)$, a union bound over all $r_2 \le r_1$ and all $u \in V_{r_1}$ implies that
\[
  \pr[\neg A_{r_1}^\circ] \;\le\; \pr[\neg A_{r_1-1}] + \sum_{r_2 \leq r_1} \sum_{u \in V_{r_2}}   \pr[A_{r_1-1} \cap \neg B_{r_2}(u)] \leq \frac{r_1-1}{n^2} + \frac{2Rn}{n^5} \in \tilde\bigo(\frac{1}{n^2}).
\]
Summing the bias over $r_1\le r$  using the triangle inequality yields for sufficiently large $n$,
\[
  \left|
    \mathbb{E}\!\left[
      \sum_{r_1 \le r} D_{r_1}(t)\mathbf{1}_{A_{r_1}^\circ}
    \right]
  \right|
  \le
  r(1+c)\,\pr[\neg A_r^\circ] + \frac{r}{n^5}
  \le \frac{c}{4}.
\]

Combining this bias bound with~\eqref{eq:Azuma_centered_total}, we obtain
\[
  \pr\Biggl[
    \Biggl|
      \sum_{r_1 \le r} D_{r_1}(t)\mathbf{1}_{A_{r_1}^\circ}
    \Biggr|
    > \frac{c}{2}
  \Biggr]
  \;\le\; \frac{2}{n^5},
\]
for any $t \in V$ . We now set
\[
  A_r(t)
  \;:=\;
  A_r^\circ \,\cap\,
  \Biggl\{
    \Biggl|
      \sum_{r_1 \le r} D_{r_1}(t)
    \Biggr|
    \le \frac{c}{2}
  \Biggr\},
  \qquad
  A_r
  \;:=\;
  \bigcap_{t\in V} A_r(t).
\]
A union bound over all $t \in V$ gives
\[
  \pr[A_r]
  \;\ge\;
  \pr[A_{r-1}] -  \sum_{r_1 \leq r} \sum_{u \in V_r}   \pr[A_{r-1} \cap \neg B_{r_1}(u)] - \sum_{t\in V_{r_1}} \pr[\neg A_r(t) \cap A_{r}^{\circ}]
  \;\ge\;
  1 - \frac{r}{n^2},
\]
where we also used the inductive hypothesis
$\pr[A_{r-1}] \ge 1 - (r-1)/n^2$ and the bound on the failure probability of the events $ A_r(t), B_{r_1}(u)$. By construction, on $A_r$ we have for all $t \in V$, $|\sum_{k=1}^{r} D_k(t)|
  \le \frac{c}{2}.$
In addition, since in our construction $A_{r} \subset A_{r-1}$, we can now claim that:
\[
\forall r_1\leq r, t\in V, \qquad \Biggl|\sum_{r_2=1}^{r_1} D_{r_2}(t)\Biggr|
  \le \frac{c}{2},
\]
which is precisely Lemma~\ref{lemma:good-events}(i) for round~$r$. This completes the inductive construction of $A_r$.\qed

%% file: floats/algo_round_est_simple.tex
    \begin{algorithm}
    \caption{\(\mathsf{\algorr{}(Graph-Prebuilt)}\)}\label{algo:round_est_simple}
    \DontPrintSemicolon
    \KwInput{Graph \(G_X\); target set \(X\); initial estimation $\hat{\pi}_G(s,\cdot)$; decay factor \(\alpha\); error parameter \(c_r\); threshold \(\delta_r\); failure probability \(p_r\), source distribution $\mathbf{\hat{\bar{S}}_r}(u)$, sum of weights $\texttt{SUM}(\mathbf{\hat{S}_r})$, constant $K_{Round}$}
    \KwOutput{Detected set $X_{out}\in X$ and updated estimation $\hat{\pi}_G(s,\cdot)$, ratio of landing in non-estimated set $X\setminus X_{out}$}
    $X_{out} \gets \algoa{}(G_X, \mathbf{\hat{\bar{S}}}, \alpha, \delta_r(1-c), p_r)$;\\
    $N_r\gets {K_{Round}\log(n/p_r)}/\left({c_r^2\delta_r}\right), H(\cdot)\gets \mathbf{0};$ \\
    \For{\(i=1,2,\dots,N_r\)}{
      Draw a node \(u\sim \mathbf{\hat{\bar{S}}}\);\\
      $v \gets$ end of i-th $\alpha$-RW from \(u\) on $G_X$;\\
      \If{\(v\in X_{out}\)}{
        \(H[v]\gets H[v]+\frac{1}{N_r}\);\\
      }
    }
    \(\hat\pi_G(s,t) \gets H(t)\cdot \texttt{SUM}(\mathbf{\hat{S}_r}) \text{ for } \forall t \in X_{out}; \)\\
    \(\mathrm{Ratio}\gets \dfrac{\sum_{x\in X\setminus X_{out}}H[x]}{\sum_{x\in X}H[x]}\);\\
    \Return{\(X_{\text{out}}\,,\,\hat{\pi}_G(s,\cdot), \mathrm{Ratio}\)}
    \end{algorithm}

%% file: floats/algo_dist_walk_adjusted.tex
\begin{algorithm}
 
\caption{$\mathsf{\algor{}(Adaptive-Threshold)}$}\label{algo:dist_walk_adjusted}  
\DontPrintSemicolon
\KwInput{Graph $G$; source node $s$; decay factor $\alpha$; error $c$; threshold $\delta$; probability $p_f$; shrinking factor $\Delta$, rescale factor $\tau$, constants $K_{RW}, K_{err}, K_{prob}$}
\KwOutput{Estimation $\hat{\pi}_G(s,\cdot).$}
$\hat{\pi}_G(s,\cdot) \gets \mathbf{0}, \delta_0\gets\frac{\alpha}{(1+c)n\Delta}, R\gets\log_\Delta(\frac{\delta_0}{\delta});$\\
% $\bar{c} \gets \frac{c}{4R},\bar\delta = \frac{1}{\alpha(1-\alpha)\Delta n}$ ;\\
$V_{1} \gets \algoa{}(G, s, \delta_0, p_f)$;\\
Cast \(\frac{48\log(8n/p_f)}{c^2\delta_0}\) $\alpha-$RWs from $s$ to estimate $\hat{\pi}_G(s,u) , \forall u\in V_{0}$;\\
\For{$r=2,\dots,R+1$}{
    $X\gets X\setminus I_{r-1}$, and construct $G_X$;\\
      $ \forall u\in \boundary{X}, \mathbf{\hat{S}_r}(u) \gets {d_{\mathrm{out}}^{G_X}(u)\hat{\pi}_G(s,u)}/{(\alpha\,d_{\mathrm{out}}^{G}(u))};$
      $ \texttt{SUM}(\mathbf{\hat{S}_r}) \gets \sum_{u \in \boundary{X}}\mathbf{\hat{S}_r}(u), \mathbf{\hat{\bar{S}}_r} \gets \mathbf{\hat{S}_r}/\texttt{SUM}(\mathbf{S_r})$;\\
    build alias table for $\hat{\bar S}[u] $;\\
    \If{$\texttt{SUM}(\mathbf{\hat{S}_r}) \frac{\log(n/p_f)}{m} \leq \delta/2$}{\textbf{break};}
    
    $c_r=\frac{c}{\sqrt{R}K_{err}},p_r=\frac{1}{n^{K_{prob}}}$;\\
    \For{$\delta_r = \frac{R}{n}, \frac{R\tau}{n}, \frac{R\tau^2}{n}, \dots$}{
        $I_r, \hat\pi_G(s, \cdot), \mathrm{Ratio} \gets $\\
        $\algorr\textsf{(Graph-Prebuilt)}(
            G_X, X, \hat \pi_G(s, \cdot), \alpha, c_r, \delta_r, p_r, \mathbf{\hat{S}_r},\texttt{SUM}(\mathbf{\hat{S}_r}),  K_{RW}.
        )$;\\
        \If{$\mathrm{Ratio} \leq \frac{1}{2\Delta}$}{\textbf{break};}
    } 
    
}

$\delta_{r}\gets \frac{\log(n/p_f)}{m}, c_r=c/4, p_r=p_f/(8n);$\\
$I_r, \hat\pi_G(s, \cdot), \mathrm{Ratio} \gets \algorr\textsf{(Graph-Prebuilt)}(
            G_X, X, \hat \pi_G(s, \cdot), \alpha, c_r, \delta_r, p_r, \mathbf{\hat{S}_r}, \texttt{SUM}(\mathbf{\hat{S}_r}), K_{RW}.
        )$;\\
\Return{$\hat{\pi}_G(s,\cdot)$}
\end{algorithm}